\documentclass{amsart}

\usepackage{amsmath}
\usepackage{amsthm}
\usepackage{amssymb}
\usepackage{esint}

\usepackage{dsfont}
\usepackage{graphicx}
\usepackage{caption}
\usepackage{subcaption}
\usepackage{color}
\usepackage[english]{babel}
\usepackage{layout}
\usepackage{hyperref}

\newtheorem{theorem}{Theorem}[section]

\newtheorem{corollary}{Corollary}
\newtheorem{lemma}[theorem]{Lemma}
\newtheorem{proposition}{Proposition}

\theoremstyle{definition}
\newtheorem{definition}[theorem]{Definition}
\newtheorem{remark}{Remark}
\newtheorem{hypo}{Hypothesis}

\numberwithin{equation}{section}
\newcommand{\dx}{\,\mathrm{d}x}

\newcommand{\e}{\varepsilon}
\newcommand{\dist}{{\rm{dist}}}
\newcommand{\Lw}{\mathcal{L}}
\newcommand{\calL}{\mathcal{L}}
\newcommand{\vol}{\mathrm{vol}}

\newcommand{\expec}[1]{\mathbb{E}[#1]}

\font\tenmsb=msbm10
\font\sevenmsb=msbm7
\font\fivemsb=msbm5
\newfam\msbfam
\textfont\msbfam=\tenmsb
\scriptfont\msbfam=\sevenmsb
\scriptscriptfont\msbfam=\fivemsb

\def\R{\mathbb{R}}
\def\N{\mathbb{N}}
\def\B{\mathbb{B}}
\def\Z{\mathbb{Z}}
\def\T{\mathbb{T}}
\def\Lip{\mathrm{Lip}}

\def\Ard{\mathcal{A}^R(D)}
\def\a{\alpha}

\newcommand{\step}[2]{\medskip  \noindent \textit{Step~}#1.~#2. \\}

\begin{document}

\author[M. Cicalese]{Marco Cicalese}
\address[Marco Cicalese]{Zentrum Mathematik - M7, Technische Universit\"at M\"unchen, Boltzmannstrasse 3, 85747 Garching, Germany}
\email{cicalese@ma.tum.de}

\author[A. Gloria]{Antoine Gloria}
\address[Antoine Gloria]{Sorbonne Universit\'e, UMR 7598, Laboratoire Jacques-Louis Lions, Paris, France \& Universit\'e Libre de Bruxelles, Brussels, Belgium}\email{gloria@ljll.math.upmc.fr}

\author[M. Ruf]{Matthias Ruf}
\address[Matthias Ruf]{Universit\'e Libre de Bruxelles (ULB), Brussels, Belgium}
\email{matthias.ruf@ulb.ac.be}

\title[From statistical polymer physics to nonlinear elasticity]{From statistical polymer physics to nonlinear elasticity}

\begin{abstract}
A polymer-chain network is a collection of interconnected polymer-chains, made themselves of the repetition of a single pattern called a monomer.
Our first main result establishes that, for a class of models for polymer-chain networks, the thermodynamic limit in the canonical ensemble
yields a hyperelastic model in continuum mechanics. In particular, the discrete Helmholtz free energy of the network converges to the infimum of a continuum integral functional (of an energy density depending only on the local deformation gradient) and the discrete Gibbs measure converges (in the sense of a large deviation principle) to a measure supported on minimizers of the integral functional.
Our second  main result establishes the small temperature limit of the obtained continuum model (provided the discrete Hamiltonian is itself independent of the temperature), and shows that it coincides with the $\Gamma$-limit of the discrete Hamiltonian, thus showing that thermodynamic and small temperature limits commute.
We eventually apply these general results to a standard model of polymer physics from which we derive nonlinear elasticity.
We moreover show that taking the $\Gamma$-limit of the Hamiltonian is a good approximation of the thermodynamic limit at finite temperature in the regime of large number of monomers per polymer-chain (which turns out to play the role of an effective inverse temperature in the analysis).
\end{abstract}

\maketitle

\section{Introduction and statement of the main results}

\subsection{Polymer physics and nonlinear elasticity}

On the one hand, rubber-like materials are the realm of continuum mechanics and constitute the paradigmatic example of hyperelastic materials at large deformations --- that is, 
their energy density and stress tensor only depend locally on the gradient of deformation. 
Consider a piece of material that occupies a domain $D$ at rest, and which is deformed according to some map $u:D\to \R^3$ (in Lagrangian coordinates).
The energy of the deformed configuration then takes the form
\begin{equation*}
{\mathcal I}(D,u)\,:=\,\int_D \overline W(\nabla u(x))dx,
\end{equation*}
where $\overline W:\R^{d\times d}\to [0,+\infty],\Lambda \mapsto \overline W(\Lambda)$ is the energy density of the material (minimal at $\Lambda=\mathrm{Id}$), and is referred to as its constitutive law.
The associated Piola stress tensor is given by $D_\Lambda \overline W(\Lambda)$.
A crucial requirement on the map $\overline W$ is frame-indifference, that is, for all rotations $R\in SO_3(\R)$
and deformation gradients $\Lambda \in R^{3\times 3}$, $\overline W(R\Lambda)=\overline W(\Lambda)$. 
Rubber materials are also usually isotropic, which reads as follows on $\overline W$: For all rotations $R\in SO_3(\R)$
and deformation gradients $\Lambda \in R^{3\times 3}$, $\overline W(\Lambda R)=\overline W(\Lambda)$. 
Finally, rubber materials are nearly-incompressible, which typically requires that $\overline W(\Lambda)$ gets large
when $|\det \Lambda-1|\gg 1$, and should not allow interpenetration of matter, which at least imposes that $\overline W(\Lambda) = +\infty$ if $\det \Lambda \le 0$.
For a given deformation $\varphi: \partial D\to \R^3$ of the boundary, the piece of deformed material (that occupied $D$ in the reference configuration) 
has now energy
\begin{equation}\label{nonlinear-elasticty}
\mathcal E(\varphi):=\inf \Big\{ \int_D \overline W(\nabla u)\,|\,u:D\to \R^3, u|_{\partial D}\equiv \varphi\Big\},
\end{equation}
and its deformation is given by the minimizer of this functional (if attained and unique).
We refer to \cite{Ciarlet} for classical mathematical aspects of nonlinear elasticity.
Standard mechanical experiments illustrate the complexity of the nonlinear response of these materials at large deformations --- see Figure~\ref{fig:Treloar}
for the Treloar experiments in uniaxial traction \cite{Treloar}.
\begin{figure}
\begin{minipage}{0.3\textwidth}
\includegraphics[width=.8\textwidth]{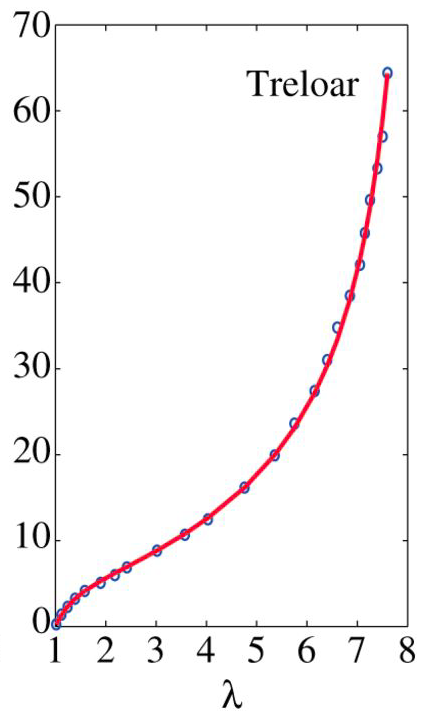} 
\end{minipage}
\begin{minipage}{0.6\textwidth}
A sample of rubber is submitted to prescribed linear deformation at its boundary: $\varphi_\Lambda:x\mapsto \Lambda \cdot x$, for $\Lambda=\mathrm{diag}(\lambda, \lambda^{-\frac12}, \lambda^{-\frac12})$ and $\lambda \in [1,7.5]$.

\medskip

Three regimes:
\begin{itemize}
\item linear response
($|\lambda -1|\ll 1$),
\item  strain softening ($2\le \lambda \le 4$), 
\item strain hardening ($\lambda \ge 6$).
\end{itemize}
\end{minipage}
\caption{Treloar experiments in unaxial traction (engineering stress versus $\lambda$)}\label{fig:Treloar}
\end{figure}

\medskip

On the other hand, rubber-like materials are also the realm of the statistical physics of polymer-chain networks and constitute the paradigmatic example
of materials for which elastic properties are purely entropic --- that is, they are only due to thermal fluctuations.
Consider a network $\mathfrak N$ of cross-linked polymer-chains, described by a (finite) set $\mathfrak L \subset \N$ of cross-links $i$,
a subset $\mathfrak L_{b} \subset \mathfrak L$ of boundary cross-links,
a (finite) set $\mathfrak B \subset \mathfrak L\times \mathfrak L$ of chains $c_{ij}=(i,j)$ (such that $\mathfrak L$ has one single
connected component via $\mathfrak B$, and such that if $(i,j)\in\mathfrak B$ then $(j,i)\notin\mathfrak B$ ). Each chain $c_{ij}$ is itself made of a sequence of $N_{ij}\in \N$ monomers, characterized by $i$, $j$ and $\mathfrak M_{ij}=\{(i,j,1),\dots,(i,j,N_{ij}-1)\}$.
For the simplicity of the presentation in this part of the introduction, we assume that the state-space is discrete, that is, that monomers are placed at edges of the lattice $\Z^3$.
A deformation $u$ of the polymer-chain network $\mathfrak N$ is then a map $\mathfrak L \cup \bigcup_{(i,j)\in \mathfrak B} \mathfrak M_{ij} \to \Z^3$
with the following properties. The map $u$ is edge-injective and has increments unity, that is, for all $(i,j)\ne (i',j')\in \mathfrak B$ and all $0\le k\le N_{ij}-1$, $0\le k'\le N_{i'j'}-1$,
we have with the notation $u_{(i,j,0)}=u_i$ and $u_{(i,j,N_{ij})}=u_j$: 
\begin{eqnarray}
\label{e.injectivity} (u_{(i,j,k)},u_{(i,j,k+1)})&\ne& (u_{(i',j',k')},u_{(i',j',k'+1)}),\\
\label{e.unity} |u_{(i,j,k)}-u_{(i,j,k+1)}|&=&1.
\end{eqnarray}
%
Given a boundary  map $\varphi_b:\mathfrak L_{b} \to \Z^3$, we denote by $\Omega(\mathfrak N,\varphi_b)$ the cardinality of
the set $\{u \text{ deformation }:\mathfrak L \cup \bigcup_{(i,j)\in \mathfrak B} \mathfrak M_{ij} \to \Z^3 \quad |\quad  u|_{\mathfrak L_b} \equiv \varphi_b\}$.
The Helmholtz free energy $\mathcal E^\beta(\mathfrak N,\varphi_b)$ of the network $\mathfrak N$ with boundary deformation $\varphi_b$ at inverse temperature $\beta$  is then given by
\begin{equation}\label{e.free-network-full}
\mathcal E^\beta(\mathfrak N,\varphi_b)\,:=\,-\frac1\beta \log \Omega(\mathfrak N,\varphi_b),
\end{equation}
with the understanding that $\mathcal E^\beta(\mathfrak N,\varphi_b)=+\infty$ if $ \Omega(\mathfrak N,\varphi_b)=0$.
Before we make further restrictions on $\mathfrak L$ and $\mathfrak B$, observe that if the network $\mathfrak B$ is made of one single polymer-chain $(1,2)$ with $N$ monomers, then given $\varphi_b:\{1,2\} \mapsto \{u(1),u(2)\} \in (\Z^3)^2$, $\Omega(\mathfrak N,\varphi_b)$ is explicit and obviously only depends on $N$ and the length $|u(1)-u(2)|$. Indeed, for large $N$, \eqref{e.free-network-full} can be explicitly computed (and typically leads to \eqref{e.KuhnGrun}, see below).

Let us make some further assumptions on the network  $\mathfrak N$ that provides a Lagrangian description 
of the polymer-chain network, and enrich the physics. To this aim, we first quickly give an informal description of the physical phenomenon of cross-linking that
leads from a set of independent strings of monomers to a network of cross-linked polymer-chains. Consider a bath of fluctuating independent strings of monomers in some region of space. If two strings are close enough to one another, they get cross-linked and the part of a string between two cross-links 
is called a polymer-chain. This yields a \emph{reference} polymer-chain network,
and $\mathfrak L$, $\mathfrak L_b$,  $\mathfrak B$ now denote the sets of reference positions of cross-links, boundary cross-links, and end-to-end points of polymer-chains,
whereas given $(x_i,x_j) \in \mathfrak B$, $\mathfrak M_{x_i,x_j}$ now denotes the set of positions $\{x_{ij}^k\}_k$ of the monomers  $(x_{ij}^k,x_{ij}^{k+1})$ of this chain (for $0\le k\le N_{ij}-1$ and with $x_{ij}^0=x_i$, $x_{ij}^{N_{ij}}=x_j$).
We assume that the connectivity of the network is between 3 and 4 (3 or 4 chains meet at every cross-link) and that  the number of monomers of a given polymer-chain $(x_i,x_j)$ between two cross-links at positions $x_i$ and $x_j$ in the 
reference network is of the order of $N_{ij}=|x_i-x_j|^2$
(so that the typical distance of the $N_{ij}$-th monomer from the first one corresponds to the typical distance of a random walker from the origin after $N_{ij}$ steps),
and that the typical distance between two (closest) cross-links is smaller than the typical length of a polymer-chain.
A deformation is now to be understood in the kinematic sense of Lagrangian deformation, say from $\mathfrak L \subset \Z^3\to \Z^3$.
For all $(x_i,x_j)\in \mathfrak B$, we call $\mu_{ij}:\{x_{ij}^0,\dots,x_{ij}^{N_{ij}}\}\to \Z^3$ an admissible deformation of the polymer-chain $(x_i,x_j)$
if it is edge-injective and has increments unity in the sense of \eqref{e.injectivity} \& \eqref{e.unity}. 
Given a deformation $u:\mathfrak L \to \Z^3$ and a polymer-chain $(x_i,x_j)\in \mathfrak B$, we 
denote by $\Omega_{ij}(u)$ the cardinality of the set of admissible deformations $\mu_{ij}$ such that $\mu_{ij}(x_i)=u(x_i)$ and $\mu_{ij}(x_j)=u(x_j)$.
This accounts for local injectivity \emph{within each chain}.
Given a boundary  map $\varphi_b:\mathfrak L_{b} \to \Z^3$, we define $\mathfrak U(\varphi_b):=\{u :\mathfrak L \to \Z^3 \quad |\quad  u|_{\mathfrak L_b} \equiv \varphi_b\}$
the subset of deformations of cross-links that coincide with the boundary map $\varphi_b$ on $\mathfrak L_b$. 
We finally make the assumption that steric effects \emph{between chains} can be accounted for by restricting admissible deformations to a suitable subset 
$\mathfrak V$ of $\{u:\mathfrak L\to \Z^3\}$  (which does not describe the positions of monomers), and define $\widehat{\mathfrak U}(\varphi_b):=\mathfrak U(\varphi_b)
\cap \mathfrak V$.
This assumption allows us to coarsen the model by factorizing the number of  admissible deformations in the form  
$$
\Omega(\mathfrak N,\varphi_b) \sim \sum_{u \in \widehat{\mathfrak U}(\varphi_b)} \prod_{(x_i,x_j)\in \mathfrak B} \Omega_{ij}(u).
$$
This has the effect to integrate out the positions of the monomers and to reduce the characterization of the model to $(\mathfrak L,\mathfrak B)$
and the definition of the state space $\widehat{\mathfrak U}(\varphi_b)$ for any boundary deformation $\varphi_b:\mathfrak L_b \to \Z^3$. Since
these quantities only depend on distances (and/or angles) between cross-links, $\Omega(\mathfrak N,\varphi_b)$ does not depend on the frame
of the Lagrangian description.

We now enrich the physics: On top of the non-interpenetrability of matter, polymer-chains feel the effect of a solvant which yields an internal energy that penalizes changes of volume with respect to the reference network, which we model in the form of $\hat H(u)$, an internal energy that only depends locally on $u$ 
at a scale larger than that of a polymer-chain (in a frame-indifferent way). 
The Helmholtz free energy of the network $\mathfrak N$ with boundary deformation $\varphi_b$ at inverse temperature $\beta$  is then given by
the following modified version of \eqref{e.free-network-full}
\begin{eqnarray*}
\mathcal E^\beta(\mathfrak N,\varphi_b)&=&-\frac1\beta \log \bigg(\sum_{u \in \widehat{\mathfrak U}(\varphi_b)} \Big(\prod_{(x_i,x_j)\in \mathfrak B} \Omega_{ij}(u)\Big)
\exp(-\beta \hat H(u))\bigg)
\\
&=& -\frac1\beta \log \bigg( \sum_{u \in \widehat{\mathfrak U}(\varphi_b)} \exp\Big(  - \beta \big(\hat H(u) + \sum_{(x_i,x_j)\in \mathfrak B}\frac{-1}\beta\log( \Omega_{ij}(u) \big)\Big) \bigg),
\end{eqnarray*}
which we rewrite as
\begin{equation}
\mathcal E^\beta(\mathfrak N,\varphi_b)\,:=\, -\frac1\beta \log \bigg( \sum_{u \in {\mathfrak U}(\varphi_b)} \exp\Big(  - \beta \big( H(u) + \sum_{(x_i,x_j)\in \mathfrak B}\frac{-1}\beta\log( \Omega_{ij}(u) \big)\Big) \bigg),
\label{polymer-physics} 
\end{equation}
by setting $H(u)=\hat H(u)+\tilde H(u)$ where $\tilde H(u)=+\infty$ if $u \notin\mathfrak V$ and $\tilde H(u)=0$ if $u \in \mathfrak V$.
The latter rewriting amounts to penalizing that a deformation $u:\mathfrak L\to \Z^3$ be admissible rather than restricting the 
set of states. The thermally fluctuating network with imposed boundary deformation $\varphi_b$ has then free Helmholtz energy $\mathcal E^\beta(\mathfrak N,\varphi_b)$,
and its configuration is described by a probability measure $\mu^\beta_{\varphi_b}$ on the set of admissible deformations defined by: For all $V\subset {\mathfrak U}(\varphi_b)$,
$$
\mu^\beta_{\varphi_b} \;:=\, \frac{\sum_{u \in V} \exp\Big(  - \beta \big( H(u) + \sum_{(x_i,x_j)\in \mathfrak B}\frac{-1}\beta\log( \Omega_{ij}(u) \big)\Big)}{\sum_{u \in {\mathfrak U}(\varphi_b)} \exp\Big(  - \beta \big( H(u) + \sum_{(x_i,x_j)\in \mathfrak B}\frac{-1}\beta\log( \Omega_{ij}(u) \big)\Big)} .
$$
We refer the reader to \cite{Rub} for more background on polymer physics.
Let us quickly interprete Treloar's experiments and the three regimes of Figure~\ref{fig:Treloar} in terms of polymer physics.
The linear regime essentially represents the fact that for $\varphi_b$ close to the identity map, \eqref{polymer-physics} is close to quadratic.
The regimes of strain softening and strain hardening are related to the entropic term, the geometry of the network, and $H$.
Let us give some intuition on the entropic term by considering a system of two cross-linked polymer-chains of possibly different length, for which the deformation 
of the boundary of the system (that is, the end points of the two polymer-chains except the cross-link) is fixed. For large boundary deformations, monomers tend to align so that there are less configurations available and the free energy of the system gets large and ultimately blows up. For moderately large boundary deformations, among the possible deformation of the cross-link, the one with the largest number of configurations is the linear interpolation of the deformation of the boundary only if the chains have the same length --- otherwise it is advantageous to deform the longer chain more, which yields redistribution of strain and therefore leads to softening.

\medskip

The aim of the present contribution, which draws its inspiration from the work \cite{KoLu}  by Koteck\'y and Luckhaus, is to rigorously relate \eqref{polymer-physics} (or more precisely the associated version with continuous state space) to \eqref{nonlinear-elasticty} in the thermodynamic limit of large number
of polymer-chains at finite temperature $\beta \ne \infty$, and to rigorously justify the model introduced, studied, and analyzed in \cite{ACG2,GLTV}.
In particular, we aim at proving the closeness of $\mathcal E^\beta(\mathfrak N,\varphi_b)$ to $\mathcal E(\varphi)$,
and the closeness of $\mu^\beta_{\varphi_b}$ to the indicator function supported on the set of minimizers associated with  $\mathcal E(\varphi)$.

\medskip

In the following subsection, we introduce the notation, the mathematical (and statistical) description of the network,
and state the main results of the paper on the thermodynamic limit
of the Gibbs measure and the free energy. We then later on apply these results to a specific model of  polymer-chain networks of 
the type above, from which we rigorously derive nonlinear elasticity.

\subsection{Discrete free energies and Gibbs measures}\label{subsec:definitions}

We start with the definition of admissible Euclidean graphs, and first fix once and for all constants $R,r,C_0>0$, and dimensions $d,n$. For technical reasons we assume that $6R<C_0$.
 For  details on  point processes, we refer to the monograph \cite{Moeller}.
\begin{definition}\label{delaunay}
Let $\mathcal{P}\subset\R^d$ be a countable set.
\begin{itemize}
	\item[(i)] $\mathcal{P}$ is said to be in general position if there are no $k+1$ points contained in a common $k-1$-dimensional affine subspace $(1\leq k\leq d)$ and no $d+2$ points lie on the boundary of the same sphere.
	\item[(ii)] A Delaunay tessellation $\mathbb{T}=\{T_i\}_{i\in\mathbb{N}}$ associated with $\mathcal{P}$ is a partition of $\R^d$ into $d$-simplices $T_i$ whose vertices are in $\mathcal{P}$ and such that no point of $\mathcal{P}$ is contained inside the circumsphere of any simplex in $\mathbb{T}$.
\end{itemize}
\end{definition}
We now  introduce our model for the reference configuration of a polymer network.
\begin{definition}\label{defadmissible}
An extended Euclidean graph $G=(\mathcal{L},E,S) \in(\R^d)^{\mathbb{N}} \times \{0,1\}^{\N\times \N}\times \{0,1\}^\N$ 
is a set of points $\calL=\{x_i\}_{i\in \N}\subset (\R^d)^{\mathbb{N}}$, an associated connectivity graph $E\in  \{0,1\}^{\N\times \N}$,
and a subset of points $\calL_1:=\cup_{i|S_i=1}\{x_i\}$. If $E_{ij}=1$,
we say that $(x_i,x_j)$ is an edge of the graph, whereas if $S_i=1$ we say 
that $x_i$ is a ``volumetric point'' (this wording will be clear later). We call $\mathbb{B}$ the set of edges of $(\mathcal L,E,S)$,
and $\mathbb{T}$ the Delaunay tessellation\footnote{uniqueness fails when points are not in general position, which we rule out by condition (v) of this definition} of $\R^d$ associated with $\cup_{i|S_i=1}\{x_i\}$.
We say that $G=(\mathcal L,E,S)$ is an admissible extended Euclidean graph (graph in short) if it satisfies
\begin{itemize}
	\item[(i)] $\dist(z,{\calL_1})\leq R$ for all $z\in\R^d$;
	\item [(ii)] $\dist(x,\mathcal{L} \setminus\{x\})\geq r$ for all $x\in\mathcal{L}$;
	\item[(iii)] for all $x\in\mathcal{L}:\; \{y\in \mathcal L:\; (x,y)\in \B\}\subset B_{C_0}(x)$;
	\item[(iv)] For all $x,y\in\mathcal L$ there exists a path $P(x,y)$ of edges of $\B$ connecting $x$ to $y$ with
	\begin{equation*}
	P(x,y)\subset [x,y]+B_{C_0}(0),
	\end{equation*}
	where $[x,y]=\{x+t (y-x):\;t\in[0,1]\}$;
	{\item[(v)] $\calL_1$ is in general position.}
\end{itemize}
We denote by $\mathcal G$ the set of graphs for which (i)--(v) hold.
\end{definition}
{Note that the set of vertices $\calL$ also satisfies (i).} Point sets with the properties (i) and (ii) are sometimes called Delone sets. This class of point sets has already been used as a reference configuration for atomistic models in elasticity in~\cite{ACG2,BLBL}. Note further that (ii) and (iii) imply that the degree of each vertex is bounded uniformly (this is one of the physical constants of the model). Assumption~(iv) is technical and ensures a coercivity property (see Lemma~\ref{lb}). Assumption~(v) is to avoid the non-uniqueness of Delaunay tessellations --- this is not essential but convenient to simplify measurability issues.

\medskip

Next we endow $\mathcal G$ with a probabilistic structure, and consider on $\mathcal{G}\subset(\R^d)^{\mathbb{N}} \times \{0,1\}^{\N\times \N}\times \{0,1\}^\N$ the $\sigma$-algebra $\Sigma$ given by the trace $\sigma$-algebra of $\mathcal{B}_{(\R^d)^{\mathbb{N}}}\otimes\mathcal{B}_{\{0,1\}^{\mathbb{N}\times\mathbb{N}}}\otimes\mathcal{B}_{\{0,1\}^{\mathbb{N}}}$, where each factor denotes the Borel $\sigma$-algebra given by the product topology on the factors.
We do not distinguish between $(\mathcal L,E,S)$  and $(\mathcal L,\B,\T)$, which we will both denote by $G$. We then give ourselves a statistics on this set of Euclidean graphs  described by a measure $\mathbb{E}$ on $(\mathcal G,\Sigma)$, and address the minimal assumptions on this distribution $\mathbb{E}$.
They are related to the operation of the shift group $(\mathbb{Z}^d,+)$ on $\mathcal G$, that is,
for any shift vector ${z\in\mathbb{Z}^d}$ and any Euclidean graph $G=(\mathcal L,\B,\T)$, the shifted graph
$G+z=(\mathcal L+z,\B+z,\T+z)$ is again a Euclidean graph.
The first assumption is \emph{stationarity}, which means that for any
shift ${z\in\mathbb{Z}^d}$ the random Euclidean graphs $G$ and $G+z$
have the same (joint) distribution. 
The second assumption is \emph{ergodicity}, which means that any (integrable) random variable
$F(G)$ (i.e.~a measurable map of the random graph) that is shift invariant, in the sense that $F(G+z)=F(G)$ for all shift vectors ${z\in\mathbb{Z}^d}$
and almost-every Euclidean graph  $G$ is actually constant, that is $F=\expec{ F}$
for almost every Euclidean graph $G$. {Throughout this paper we will tacitly assume stationarity and ergodicity. Our results remain valid under the mere assumption of stationarity, but all asymptotic quantities may still be non-deterministic.}

\medskip

We are now in the position to introduce the Hamiltonian and the free energy at the microscopic level. Let $D\subset\R^d$ be an open bounded reference domain with Lipschitz boundary. Given a small parameter $0<\e\ll 1$ and any $U\subseteq \R^d$ we set $U_{\e}=\frac{U}{\e}$, and use the short-hand notation $U_\e^\calL:= \calL\cap U_{\e}$ --- with this choice the microscopic scale is set to 1 and the macroscopic scale to $\frac1\e$. We consider microscopic deformations $u:\calL\cap D_{\e}\to\R^n$, whose internal energy takes the form
\begin{equation}\label{def:hamiltonian}
H_{\e}(D,u)=\sum_{\substack{(x,y)\in \B\\ x,y\in D_{\e}}}f\left(x-y,u(x)-u( y)\right)+H_{{\rm vol},\e}(D,u),
\end{equation}
for some map $f$ and {a volumetric term that penalizes large changes of volume and change of ``orientation'' (if it is not identically zero we always consider the case $n=d$). In order to define such a term we need to introduce some further notation. Denote by $\mathcal{V}_1=\{\mathcal{C}_1(x)\}_{x\in\calL_1}$ the Voronoi tessellation of $\R^d$ with respect to the volumetric points $\calL_1$,
and recall that 
\begin{equation*}
\mathcal{C}_1(x):=\{z\in\R^d:\;|z-x|\leq|z-y|\quad\forall y\in\calL_1\}.
\end{equation*}
We define the interior Voronoi cells by
\begin{equation*}
\mathcal{V}_{1,\e}(D)=\{\mathcal{C}_1(x)\in\mathcal{V}_1:\,T\subset D_{\e} \text{ for all }T\in\mathbb{T}\text{ such that }T\cap\mathcal{C}_1(x)\neq\emptyset\}.
\end{equation*}
The intuition behind these cells is that we want to define $H_{\rm vol,\e}(D,\cdot)$ using only the volumetric points inside the domain $D_{\e}$. Given $u:D^{\calL}_{\e}\to\R^d$ we denote by $u_{\rm aff}:\bigcup_{T\subset D_{\e}}\to \R^d$ the continuous and piecewise affine interpolation with respect to the triangulation $\mathbb{T}$ and the values $\{u(x)\}_{x\in\calL_1}$. With these quantities at hand the volumetric term takes the form
\begin{equation}\label{def:volHam}
H_{\rm vol,\e}(D,u)=\sum_{\mathcal{C}_1(x)\in\mathcal{V}_{1,\e}(D)}|\mathcal{C}_1(x)|W\Big({\det}_{\mathcal C_1(x)}(\nabla u_{\rm aff})\Big)
\end{equation}
for some map $W$
and the short-hand notation ${\det}_{\mathcal C_1(x)}(\nabla u_{\rm aff}):= \fint_{\mathcal{C}_1(x)}\det(\nabla u_{\rm aff})\,\mathrm{d}z$.
By definition of the interior Voronoi cells, this sum is well-defined and, since $u_{\rm aff}$ is piecewise affine, the integrals over the Voronoi cells can be rewritten as finite sums.}
The random character of this Hamiltonian $H_{\e}$ is encoded by $\B$ and $\T$ (which is a more descriptive notation of the actual event -- a random graph -- than the standard ``$\omega$'').
We make three assumptions on the discrete energy densities $f$ and $W$.
The first set of assumptions is used for the general results:
\begin{hypo}\label{Hypo1}
The functions $f:\R^d\times \R^n \to \R_+$ 
and $W:\R\to \R_+$ are (jointly) measurable, nonnegative and there exist a constant $C>0$ and an exponent $p>1$ such that for all $z\in\R^d,\xi,\zeta \in \R^n,\lambda\in\R$ we have the (two-sided) $p$-growth condition
\begin{equation}\label{Hyp-1.1}
\frac{1}{C}|\xi|^p-C\leq  f(z,\xi)\leq C(1+|\xi|^p), \quad 0 \leq   W(\lambda) \leq C(1+|\lambda|^{\frac{p}{d}}).
\end{equation}
\end{hypo}
{Some of our results require a slightly stronger set of assumptions.
\begin{hypo}\label{Hypo2}
The function $f:\R^d\times \R^n \to \R_+$  is jointly  measurable, $W:\R\to \R_+$ is continuous, and there exist a constant $C>0$ and an exponent $p>1$ such that for all $z\in\R^d,\xi,\zeta \in \R^n,\lambda,\lambda'\in \R$ we have the (two-sided) $p$-growth condition
\begin{equation}\label{Hyp-2.1}
\frac{1}{C}|\xi|^p-C\leq  f(z,\xi)\leq C(1+|\xi|^p), \quad 0 \leq   W(\lambda) \leq C(1+|\Lambda|^{\frac{p}{d}}),
\end{equation}
and the local Lipschitz conditions
\begin{equation}\label{Hyp-2.2}
\begin{array}{rcl}
|f(z,\xi)-f(z,\zeta)|&\leq& C|\xi-\zeta|(1+ |\xi|^{p-1}+|\zeta|^{p-1}), 
\\
|W(\lambda)-W(\lambda')|&\leq& C|\lambda-\lambda'|(1+|\lambda|^{\frac{p}{d}-1}+|\lambda'|^{\frac{p}{d}-1}).
\end{array}
\end{equation}
If $W\not\equiv 0$, then we assume in addition that $n=d$ and $p\geq d$.
\end{hypo}}

The third set of assumptions is similar to Hypothesis~\ref{Hypo2}, but is tuned for our applications to polymer physics, and exploits the specific form 
of the model.
\begin{hypo}\label{Hypo3}
The function $f:\R^d\times \R^d \to \R_+$  is jointly  measurable, $W:\R\to \R_+$ is continuous, and there exist an exponent $p\ge d$
and constants $0<C_p,C_2,C_p',C_2',C,C'  \sim 1$  such that for all $z\in\R^d,\xi,\zeta \in \R^d,\lambda,\lambda'\in \R$ we have the (two-sided) $p$-growth condition
\begin{equation}\label{Hyp-3.1}
C_2|\xi|^2+C_p|\xi|^p\leq  f(z,\xi)\leq C_2'|\xi|^2+C_p'|\xi|^p, \quad 0 \leq   W(\lambda) \leq C(1+|\lambda|^{\frac{p}{d}}),
\end{equation}
and the local Lipschitz conditions
\begin{equation}\label{Hyp-3.2}
\begin{array}{rcl}
|f(z,\xi)-f(z,\zeta)|&\leq& |\xi-\zeta|(C_2'(|\xi|+|\zeta|)+ C_p' (|\xi|^{p-1}+|\zeta|^{p-1})), 
\\
 |W(\lambda)-W(\lambda')|&\leq& C'|\lambda-\lambda'|(1+|\lambda|^{\frac{p}{d}-1}+|\lambda'|^{\frac{p}{d}-1}).
\end{array}
\end{equation}
\end{hypo}
Note that the second condition in \eqref{Hyp-3.2} follows automatically from \eqref{Hyp-3.1} {if $W$ is assumed to be convex.}

\medskip
Under the above Hypotheses the passage from discrete Hamiltonians to continuum energies is well-understood (e.g. by $\Gamma$-convergence in \cite{AC,ACG2,BCS,BK,BS2}; see also \cite{BS} for results on local minimizers). In this paper we are interested in the asymptotic behavior of the free energy when we prescribe boundary conditions. To this end, given
$\varphi\in \Lip(D,\R^n)$ we define the class of states associated with $\varphi$ at scale $\frac1\e$ as
\begin{equation*}
\mathcal{B}_{\e}(D,\varphi)=\{u:D_{\e}\cap \calL \to\R^n,\;|u(x)-\tfrac1\e \varphi(\e x)|<1\text{ if }\dist(x,\partial D_{\e})\leq C_0\}.
\end{equation*}
We denote by  $\mathcal{V}:=\{\mathcal{C}(x)\}_{x\in\calL}$ the Voronoi tessellation of $\R^d$ associated with $\calL$ (note that this is not necessarily the dual tessellation of $\T$ --- the latter is given by $\mathcal{V}_1$ which could and \emph{will} be coarser).
We shall identify functions of $\mathcal{B}_{\e}(D,\varphi)$ with their piecewise constant extensions on the union of Voronoi cells 
$\mathcal{C}(x)$ for $x \in D_{\e}\cap \calL$.
%
%
The partition function $Z^{\beta}_{\e,D,\varphi}$ at inverse temperature $\beta>0$ with boundary condition $\varphi\in \Lip(D,\R^n)$ is defined as
\begin{equation}\label{defpartfunc}
Z^{\beta}_{\e,D,\varphi}:=\int_{\mathcal{B}_{\e}(D,\varphi)}\exp(-\beta H_{\e}(D,u))\,\mathrm{d}{u},
\end{equation}
where the integration is understood in the sense of the product measure $\mathrm{d}u=\prod_{x_j \in D_{\e}\cap \calL}\mathrm{d}u(x_j)$, whereas the Helmholtz free energy writes
\begin{equation}\label{deffreeenergy}
\mathcal{E}_{\e}^\beta(D,\varphi):=-\frac{1}{\beta|D_{\e}| }\log \big(Z^{\beta}_{\e,D,\varphi}\big).
\end{equation}

\medskip

We conclude this section by the definition of the Gibbs measure.
For all $\e>0$ and ${v} \in L^p(D)$, we introduce the rescaled version  $u:=\Pi_{1/\e}{v}$ of ${v}$ as
\begin{equation*}
\Pi_{1/\e}{v}: D_\e \to \R^n, z \mapsto \tfrac1\e {v}(\e z).
\end{equation*}
%
We define the Gibbs measure $\mu^{\beta}_{\e,D,\varphi}$ at temperature $\beta$ associated with the Hamiltonian $H_{\e}(D,\cdot)$ 
and the boundary condition $\varphi$ as the probability measure on $L^p(D,\R^n)$ characterized by
\begin{equation}\label{defgibbs}
L^p(D)\ni V\mapsto \mu^{\beta}_{\e,D,\varphi}(V):=\frac{1}{Z^{\beta}_{\e,D,\varphi}}\int_{\Pi_{1/\e}V\cap \mathcal{B}_{\e}(D,\varphi)}\exp(-\beta H_{\e}(D,u))\,\mathrm{d}{u},
\end{equation}
where we divided by the partition function
to ensure that $\mu^{\beta}_{\e,D,\varphi}(L^p(D))=1$ {(see Section \ref{Sec4} for a rigorous definition).}
The main aim of this article is to study the thermodynamic limit of $\mathcal{E}_{\e}^\beta(D,\varphi)$ and $\mu^{\beta}_{\e,D,\varphi}$,
that is their asymptotic behavior as $\frac1\e \uparrow \infty$ (large-volume limit).

In all the results to come, and in the proofs, quantities of interest are random variables (or random measures or functionals). As such, they depend on the realization of the random graph.
We do not make this dependence explicit in the notation, except when it is strictly necessary (in which case we put  an additional argument, e.g. we write $H_{\e}(D,u,G)$ instead of $H_{\e}(D,u)$).

\subsection{Thermodynamic limit}

We start with the convergence of the Helmholtz free energy
for linear boundary conditions ${\overline{\varphi}_\Lambda:x\mapsto \Lambda x}$ and the definition of the limiting (free) energy density of the continuum hyperelastic model.
\begin{theorem}\label{th:W}
Assume Hypothesis~\ref{Hypo1}.
Then for all $\beta>0$ there exists a deterministic quasiconvex function $\overline W^\beta:\R^{{n\times d}} \to {\R}$ 
satisfying the two-sided $p$-growth condition
\begin{equation*}
\forall\Lambda \in \R^{{n\times d}}:\qquad \frac1C |\Lambda|^p-C \le \overline W^\beta(\Lambda)\le C(1+|\Lambda|^p),
\end{equation*}
and {for} all bounded Lipschitz domains $D\subset \R^d$,
the Helmholtz free energy defined in \eqref{deffreeenergy} satisfies almost surely 
$$
\forall\Lambda \in  \R^{{n\times d}}: \qquad \lim_{\e \downarrow 0} \mathcal{E}_{\e}^\beta(D,{\overline{\varphi}_\Lambda})\,=\,\overline{W}^\beta(\Lambda).
$$
\end{theorem}
The extension of this result to general boundary conditions $\varphi \in \Lip(D,\R^n)$ is as follows,
and implies the convergence of the Helmholtz free energy to the infimum of 
an energy functional associated with the free energy density $\overline{W}^\beta$ --- a continuum hyperelastic model.
\begin{theorem}\label{th:Helmholtz}
Assume Hypothesis~\ref{Hypo1} and for all $\beta>0$, let $\overline{W}^\beta$ be the well-defined energy density of Theorem~\ref{th:W}.
Then for all bounded Lipschitz domains $D\subset \R^d$ we have almost surely
for all boundary conditions $\varphi \in \Lip(D,\R^n)$
$$
\lim_{\e \downarrow 0} \mathcal{E}_{\e}^\beta(D,\varphi)\,=\,\inf\Big\{ \fint_D \overline W^\beta(\nabla u(x))dx\,:\,
u\in \varphi+W^{1,p}_0(D)\Big\}.
$$
\end{theorem}
We conclude the study of the thermodynamic limit by establishing a large-deviation principle which 
ensures that the Gibbs measure concentrates as $\e \downarrow 0$ on states that minimize
the energy functional associated with $\overline{W}^\beta$ in the set of continuum deformations 
that satisfy the boundary condition $\varphi$. For a general introduction to the subject we refer to \cite{DeZe}.
\begin{theorem}\label{LDP}
Assume Hypothesis~\ref{Hypo1} and for all $\beta>0$, let $\overline{W}^\beta$ be the well-defined energy density of Theorem~\ref{th:W}.
Then for all bounded Lipschitz domains $D\subset \R^d$, almost surely, and for all boundary conditions $\varphi \in \Lip(D,\R^n)$, the measure $\mu^\beta_{\e,D,\varphi}$ satisfies a strong large deviation principle with speed ${(\beta}|D_{\e}|{)}^{-1}$ and good rate functional $\mathcal{I}^\beta_{D,\varphi}:L^p(D,\R^n)\to [0,+\infty]$ finite only on $\varphi+W^{1,p}_0(D,\R^n)$
and characterized by 
\begin{equation*}
\varphi+W^{1,p}_0(D,\R^n)\ni u \mapsto \mathcal{I}^\beta_{D,\varphi}(u):=\fint_D \overline W^\beta(\nabla u(x))\,\mathrm{d}x-\inf_{v\in \varphi+W^{1,p}_0(D,\R^n)}\fint_D \overline W^\beta(\nabla v(x))\,\mathrm{d}x.
\end{equation*}
More precisely, for every open and closed sets $U\subset L^p(D,\R^n)$ and $V\subset L^p(D,\R^n)$,
we have
\begin{eqnarray*}
\liminf_{\e\downarrow 0}\frac{1}{|D_{\e}|}\log(\mu^\beta_{\e,D,\varphi}(U))\geq -\inf_{u\in U}\mathcal{I}^\beta_{D,\varphi}(u),
\\
\limsup_{\e\downarrow 0}\frac{1}{|D_{\e}|}\log(\mu^\beta_{\e,D,\varphi}(V))\leq -\inf_{u\in V}\mathcal{I}^\beta_{D,\varphi}(u).
\end{eqnarray*}
\end{theorem}
\begin{remark}\label{r.quasiconvexity}
In the scalar case in some regimes the strict convexity of $\overline{W}^{\beta}$ is known (see, e.g., \cite{AKM,CDM,CD,DGI}). Given such a result the large deviation principle immediately implies that the Gibbs measures converge to the Dirac measure supported on the unique minimizer of the rate functional. However, in our vectorial setting, in general we don't even expect convexity of $\overline{W}^{\beta}$; see also \cite[Lemma 8]{KoLu}.
\end{remark}
In the following section we complete the study of the thermodynamic limit by analyzing the behavior of $\overline W^\beta$
and $\mathcal I^\beta_{D,\varphi}$ when the temperature tends to zero, i.e.~in the regime $\beta\uparrow \infty$.

\subsection{Zero-temperature limit}

A natural guess for the zero-temperature limit of the Helmholtz free energy $\overline W^\beta$ is the large-volume
limit of the infimum of the Hamiltonian $H_\e(D,\cdot)$.
The following result establishes rigorously the $\Gamma$-convergence of the rate functional $\mathcal I^\beta_{D,\varphi}$
towards the $\Gamma$-limit of the discrete Hamiltonian {studied in \cite{ACG2}}, and therefore indeed implies the commutation of the limits 
$\beta \uparrow +\infty$ and $\e \downarrow 0$. Note that we require the stronger Hypothesis~\ref{Hypo2}.
\begin{theorem}\label{th:smallT}
Assume Hypothesis~\ref{Hypo2} and for all $\beta>0$, let $\overline{W}^\beta$ be the well-defined energy density of Theorem~\ref{th:W}
and for all Lipschitz domains $D$ and boundary conditions $\varphi \in \Lip(D,\R^n)$, let $\mathcal{I}^\beta_{D,\varphi}:L^p(D,\R^n)\to [0,+\infty]$ be the rate functional of Theorem~\ref{LDP}.
Then, as $\beta \uparrow +\infty$, $\mathcal{I}^\beta_{D,\varphi}$ almost-surely $\Gamma(L^p)$-converges towards the integral functional
$\mathcal{I}^\infty_{D,\varphi}:L^p(D,\R^n)\to [0,+\infty]$ finite only on $\varphi+W^{1,p}_0(D,\R^n)$
and characterized by 
\begin{equation*}
\varphi+W^{1,p}_0(D,\R^n)\ni u \mapsto \mathcal{I}^\infty_{D,\varphi}(u):=\fint_D \overline W^\infty(\nabla u(x))\,\mathrm{d}x-\inf_{v\in \varphi+W^{1,p}_0(D,\R^n)}\fint_D \overline W^\infty(\nabla v(x))\,\mathrm{d}x,
\end{equation*}
where $\overline W^\infty$ is an almost-surely well-defined quasiconvex energy density satisfying the two-sided growth condition 
\begin{equation}\label{e.p-growth2}
\forall\Lambda \in \R^{d\times n}:\qquad \frac{1}{C} |\Lambda|^p-C \le \overline W^{{\infty}}(\Lambda)\le C( |\Lambda|^p+1),
\end{equation}
and given for all $\Lambda \in \R^{n\times d}$ by 
$$
\overline{W}^\infty(\Lambda):=\lim_{\e \downarrow 0} \inf_{u \in \mathcal{B}_{\e}(D',\varphi_\Lambda)} {\frac{1}{|D'_{\e}|}}H_{\e}(D',u),
$$
for any Lipschitz bounded domain $D' \subset \R^d$. In addition, for all $\Lambda \in \R^{n\times d}$, 
\begin{equation}\label{e.diff-Tinfinite}
|\overline W^\infty(\Lambda)-\overline W^\beta(\Lambda)|\,\le\,\frac{\log \beta}{\beta}C\big(1+|\Lambda|^{p-1}\big) .
\end{equation}
\end{theorem}
{\begin{remark}\label{onsmallTtheorem}
Theorem~\ref{th:smallT} implies in particular that the minimizers of the rate functionals $\mathcal{I}^{\beta}_{D,\varphi}$ given by Theorem~\ref{LDP} at inverse temperature $\beta$ converge weakly in $W^{1,p}(D,\R^n)$ to minimizers of $I^{\infty}_{D,\varphi}$. Moreover, due to equicoercivity of both functionals, from \cite[Proposition~1.18]{GCB} we infer that for every open and closed sets $U\subset L^p(D,\R^n)$ and $V\subset L^p(D,\R^n)$
\begin{itemize}
	\item[(i)] $\limsup_{\beta\uparrow +\infty}\inf_{u\in U}\mathcal{I}^{\beta}_{D,\varphi}(u)\leq \inf_{u\in U}\mathcal{I}^{\infty}_{D,\varphi}(u)$;
	\item[(ii)] $\liminf_{\beta\uparrow +\infty}\inf_{u\in V}\mathcal{I}^{\beta}_{D,\varphi}(u)\geq \inf_{u\in V}\mathcal{I}^{\infty}_{D,\varphi}(u)$.
\end{itemize}
Those inequalities allow to pass to the limit $\beta\to +\infty$ in the inequalities of the large deviation principle. For quadratic functionals we shall prove a much stronger statement in Corollary \ref{cor.quad}, namely the limit free energy and the density of the $\Gamma$-limit differ only by a $\beta$-dependent constant (which does not affect minimization). This provides a rigorous justification of the so-called phantom model (for which the free energies of polymer-chains are assumed to be Gaussian), an elementary linear model of polymer physics (see e.g.~\cite[Section~7.2.2]{Rub}).
\end{remark}

\begin{remark}
If Hypothesis~\ref{Hypo2} is replaced by Hypothesis~\ref{Hypo3}, the conclusion \eqref{e.diff-Tinfinite} can be strengthened to
\begin{equation}\label{e.diff-Tinfinite-hypo3}
|\overline W^\infty(\Lambda)-\overline W^\beta(\Lambda)|\,\le\,\frac{\log \beta}{\beta}(C_2'' |\Lambda|^2+C''_p|\Lambda|^{p}+d\big) ,
\end{equation}
for some $C''_2$ and $C''_p$ depending on $d$, $p$, $C,C',C_2,C_2',C_p,C_p'$.
\end{remark}

\subsection{Application to polymer physics}\label{subsec:polyph}
In this last section of the introduction, we apply the above results to the physical model of polymer physics introduced above.
To this aim, we make precise the form of the free energies of isolated polymer-chains in function of the number of monomers in the chain. The Kuhn and {Gr\"un} formula (see e.g.~\cite{KuhnGrun} and \cite[Section~3.4]{Rub}) for the free energy of an isolated chain made of $N$ monomers of size $\ell$  with end-to-end length $L$ at temperature $\beta$ is given by
\begin{equation*} 
 f^\beta(L,N)\,:=\,\frac1\beta N\left( \frac{L}{N\ell}\theta\left(\frac{L}{N\ell}\right)+\log \frac{\theta\left(\frac{L}{N\ell}\right)}{\sinh\theta\left(\frac{L}{N\ell}\right)} \right),
\end{equation*}
where $\theta$ is the inverse of the Langevin function $t\mapsto \coth t- \frac{1}{t}$. In particular, $L\mapsto f^\beta(L,N)$ is a non-negative convex increasing function in the variable $L^2$, that vanishes at $L=0$ and blows up as $L\uparrow N\ell$.
This formula is based on a self-avoiding random bridge.
For technical considerations, we replace this function by a function with $p$-growth from above and below, which yields our starting point
\begin{equation}\label{e.KuhnGrun}
 f^{\beta,(p)}(L,N)\,:=\,\frac{N} \beta f^{(p)}\left(\frac{L}{N\ell}\right),
\end{equation}
where $f^{(p)}$ is a suitable approximation of $t\mapsto  f(t):=t\theta\left(t\right)+\log \frac{\theta\left(t\right)}{\sinh\theta\left(t\right)}$
(that remains convex and increasing). At order $p=10$, a Taylor-expansion (cf.~\cite{GLTV}) simply yields
\begin{equation*}
 f^{\beta,(10)}(L,N)\,=\,\frac{N}{\beta}\left[\frac{3}{2} \left( \frac{L}{N\ell} \right)^2 +\frac{9}{20}\left( \frac{L}{N\ell} \right)^4 + \frac{9}{350} \left( \frac{L}{N\ell} \right)^6 +\frac{81}{7000} \left( \frac{L}{N\ell} \right)^8 
+\frac{243}{673750} \left( \frac{L}{N\ell} \right)^{10}\right].
\end{equation*}

\medskip

Consider now an ergodic random graph $G^\circ=(\mathcal L,\B,\T)$,  a fixed inverse temperature $\beta^\circ\ge 1$, and fix $p$ (say, $p=10$).
Recall that we assume that the length of an edge $b\in \B$ of the random graph writes $\sqrt{N_b^\circ}\ell$, which we use to define the number $N_b^\circ$ of monomers  in the polymer-chain $b$.
We denote by $N^\circ:=\expec{N_b^\circ:b\in \B}$ the average number of monomers per polymer-chain in the graph.
We then rewrite \eqref{e.KuhnGrun} in terms of the deformation ratio
$\lambda=\frac{L}{\sqrt{N_b^\circ}\ell}$ as
\begin{equation*}
f^{\beta^\circ}(L,N_b^\circ)\,:=\,\frac{N_b^\circ}{\beta^\circ}  f^{(p)}\left(\lambda \frac{1}{\sqrt{N_b^\circ}}\right),
\end{equation*}
and make the volumetric term more precise by considering for some $K>0$
$$
W(\Lambda):=  \frac1K W_\vol(\det \Lambda),
$$
where $W_\vol:\R\to \R_+$ is a convex function that is minimal at $t=1$ and satisfies the growth condition
\begin{equation}
\forall t\ge 0:\quad W_\vol(t)\,\le\, 1+t^{\frac pd}.
\end{equation}
For all Lipschitz domains $D$ and microscopic deformations $u:\calL\cap D_{\e}\to\R^d$,
the discrete Hamiltonian takes the form for all $\e>0$
\begin{equation*} 
H_{\e}^\circ(D,u)=\sum_{\substack{(x,y)\in \B\\ x,y\in D_{\e}}}\frac{N_{xy}^\circ}{\beta^\circ} f^{(p)} \left(\frac{|u(x)-u(y)|}{|x-y|} \frac1{\sqrt{N^\circ_{xy}}}\right)+
\sum_{\mathcal{C} \in\mathcal{V}_{1,\e}(D)}|\mathcal{C}|\frac1KW_\vol \Big( {\det}_{\mathcal C}(\nabla u_{\rm aff})\Big),
%
\end{equation*}
which we rewrite in the equivalent form
\begin{equation}\label{e.good-form}
H_{\e}^\circ(D,u)=\frac{N^\circ}{\beta^\circ}\tilde H^\circ_{\e}(D,u),
\end{equation}
where $\tilde H^\circ_{\e}(D,u)$ is given   by
\begin{equation*}
 \tilde H^\circ_{\e}(D,u) := \sum_{\substack{(x,y)\in \B\\ x,y\in D_{\e}}}f^{\circ,(p)}_{xy}\left(\frac{|u(x)-u(y)|}{|x-y|}\right)+
 \sum_{\mathcal{C} \in\mathcal{V}_{1,\e}(D)}|\mathcal{C}|W_\vol^\circ \Big( {\det}_{\mathcal C}(\nabla u_{\rm aff})\Big),
\end{equation*}
and for all $\lambda\ge 0, t\in \R$,
\begin{equation*}
f^{\circ,(p)}_{xy}(\lambda)\,:=\,\frac{N^\circ_{xy}}{N^\circ} f^{(p)} \left(\lambda \frac1{\sqrt{N^\circ_{xy}}} \right),
\qquad
W_\vol^\circ(t)\,:=\, \frac{\beta^\circ}{N^\circ} \frac1KW_\vol(t).
\end{equation*}
In terms of scaling, since volumetric and entropic terms compete, we choose $\frac{\beta^\circ}K \sim 1$, in which case 
\eqref{Hyp-3.1} and \eqref{Hyp-3.2} are valid for $p=10$ with the constants
\begin{equation}\label{e.parameters}
C \sim \frac1{N^\circ}, C_2\sim \frac{3}{2{N^\circ}},C_{10} \sim \frac{243}{673750 \sqrt{N^\circ}^5}.
\end{equation}
We are in the position to apply our general results.
In what follows, $\beta^\circ$ and $N^\circ$ are fixed physical quantities,
whereas $\beta_1$ and $N_1$ are  dummy variables.

\subsubsection{Thermodynamic limit for $H_{\e}^\circ(D,u)$}

By Theorems~\ref{th:W}, \ref{th:Helmholtz}, and~\ref{LDP}, for all temperatures $\beta_1$ there exists a macroscopic energy density $\overline W^{\circ,\beta_1}_{N^\circ}$ associated with the Hamiltonian $H_{\e}^\circ(D,u)$ (recall that $\beta^\circ$ and $N^\circ$ are \emph{fixed} parameters)
via
\begin{equation}\label{e.def-W-orig}
\forall\Lambda \in  \R^{d\times d}: \qquad \lim_{\e \downarrow 0}-\frac1{\beta_1 |D_\e|}\log \int_{\mathcal{B}_{\e}(D,\varphi_\Lambda)}\exp(-\beta_1 H_{\e}^\circ(D,u))\,\mathrm{d}{u}\,=\,\overline W^{\circ,\beta_1}_{N^\circ}(\Lambda).
\end{equation}
For the physical choice $\beta_1=\beta^\circ$, this implies that the free energy of the discrete network of polymer-chains and the associated Gibbs measure are well-described at the thermodynamic limit (with given Dirichlet boundary data $\varphi$) by the infimum of the continuum energy functional
$$
\varphi+W^{1,p}_0(D)\ni u \mapsto \mathcal I^{\circ,\beta_1}_{N^\circ}(u):= \fint_{D} \overline W^{\circ,\beta_1}_{N^\circ} (\nabla u(x))dx,
$$
and by the Dirac mass at the set of minimizers.
Next we argue that a direct application of Theorem~\ref{th:smallT} does not allow to justify the two-temperatures model which
amounts to taking the limit $\beta_1\uparrow \infty$ while keeping $\beta^\circ$ fixed.
In this setting, Theorem~\ref{th:smallT} yields the existence of some energy density  $\overline W^{\circ,\infty}_{N^\circ}$ 
such that 
$$
\forall\Lambda \in  \R^{d\times d}: \qquad \lim_{\beta_1\uparrow \infty}\overline W^{\circ,\beta_1}_{N^\circ}(\Lambda)\,=\,\overline W^{\circ,\infty}_{N^\circ}(\Lambda).
$$
However, the quantitative estimate \eqref{e.diff-Tinfinite} of Theorem~\ref{th:smallT}, that takes the form
\begin{equation}\label{e.bad-grail}
|\overline W^{\circ,\beta_1}_{N^\circ}(\Lambda)-\overline W^{\circ,\infty}_{N^\circ}(\Lambda)|\,\le \, \frac{\log \beta_1}{\beta_1}(d+\frac1{\beta^\circ}C(1+|\Lambda|^p)),
\end{equation}
is not precise enough for $\beta_1=\beta^\circ$ since $\overline W^{\circ,\infty}_{N^\circ}$ is itself of order $\frac1{\beta^\circ} C(1+|\Lambda|^p)$.
The rest of this section aims at justifying the two-temperatures model in the regime $N^\circ\gg 1$ rather than $\beta^\circ\gg 1$.

\subsubsection{Thermodynamic limit for $\tilde H^\circ_{\e}(D,u)$}

We denote by $\overline W^{\circ,N_1}$ the macroscopic free energy at temperature ``$N_1$''
 (the number of monomers will indeed play
the role of an inverse physical temperature in what follows)
associated with the Hamiltonian $\tilde H^\circ_{\e}$ via Theorem~\ref{th:W}, that is, 
$$
\forall\Lambda \in  \R^{d\times d}: \qquad \lim_{\e \downarrow 0}-\frac1{N_1 |D_\e|}\log \int_{\mathcal{B}_{\e}(D,\varphi_\Lambda)}\exp(-N_1 \tilde H^\circ_{\e}(D,u))\,\mathrm{d}{u}\,=\,\overline W^{\circ,N_1}(\Lambda).
$$
In view of \eqref{e.good-form} and \eqref{e.def-W-orig}, we have the identity 
\begin{equation}\label{e.Nbetacirc}
 \overline W^{\circ,\beta_1}_{N^\circ} |_{\beta_1=\beta^\circ} \,=\, \frac{N^\circ}{\beta^\circ} \overline W^{\circ,N_1}|_{N_1=N_\circ}.
\end{equation}
Whereas the $\overline W^{\circ,\beta_1}_{N^\circ}$ is well-suited to take the zero-temperature limit $\beta_1\uparrow \infty$,
$\overline W^{\circ,N_1}$ is well-suited  to take the limit of large number of monomers per chain $N_1\uparrow \infty$.
By Theorem~\ref{th:smallT} (in form of \eqref{e.diff-Tinfinite-hypo3}), there exists a macroscopic energy density 
$\overline W^{\circ,\infty}$ such that for all $N_1\gg 1$
\begin{equation}\label{grail}
\forall \Lambda \in \R^{d\times d}: \quad |\overline W^{\circ,\infty}(\Lambda)-\overline W^{\circ,N_1}(\Lambda)|\,\le\,\frac{\log N_1}{N_1}\Big(C_2'' |\Lambda|^2+C_{p}''|\Lambda|^p+d\Big),
\end{equation}
and so that the integral functional $u \mapsto\mathcal I^{\circ,N_1}:= \fint_{D} \overline W^{\circ,N_1} (\nabla u(x))dx$ $\Gamma(L^{p})$-converges towards  $u \mapsto \mathcal I^{\circ,\infty}(u):= \fint_{D} \overline W^{\circ,\infty} (\nabla u(x))dx$ on $\varphi+W^{1,p}_0(D)$ as $N_1\uparrow \infty$.
Note that the lower and upper bounds in \eqref{e.p-growth2} are crude and could be largely
improved if more precise assumptions are made on the random graph --- in particular,  we expect the coefficients of the terms of order $|\Lambda|^p$ to be comparable in both sides of the two-sided estimate, so that the RHS of \eqref{grail} would indeed scale like $\frac{\log N_1}{N_1}$ times the order of 
magnitude of $\overline W^{\circ,N_1}(\Lambda)$.

\subsubsection{Justification of the two-temperatures model in the regime $N^\circ\gg 1$} 
The combination of \eqref{grail} and \eqref{e.Nbetacirc} yields
\begin{equation}\label{grail2}
\forall \Lambda \in \R^{d\times d}: \quad |\frac{N^\circ}{\beta^\circ} \overline W^{\circ,\infty}(\Lambda)-\overline W^{\circ,\beta^\circ}_{N^\circ}(\Lambda)|\,\le\,\Big(\frac{\log N^\circ}{N^\circ} \Big)\times \frac{N^\circ}{\beta^\circ}\Big(C_2'' |\Lambda|^2+C_{p}''|\Lambda|^p+d\Big).
\end{equation}
In view of the parameters \eqref{e.parameters} and lower bounds for the $\Gamma$-limit, for deformations $\Lambda$ such that $ |\Lambda| \sim \sqrt{N^\circ}$ (that is, in the nonlinear regime), we have
\begin{equation*}
\begin{array}{rcl}
 \overline W^{\circ,\infty}(\Lambda) &\sim&   C_2 |\Lambda|^2+(C_{p}+C)|\Lambda|^p 
 \\
& \sim &  C_2'' |\Lambda|^2+C_{p}''|\Lambda|^p 
\end{array}
 \gtrsim 1
 \quad \implies \quad 
\frac{N^\circ}{\beta^\circ} \overline W^{\circ,\infty}(\Lambda) \gtrsim  \frac{N^\circ}{\beta^\circ}\Big(C_2'' |\Lambda|^2+C_{p}''|\Lambda|^p+d\Big),
\end{equation*}
so that \eqref{grail2} shows that the relative error between $\overline W^{\circ,\beta^\circ}_{N^\circ}$ and its approximation
$\frac{N^\circ}{\beta^\circ} \overline W^{\circ,\infty}$ is of order $\frac{\log N^\circ}{N^\circ}\ll 1$ in the regime $N^\circ\gg 1$ of large number of monomers per polymer-chain.  
Combined with the observation that the identity \eqref{e.good-form} also yields
$$
\forall \Lambda \in \R^{d\times d}: \quad\frac{N^\circ}{\beta^\circ} \overline W^{\circ,\infty}(\Lambda)=\lim_{\beta_1\uparrow \infty} 
\overline W^{\circ,\beta_1}_{N^\circ}(\Lambda)=\overline W^{\circ,\infty}_{N^\circ}(\Lambda),
$$
\eqref{grail2} takes the form 
\begin{eqnarray*}
\forall \Lambda \in \R^{d\times d}: && | \overline W^{\circ,\infty}_{N^\circ}(\Lambda)-\overline W^{\circ,\beta^\circ}_{N^\circ}(\Lambda)|\,\le\,\Big(\frac{\log N^\circ}{N^\circ} \Big)\times \frac{N^\circ}{\beta^\circ}\Big(C_2'' |\Lambda|^2+C_{p}''|\Lambda|^p+d\Big),
\\
&&|\Lambda|\sim \sqrt{N^\circ} \implies  \overline W^{\circ,\infty}_{N^\circ}(\Lambda) \gtrsim  \frac{N^\circ}{\beta^\circ}\Big(C_2'' |\Lambda|^2+C_{p}''|\Lambda|^p+d\Big),
\end{eqnarray*}
which improves on \eqref{e.bad-grail}.
The above applications of Theorems~\ref{th:W}, \ref{th:Helmholtz},~\ref{LDP}, and~\ref{th:smallT} therefore yield a rigorous justification of the two-temperatures model $\overline W^{\circ,\infty}_{N^\circ}$, which consists in assuming that the monomers of the polymer-chains fluctuate at inverse temperature $\beta^\circ$, whereas
cross-links are considered at zero temperature ($\beta_1=+\infty$).
This sets on rigorous ground the approach introduced and analyzed in \cite{GLTV,ACG2} to derive nonlinear elasticity from polymer physics, and concludes this series of works.

\subsubsection{Extensions and comments} 

The process of vulcanization of rubber generates metallic inclusions (zinc oxides) in the matrix phase, which modifies the elastic behavior of rubber-like materials at large deformation since the former are more rigid than the polymer-chains.
This can be included in the discrete model as follows.
Enrich the probability space by adding a state $Z\in \{0,1\}^{\N}$, and say that a vertex $i$ is in the set of zinc oxides if $Z(i)=1$.
If an edge $b=(x_i,x_j) \in \B$ is such that $Z(i)+Z(j)\ge 1$, then the free energy of the polymer-chain $f_{ij}$ is multiplied by some large constant $K\gg 1$ (which encodes the larger rigidity of the zinc oxides). We may then perform the same analysis as above.

\medskip

We now comment on the main two analytical simplifications of this work, namely that $f_{ij}$ and $W$ have $p$-growth from above.
We believe that at least parts of the results survive if we let $f_{ij}$ blow up at finite deformation, following the approach developed by Duerinckx and the second author in \cite{DG} in the continuum setting.
In contrast, the growth condition on $W$ is  crucial for our arguments to work. Relaxing this assumption constitutes the major open problem of homogenization of integral functionals with quasiconvex integrands. For first results in that direction (with small data) we refer to \cite{NS17}.

\subsection{Outline of the article}

The rest of this article is organized as follows.
To simplify the exposition, we assume from Section~\ref{Sec2} to Section~\ref{Sec5} that $W\equiv 0$ in \eqref{def:hamiltonian}.
In Section~\ref{Sec2} we introduce some notation and establish some preliminary results that are essentially of geometric nature.
Section~\ref{Sec3} is dedicated to the definition of $\overline{W}^\beta$ and to the proof of Theorem~\ref{th:W}.
The large-deviation principle for the Gibbs measure is addressed in Section~\ref{Sec4} and the Theorem~\ref{LDP} is proved. {As a by-product of the argument we obtain Theorem~\ref{th:Helmholtz}.}
The small-temperature limit is analyzed in Section~\ref{Sec5}, where we prove  Theorem~\ref{th:smallT}.
The last section is dedicated to some extensions, namely the addition of the volumetric term $W$ to \eqref{def:hamiltonian}.

\tableofcontents

\section{Notation and preliminary geometric estimates}\label{Sec2}

Let us fix some notation. Given a measurable set $B\subset\mathbb{R}^d$ we denote by $|B|$ its $d$-dimensional Lebesgue measure. The same notation is used to denote the cardinality of $B$ whenever it is a finite set. More generally we denote by $\mathcal{H}^k(B)$ the $k$-dimensional Hausdorff measure of $B$. Given $x\in\mathbb{R}^d$ we let $|x|$ denote its Euclidean norm and we let $B_r(x)$ be the open ball with center $x$ and radius $r$. Moreover, $Q(x,r)=x+(-r/2,r/2)^d$ denotes the open cube with center $x$ and side length $r$. We set $\dist(x,B)=\inf_{y\in B}|x-y|$.  Given an open set $U\subset\mathbb{R}^d$ we define $\mathcal{A}^R(U)$ to be the family of open, bounded subsets of $U$ with Lipschitz boundary. We denote by $L^p(U,\mathbb{R}^n),\,W^{1,p}(U,\mathbb{R}^n)$ the usual vector-valued Lebesgue and Sobolev spaces. We use the short-hand notation  $L^p(U)$ or $W^{1,p}(U)$ when we refer to convergence in these spaces and no confusion about the co-domain is possible. In the proofs $C$ denotes a generic constant (depending only on the dimension or other fixed parameters) that may change every time it appears.
%
%


\subsection{Geometric considerations}
In this subsection we establish some geometric properties of admissible extended Euclidean graphs that will be useful throughout this article. Recall that given $G=(\mathcal{L},E,S)$, we denote by $\mathcal{V}=\{\mathcal{C}(x)\}_{x\in\Lw}$ the Voronoi tessellation associated to the vertices $\Lw$.  
Note that if the vertices fulfill conditions (i) and (ii) of Definition \ref{defadmissible}, then the Voronoi cells satisfy $B_{\frac{r}{2}}(x)\subset \mathcal{C}(x)\subset B_R(x)$ for all $x\in\Lw$. In particular it holds that 
\begin{equation}\label{e.volVoronoi}
\forall x\in \Lw:\quad \frac{1}{C}\leq |\mathcal{C}(x)|\leq C
\end{equation}
and, for fixed $O\in\Ard$ and $\e$ small enough, we have the estimate
\begin{equation}\label{volumeestimate}
\frac{1}{C}|O|\e^{-d}\leq |O_{\e}^\calL|\leq C|O|\e^{-d}.
\end{equation} 
%
In some geometric constructions we will also need a bound on the cardinality of sets of the form
\begin{equation*}
\{x\in\Lw:\;\dist(x,\partial O_{\e})\leq C_0\}.
\end{equation*} 
For $x$ in this set the rescaled Voronoi cells $\e\mathcal{C}(x)$ are contained in the $(C_0+R)\e$-tubular neighbourhood of $\partial O$. Since for Lipschitz boundaries, the Minkowski content agrees (up to a dimensional constant) with $\mathcal{H}^{d-1}(\partial O)$, we deduce that for $\e$ small enough we have
\begin{equation}\label{surfaceesimtate}
|\{x\in\Lw:\;\dist(x,\partial O_{\e})\leq C_0\}|\leq C\e^{1-d}\mathcal{H}^{d-1}(\partial O).
\end{equation}
Similar estimates hold for finite unions or intersections of Lipschitz sets. 

We shall identify functions $u:\Lw\to\R^n$ with their piecewise constant interpolations on the Voronoi tessellation $\mathcal{V}$
associated with $\mathcal L$. Conversely, given a function $u\in L_{\rm loc}^p(\R^d,\R^n)$ we define a (random) discrete approximation $u_{\e}:\Lw\to\R^n$ via
\begin{equation}\label{discreteapprox}
u_{\e}(x):=\frac{1}{|\e \mathcal{C}(x)|}\int_{\e \mathcal{C}(x)}u(z)\,\,\mathrm{d}z.
\end{equation}
\begin{remark}\label{approximation}
The rescaled (piecewise constant) functions $\tilde{u}_{\e}(\e x):=u_{\e}(x)$ converge to $u$ in $L_{\rm loc}^p(\R^d,\R^n)$. Indeed, given any bounded set $B\subset\R^d$ we choose another bounded, open set $U\subset\R^d$ such that $B\subset\subset U$ and redefine $u\equiv 0$ on $\R^d\backslash U$. This does not affect the values of $\tilde{u}_{\e}$ and $u$ on $B$, but now $u\in L^p(\R^d,\R^n)$. Assume that $n=1$. From \eqref{e.volVoronoi} and Lebesgue's differentiation Theorem~we infer that $\tilde{u}_{\e}\to u$ almost everywhere in $B$. Moreover, again by definition~\eqref{discreteapprox}, we have $|\tilde{u}_{\e}|\leq C\mathcal M u$, where $\mathcal M u$ denotes the Hardy-Littlewood maximal function. Hence $\tilde{u}_{\e}\to u$ in $L^p(B,\R^d)$ by dominated convergence. The general case $n\geq 1$ follows by treating each component separately . 	
\end{remark}
For notational convenience we also define discrete $\ell^p$ norms as follows: for all $\e>0$ and $u:\Lw\to\mathbb{R}^n$
\begin{equation*}
\|u\|_{\ell_\e^p(O)}:=\bigg(\sum_{x\in O_{\e}^\calL}|u(x)|^p\bigg)^{\frac{1}{p}},
\quad\quad\|\nabla_{\B} u\|_{\ell_\e^p(O)}:=\bigg(\sum_{\substack{(x,y)\in \B\\ x,y\in O_{\e}^\calL}}|u(x)-u(y)|^p\bigg)^{\frac{1}{p}},
\end{equation*}
where $\nabla_\B $ denotes the gradient on the graph, which maps functions on vertices to functions on edges (and is convenient to estimate the Hamiltonian).

As we show now, admissible graphs enjoy discrete Poincar\'e-type inequalities with respect to these norms. Recall that for any set $O\subset\R^d$ and $\varphi\in \Lip(O,\R^n)$ we let
\begin{equation*}
\mathcal{B}_{\e}(O,\varphi)=\{u:O_{\e}\cap \calL \to\R^n,\;|u(x)-\tfrac1\e \varphi(\e x)|<1\text{ if }\dist(x,\partial D_{\e})\leq C_0\}.
\end{equation*}
\begin{lemma}\label{poincare}
Let $G\in\mathcal{G}$ and let $O\in\mathcal{A}^R(\R^d)$. Then there exists a constant $C=C_{O,p}$ such that for all $\e$ small enough and all $u\in\mathcal{B}_{\e}(O,0)$ we have
\begin{equation*}
\|u\|^p_{\ell_\e^p(O)}\leq \frac{C}{\e^p}\left(\|\nabla_{\B} u\|^p_{\ell_\e^p(O)}+\e^{1-d}\right).
\end{equation*}	
\end{lemma}
\begin{proof}[Proof of Lemma~\ref{poincare}]
We extend $u$ setting $u(x)=0$ for $x\in\Lw\backslash O_{\e}^\calL$. Take any cube $Q\subset\R^d$ such that $O\subset\subset Q$. For $x\in O_{\e}^\calL$, define the ray $R_x:=\{x+te_1:\,t\geq 0\}$. Then there exists a smallest number $t_*>0$ such that $x+t_*e_1\in \mathcal{C}(z)$ for some $z\in\Lw\backslash O_{\e}^\calL$. We let $z_x\in\Lw$ be (one of) such point(s). Then $z_x\in Q_{\e}$ for $\e$ small enough and moreover $|x-z_x|\leq \e^{-1}{\rm diam}\,O+2R$. As $G$ is admissible, there exists a path $P(x)$ connecting $x$ and $z_x$ such that $P(x)\subset [x,z_x]+B_{C_0}(0)$. By \eqref{e.volVoronoi} the number of edges in such a path is bounded by $\#\{(x^{\prime},x^{\prime\prime})\in P(x)\}\leq C\e^{-1}{\rm diam}\,O$ and moreover we may assume that $P(x)\subset Q_{\e}$ for $\e$ small enough. Jensen's inequality then yields
\begin{align}\label{jensen}
|u(x)|^p&=|u(x)-u(z_x)|^p\leq \Big(\sum_{(x^{\prime},x^{\prime\prime})\in P(x)}|u(x^{\prime})-u(x^{\prime\prime})|\Big)^{p}\nonumber
\\
&\leq  C\left(\frac{{\rm diam}\,O}{\e}\right)^{p-1}\sum_{(x^{\prime},x^{\prime\prime})\in P(x)}|u(x^{\prime})-u(x^{\prime\prime})|^p,
\end{align}
where we used that $u(z_x)=0$. Next, for any edge $(x^{\prime},x^{\prime\prime})\in \B$ we set
\begin{equation*}
K_{\e}(x^{\prime},x^{\prime\prime}):=\{x\in O_{\e}^\calL:\;(x^{\prime},x^{\prime\prime})\in P(x)\}.
\end{equation*}
We need to bound the cardinality of this set. If $x\in K_{\e}(x^{\prime},x^{\prime\prime})$, then there exists $\lambda\in[0,1]$ such that the point $x_{\lambda}=x+\lambda (z_x-x)$ satisfies $|x_{\lambda}-x^{\prime}|\leq C_0$. Hence we infer
\begin{equation*}
x=x-x_{\lambda}+x^{\prime}+(x_{\lambda}-x^{\prime})=-\lambda (z_x-x)+x^{\prime}+(x_{\lambda}-x^{\prime})\in (-R_{-x^{\prime}}+B_{R+C_0}(0))\cap\e^{-1}O.
\end{equation*}
By \eqref{e.volVoronoi} we conclude that $\#R_{\e}(x^{\prime},x^{\prime\prime})\leq C\e^{-1}{\rm diam}\,O$, so that summing (\ref{jensen}) over $x\in O_{\e}^\calL$ yields
\begin{equation}\label{nearlypoincare}
\|u\|^p_{\ell_\e^p(O)}\leq C\left(\frac{{\rm diam}\,O}{\e}\right)^{p}\|\nabla_\B u\|^p_{\ell_\e^p(Q)}.
\end{equation}
Due to the constant extension and the soft boundary conditions, for small $\e$ the contributions on the large cube $Q$ can be bounded via the estimate
\begin{equation*}
\|\nabla_\B u\|^p_{\ell_\e^p(Q)}\leq \|\nabla_\B u\|^p_{\ell_\e^p(O)}+\sum_{\substack{(x,y)\in \B\\ [x,y]\cap\partial O_{\e}\neq\emptyset}}|u(x)-u(y)|^p\leq\|\nabla_\B u\|^p_{\ell_\e^p(O)}+C\e^{1-d}\mathcal{H}^{d-1}(\partial O),
\end{equation*}
where we used (\ref{surfaceesimtate}). Inserting this estimate in (\ref{nearlypoincare}) concludes the proof. 
\end{proof}
\begin{remark}\label{reverse}
In the discrete setting there is also a trivial reverse Poincar\'e inequality. Indeed, as the degree of every vertex in $\Lw$ is equibounded due to \eqref{e.volVoronoi}, there exists $C=C_p$ such that  for all $O \in\mathcal{A}^R(\R^d)$ and  $u:\Lw \to\mathbb{R}^n$, $\|\nabla_\B u\|^p_{\ell^p_{\e}(O)}\leq C\|u\|^p_{\ell^p_{\e}(O)}$.	
\end{remark}
Next we prove a technical Lemma~which is the analogue of Lemma~12 in \cite{KoLu} for non-periodic graphs. 
\begin{lemma}\label{treeargument}
Let $G^{\prime}=(\mathcal{L}^{\prime},\B^\prime)$ be a finite connected subgraph of $G$ and $\bar x\in\mathcal{L}^{\prime}$. Then there exists a dimensional constant $C_1$ such that, for all $z\in\R^n$ and $\alpha,\gamma>0$, 
\begin{equation*}
\int_{(\R^n)^{\mathcal{L}^{\prime}}}\mathds{1}_{\{|u(\bar x)-z|<\gamma\}}\exp\big(-\alpha\sum_{(x,y)\in \B^{\prime}}|u(x)-u(y)|^p\big)\,\mathrm{d}u\leq C_1\gamma^n \Big(\a^{-\frac{n}{p}}C_1\Big)^{|\mathcal{L}^{\prime}|-1}.
\end{equation*}
\end{lemma}
\begin{proof}[Proof of Lemma~\ref{treeargument}]
As $G^{\prime}$ is connected, there exists a rooted spanning tree $T_{\bar x}=(\mathcal{L}^{\prime},\B_{\bar x})$ with root $\bar x$ (note that here we exceptionally consider a directed graph). We now prove inductively that we can integrate out all the vertices except the root. Since $T_{\bar{x}}$ has less edges than $G^{\prime}$, it holds that
\begin{equation*}
\sum_{(x,y)\in \B^{\prime}}|u(x)-u(y)|^p\geq \sum_{(x,y)\in \B_{\bar x}}|u(x)-u(y)|^p.
\end{equation*}
Consider any leaf $x_0\in \mathcal{L}^{\prime}$, that means $x_0$ has no outgoing edges and only one incoming edge $(x_1,x_0)\in \B_{\bar x}$. Then, by Fubini's Theorem~and a change of variables, we deduce that
\begin{align*}
&\int_{(\R^n)^{\mathcal{L}^{\prime}}}\mathds{1}_{\{|u(\bar x)-z|<\gamma\}}\exp(-\alpha\sum_{(x,y)\in \B^{\prime}}|u(x)-u(y)|^p)\,\mathrm{d}u\\
&\leq \int_{(\R^n)^{\mathcal{L}^{\prime}}}\mathds{1}_{\{|u(\bar x)-z|<\gamma\}}\exp(-\a\sum_{(x,y)\in \B_{\bar x}}|u(x)-u(y)|^p)\,\mathrm{d}u\\
&\leq \int_{(\R^n)^{\mathcal{L}^{\prime}\backslash x_0}}\mathds{1}_{\{|u(\bar x)-z|<\gamma\}}\exp(-\a\sum_{\substack{(x,y)\in \B_{\bar x}\\ (x,y)\neq (x_1,x_0)}}|u(x)-u(y)|^p)\int_{\R^n}\exp(-\a |u(x_1)-u(x_0)|^p)\,\mathrm{d}u(x_0)\,\mathrm{d}u
\\
&=\left(\a^{-\frac{n}{p}}\int_{\R^n}\exp(-|\zeta|^p)\,\mathrm{d}\zeta\right)\int_{(\R^n)^{\mathcal{L}^{\prime}\backslash x_0}}\mathds{1}_{\{|u(\bar x)-z|<\gamma\}}\exp(-\a\sum_{\substack{(x,y)\in \B_{\bar x}\\ (x,y)\neq (x_1,x_0)}}|u(x)-u(y)|^p)\,\mathrm{d}u.
\end{align*}
The (directed) graph $(\mathcal{L}^{\prime}\backslash \{x_0\},\B_{\bar x}\backslash (x_1,x_0))$ is still a rooted tree for the set of edges $\mathcal{L}^{\prime}\backslash \{x_0\}$ with root $\bar x$. By iteration we thus obtain the claim upon setting 
\begin{equation*}
C_1=\max\Big\{\int_{B_1(0)}\mathrm{d}\zeta,\int_{\R^n}\exp(-|\zeta|^p)\,\mathrm{d}\zeta\Big\},
\end{equation*}
where the volume of the unit ball is the remaining term when we integrated out the contributions of all the edges in $\B_{\bar x}$.	
\end{proof}
\begin{remark}\label{numbercomponents}
Given a set $U\in\mathcal{A}^R(\mathbb{R}^d)$, the graph $G_{U}=(U\cap\Lw,\{(x,y)\in \B:\;x,y\in U\})$ is in general not connected but can be decomposed into its connected components. If $N_U$ denotes the number of such components, then it follows that 
\begin{equation*}
N_U\leq\#\{x\in\Lw:\;\dist(x,\partial U)\leq C_0\}.
\end{equation*} 
Indeed, for any component $G_j=(V_j,\B_j)$ take $x\in V_j$ and $y\in\Lw\backslash U$. As $G$ is connected we find a path in $G$ connecting $x$ and $y$. Starting at $x$, let $y_j$ be the first vertex of the path such that $y_j\notin U$. Then its preceding vertex $x_j$ satisfies $\dist(x_j,\partial U)\leq C_0$ because $G$ is admissible. By construction $x_j\in V_j$.
\end{remark}
Combining Remark~\ref{numbercomponents}, Lemma~\ref{treeargument}, and Fubini's theorem, we immediately obtain the following bound for possibly disconnected subgraphs.
\begin{lemma}\label{forrestargument}
Let $\e>0$. Given $O\in\mathcal{A}^{R}(\R^d)$, we define the graph $G_{O,\e}=(O_{\e}^\calL,\{(x,y)\in \B:\;x,y\in O_{\e}^\calL\})$. Consider a set $V$ such that there exist $\gamma>0$ and $\{z_x\}_{\{x\in\calL:\;\dist(x,\partial O_{\e})\leq C_0\}}\subset\R^n$ with
\begin{equation*}
V\subset \{u:O_{\e}^\calL\to\R^n:\;|u(x)-z_x|<\gamma \text{ for all }x\in\calL\text{ such that }\dist(x,\partial O_{\e})\leq C_0\}.
\end{equation*}
Then there exists $C_1>0$ such that for all $\alpha>0$
\begin{equation*}
\int_V\exp(-\alpha \|\nabla_\B u\|^p_{\ell_\e^p(O)})\,\mathrm{d}u\leq \Big(C_1\gamma^n\Big)^{N_{O,\e}}\Big(\alpha^{-\frac{n}{p}}C_1\Big)^{|O_{\e}^\calL|-N_{O,\e}},
\end{equation*}	
where $N_{O,\e}$ denotes the number of connected components of the graph $G_{O,\e}$.
\end{lemma}

\subsection{Estimates on the partition function}

For the analysis, we need to introduce further functional spaces.
Given $O\in\mathcal{A}^R(\R^d)$, $v \in L_{\rm loc}^p(\R^d,\mathbb{R}^n)$, $w:O_{\e}^\calL\to\R^n$, and $\kappa,M>0$, we define the following three sets:
\begin{equation}\label{defsets}
\begin{split}
\mathcal{N}_p(v,O,\e,\kappa)&:=\{u:O_{\e}^\calL\to\R^n,\;\sum_{O_{\e}^\calL}\e^{d}|v_{\e}(x)-\e u(x)|^p<\kappa^p|O|^{1+\frac{p}{d}}\},
\\
\mathcal{N}_{\infty}(w,O,\e)&:=\{u:O_{\e}^\calL\to\R^n:\;\|w-u\|_{\infty}<1\},
\\
\mathcal{S}_{M}(O,\e)&:=\{u:O_{\e}^\calL\to\R^n:\;H_{\e}(O,u) \leq M|O_{\e}^\calL|\}.
\end{split}
\end{equation}
The first two sets define neighborhoods of $\varphi_{\e}$ (defined via \eqref{discreteapprox}) and $v$, respectively, in a suitable topology. The third set contains deformations of uniformly finite energy. 

Next, we introduce a localized version of the partition function \eqref{defpartfunc}, and define for all sets $V\subset\{u:O_{\e}^\calL\to\R^n\}$  
\begin{equation*}
Z^\beta_{\e,O}(V):=\int_{V}\exp(-\beta H_{\e}(O,u))\,\mathrm{d}u,
\end{equation*}
and two ($\beta$-dependent) quantities that play a major role in the analysis: For $v \in L_{\rm loc}^p(\R^d,\R^n)$ and $O\in\mathcal{A}^R(\R^d)$ we set
\begin{equation*}
\begin{split}
&\mathcal{F}^-_{\kappa}(O,v)=\liminf_{\e\downarrow 0}-\frac{1}{\beta|O_{\e}
|}\log(Z_{\e,O}^\beta(\mathcal{N}_p(v,O,\e,\kappa))),\\
&\mathcal{F}^+_{\kappa}(O,v)=\limsup_{\e\downarrow 0}-\frac{1}{\beta|O_{\e}
|}\log(Z_{\e,O}^\beta(\mathcal{N}_p(v,O,\e,\kappa))).
\end{split}
\end{equation*} 
Since both quantities are decreasing in $\kappa$, we can consider their limits as $\kappa\downarrow  0$ and define
\begin{equation*}
\begin{split}
\mathcal{F}^-(O,v)=\lim_{\kappa\to 0}\mathcal{F}^-_{\kappa}(O,v)=\sup_{\kappa>0}\mathcal{F}^-_{\kappa}(O,v),\\
\mathcal{F}^+(O,v)=\lim_{\kappa\to 0}\mathcal{F}^+_{\kappa}(O,v)=\sup_{\kappa>0}\mathcal{F}^+_{\kappa}(O,v).
\end{split}
\end{equation*}

\medskip

We conclude this section with two results. The first one rules out concentration on high energy configurations, and
the second is an interpolation result, which will both be crucial to prove the exponential tightness at the origin of the
large deviation principle for the Gibbs measure.
%
%
%
\begin{lemma}\label{tightness}
Assume Hypothesis~\ref{Hypo1} and let $G\in\mathcal{G}$. Fix $O\in\mathcal{A}^R(\R^d)$, $v \in L^p_{\rm loc}(\R^d,\R^n)$ and $\varphi\in {\rm Lip}(O,\R^n)$. Then there exists a constant $C_{\beta}>0$ such that for all $\kappa>0$, $\e=\e(\kappa)>0$ small enough, all $\beta>0$ and $M \ge C_{\beta}$, 
\begin{align*}
&Z_{\e,O}^\beta(\mathcal{N}_p(v,O,\e,\kappa)\backslash \mathcal{S}_M(O,\e))\leq \exp(-\frac{M}{2}\beta |O_{\e}^\calL|)\exp(C_{\beta}|O_{\e}^\calL|),
\\
&Z_{\e,O}^\beta(\mathcal{B}_{\e}(O,\varphi)\backslash \mathcal{S}_M(O,\e))\leq \exp(-\frac{M}{2}\beta |O_{\e}^\calL|)\exp(C_{\beta}|O_{\e}^\calL|).
\end{align*}
The constant $C_{\beta}$ can be chosen as
\begin{equation*}
C_{\beta}=\begin{cases}
 C &\mbox{if $\beta\geq \frac{1}{2}$,}\\
-C\log(\beta) &\mbox{$0<\beta<\frac{1}{2}$.}
\end{cases}
\end{equation*}
\end{lemma}
\begin{proof}[Proof of Lemma~\ref{tightness}]
Note that by Hypothesis~\ref{Hypo1}, for any $u\notin\mathcal{S}_M(O,\e)$ it holds that
\begin{equation*}
H_{\e}(u,O)\geq \frac{3M}{4}|O_{\e}^\calL|+\frac{1}{4}H_{\e}(u,O)\geq \frac{3M}{4}|O_{\e}^\calL|+\frac{1}{4C}\|\nabla_\B u\|^p_{\ell^p_{\e}(O)}-\frac{C}{4}|O_{\e}^\calL|.
\end{equation*}
Hence we obtain that
\begin{equation*}
Z_{\e,O}^\beta(\mathcal{N}_p(v,O,\e,\kappa)\backslash \mathcal{S}_M(O,\e))\leq\exp(-\frac{M}{2}\beta|O_{\e}^\calL|)\int_{\mathcal{N}_p(v,O,\e,\kappa)}\exp(-\beta\frac{1}{C}\|\nabla_\B u\|^p_{\ell^p_{\e}(O)})\,\mathrm{d}u,
\end{equation*}
up to redefining $C$.
In order to bound the last integral, first note that for every $u\in\mathcal{N}_p(v,O,\e,\kappa)$ the definition (\ref{defsets}) implies that for all $x\in O_{\e}^\calL$ we have
\begin{equation*}
|u(x)-\e^{-1}v_{\e}(x)|<\kappa (|O|\e^{-d})^{\frac{1}{p}+\frac{1}{d}}.
\end{equation*}
Therefore we may apply Lemma~\ref{forrestargument} with the family $z_x=\e^{-1}v_{\e}(x)$ and obtain the estimate
\begin{equation}\label{roughbound}
\int_{\mathcal{N}_p(v,O,\e,\kappa)}\exp(-\beta\frac{1}{C}\|\nabla_\B u\|^p_{\ell^p_{\e}(O)})\,\mathrm{d}u\leq \left( C\kappa^n\left(|O|\e^{-d}\right)^{\frac{n}{p}+\frac{n}{p}}\right)^{N_{O,\e}}\left(C{\beta}^{-\frac{n}{p}}\right)^{|O_{\e}^\calL|-N_{O,\e}},
\end{equation}
where the graph $G_{O,\e}$ is defined as in Lemma~\ref{forrestargument}. By Remark~\ref{numbercomponents} and (\ref{surfaceesimtate}), taking $\e$ small enough (depending on $O$) the number of connected components of $G_{O,\e}$ can be bounded via
\begin{equation*}
N_{O,\e}\leq C\mathcal{H}^{d-1}(\partial O)\e^{1-d}.
\end{equation*}
Set $C_{\beta}$ as in the statement. Up to further decreasing $\e=\e(O,\kappa)$ , we deduce from \eqref{roughbound} the estimate
\begin{equation*}
\int_{\mathcal{N}_p(v,O,\e,\kappa)}\exp(-\frac{1}{C}\|\nabla_\B u\|^p_{\ell^p_{\e}(O)})\,\mathrm{d}u\leq \exp(C_{\beta}|O_{\e}^\calL|).
\end{equation*}
This proves the first estimate. The second one is easier as we have a better control for Lemma~\ref{forrestargument} using the boundary conditions. We leave the details to the reader.
\end{proof}
\begin{remark}\label{Mdepends}
Observe that in Lemma~\ref{tightness} the condition on $\e$ is independent of $M$, so that the estimate holds
uniformly with respect to $M\ge C$.
\end{remark}

The last result we state in this section is one of the main tools in \cite{KoLu} to prove large deviation principles for the Gibbs measures associated with elastic energies on periodic lattices. It is an interpolation inequality that allows to impose additional boundary conditions. We extend the validity of this inequality to admissible graphs. Although there are only minor changes in the argument, we display the proof with our notation in the appendix. Since it is a technical tool, we don't quantify the dependence on $\beta$ here. However, we stress that we have to keep track of how the estimate depends on the set $O$ after letting $\e\downarrow 0$ (see Remark~\ref{invariance} in the appendix).
\begin{proposition}\label{interpolation}
Assume Hypothesis~\ref{Hypo1} and let $G\in\mathcal{G}$. Fix $O\in\Ard$ and $\beta>0$. Let $v\in L^p_{\rm loc}(\R^d,\mathbb{R}^n)$. For $\delta>0$ we set $O^{\delta}=\{x\in O:\;\dist(x,\partial O)<2\delta\}$. Then for all $\delta>0$ small enough, $N\in\mathbb{N}$ and $\kappa>0$ there exists $\e_0>0$ and $C=C_{\beta}<+\infty$ such that for all $0<\e<\e_0$ and all $\varphi\in\mathcal{N}_p(v,O,\e,\kappa)$ we have
\begin{align*}
\Big(Z_{\e,O}^\beta(\mathcal{N}_p(v,O,\e,\kappa))\Big)^{\frac{N-C}{N}}\leq& \, 2NZ_{\e,O}^\beta(\mathcal{N}_p(v,O,\e,3\kappa)\cap\mathcal{B}_{\e}(O,\varphi))
\\
&\times \exp\Big(C\big(|(O^{\delta})_{\e}|+\Big(\frac{(N\kappa|O|^{\frac{1}{d}})^p}{\delta^p}+\frac{1}{N}\Big)|O_{\e}^\calL|+ H_{\e}(O^{\delta},\varphi)\big)\Big).
\end{align*}	
\end{proposition} 
%


\section{Thermodynamic limit of the free energy: Proof of Theorem~\ref{th:W}}\label{Sec3}

As made clear in the statements of the main result, linear boundary conditions 
are the basic ingredients to define the continuum free energy density. Let $\Lambda \in \R^{n\times d}$, and $D$ be a Lipschitz subset of $\R^d$. We denote by $\varphi_\Lambda:\calL \to\R^n$ the function defined by $\varphi_\Lambda(x)=\Lambda x$ and by $\overline{\varphi}_\Lambda :\R^d\to\R^n$ its continuum version $x\mapsto \Lambda x$
(this distinction will be needed when we identify $\varphi_\Lambda $ with its piecewise constant interpolation on the Voronoi tessellation; see below).
In this section we better characterize the asymptotic behavior of the functionals $\mathcal{E}^{\beta}_{\e}(D,\overline{\varphi}_{\Lambda})$. We first show that for stationary graphs there exists a limit of the free energy when $\e\to 0$ and, following the approach of \cite{KoLu}, we give some useful equivalent characterizations. Then we show that the limit inherits the $p$-growth conditions of Hypothesis~\ref{Hypo1}. Finally, we prove its quasiconvexity and conclude with Theorem~\ref{th:W}. 

\subsection{Existence of $\overline W^\beta$ and equivalent definitions}
We shall prove the almost sure existence of the limit  $\lim_{\e \downarrow 0} \mathcal{E}_{\e}(D,\overline{\varphi}_\Lambda)$ using the subadditive ergodic theorem, cf.~\cite[Theorem~2.7]{AkKr}. 
We set $\mathcal{I}=\{[a,b):a,b\in \mathbb{R}^d, a\neq b\}$, where $[a,b):=\{x\in\R^{d}:\;a_i\leq x_i<b_i\;\forall i\}$. 
\begin{proposition}\label{existence}
Assume Hypothesis~\ref{Hypo1}.
Fix $\Lambda \in\R^{n\times d}$. Then there exists a deterministic constant $\overline W^{\beta}(\Lambda )$
such that for all Lipschitz domains $D$ we have almost surely
\begin{equation*}
\overline W^\beta(\Lambda )=\lim_{\e\downarrow 0}\mathcal{E}_{\e}^\beta(D,\overline{\varphi}_\Lambda).
\end{equation*}
\end{proposition}
\begin{remark}\label{Ffixed}
In the above statement the exceptional set may depend on $\Lambda $ (and $\beta$). Later on we shall prove that $\overline W$ is continuous, which implies that the set can be taken independent of $\Lambda $ (and $\beta$).
\end{remark}
\begin{proof}[Proof of Proposition~\ref{existence}]
We drop the superscript $\beta$, and start with defining a suitable stochastic process (that is, a measurable function on the set of graphs $\mathcal G$).
Given $I\in\mathcal{I}$, set
\begin{equation}\label{defprocess}
\sigma(I):=-\log\Big(\int_{\mathcal{B}_1(I,\varphi_\Lambda )}\exp\big(- H_{1}(I,u)\big)\,\mathrm{d}u\Big)+C_{\Lambda}\,\mathcal{H}^{d-1}(\partial I),
\end{equation}
where $C_{\Lambda}$ will be chosen later to make the process subadditive.
%
In order to apply the subadditive ergodic Theorem~it is enough to prove:
\begin{itemize}
\item[(a)] that $|\sigma(I,G)|$ is bounded uniformly with respect to $G$,
\item[(b)] that $G\mapsto \sigma(I,G)$ is a stationary process,
\item[(c)] that $I\mapsto \sigma(I,G)$ is subadditive.
\end{itemize}
We split the rest of the proof into four steps, prove (a), (b), and (c) separately, and then conclude.

\step1{ Proof of (a)}
In order to show that $\sigma(I,\cdot)$ is integrable, we use Hypothesis~\ref{Hypo1}, \eqref{e.volVoronoi}, and Remark~\ref{reverse} in the form
\begin{align*}
H_1(I,u)&\leq C\|\nabla_\B u\|^p_{\ell^p_1(I)}+C|I|\leq C\|\nabla_\B (u-\varphi_\Lambda )\|^p_{\ell^p_1(I)}+C\|\nabla_\B \varphi_\Lambda \|^p_{\ell^p_1(I)}+C|I|
\\
&\leq C\|u-\varphi_\Lambda \|^p_{\ell^p_1(I)}+C(|\Lambda |^p+1)|I|.
\end{align*}
Since $\mathcal{B}_1(I,\varphi_\Lambda )-\varphi_\Lambda =\mathcal{B}_1(I,0)$, we obtain by a change of variables, monotonicity, Fubini's theorem, and~\eqref{e.volVoronoi} again,  
\begin{align}\label{upperestimate}
\sigma(I)\leq & \,C(|\Lambda |^p+1)|I|-\log\Big(\int_{\mathcal{B}_1(I,0)}\exp\big(- C\|u\|^p_{\ell^p_1(I)}\big)\,\mathrm{d}u\Big)+C_{\Lambda}\mathcal{H}^{d-1}(\partial I)\nonumber
\\
\leq &\,C(|\Lambda |^p+1)|I|-\#\{x\in\Lw\cap I:\;\dist(x,\partial I)>C_0\}\log\Big(\int_{\R^n}\exp(- C|\zeta|^p)\,\mathrm{d}\zeta\Big)\nonumber
\\
&-\#\{x\in\Lw\cap I:\;\dist(x,\partial I)\leq C_0\}\log\Big(\int_{B_1(0)}\exp(- C|\zeta|^p)\,\mathrm{d}\zeta\Big)+C_{\Lambda}\mathcal{H}^{d-1}(\partial I)\nonumber
\\
\leq &\,C(|\Lambda |^p+1)|I|+C|I|\big|\log C\big|+C\mathcal{H}^{d-1}(\partial I)\big|\log C\big|+C_{\Lambda}\mathcal{H}^{d-1}(\partial I).
\end{align}
From Hypothesis~\ref{Hypo1} and Lemma~\ref{forrestargument} applied with $V=\mathcal{B}_{1}(I,\varphi_\Lambda )$, $z_x=\Lambda x$, $\gamma=1$ and $\alpha=\frac{1}{C}$, we also deduce that
\begin{align*}
Z_{1,I,\varphi_\Lambda} \leq \int_{\mathcal{B}_1(I,\varphi_\Lambda )}\exp\big(- \frac{1}{C}\|\nabla_\B u\|^p_{\ell^p_1(I)}+C|I| \big)\,\mathrm{d}u
\leq \exp(C|I|)C^{|\Lw\cap I|}\leq C^{|I|}.
\end{align*} 
Taking minus the logarithm we obtain that
\begin{equation}\label{lowerestimate}
\sigma(I)\geq -|I|\log C .
\end{equation}
The desired estimate (a) follows from the combination of  \eqref{upperestimate} and  \eqref{lowerestimate}.

\step2{Proof of (b)}
Let $z\in \Z^d$. By definition of a graph and of $ \mathcal{B}_1$, we have the equivalence
\begin{equation*}
u \in \mathcal{B}_1(I,\varphi_\Lambda ,G+z)\iff w(\cdot)=u(\cdot+z)-\Lambda z\in\mathcal{B}_1(I-z,\varphi_\Lambda,G)
\end{equation*}
and $H_1(I,u,G+z)=H_1(I-z,w,G)$. This implies stationarity of $\sigma$ in the form  $\sigma(I,z+G)=\sigma(I-z,G)$.

\step3{Proof of (c)}
In order to prove subadditivity, let $I\in\mathcal{I}$ and consider a finite partition $I=\bigcup_iI_i$ with $I_i\in \mathcal{I}$. By definition,
\begin{equation}\label{tensorize}
\mathcal{B}_1(I,\varphi_\Lambda )\supset\{u:\Lw\cap I\to\R^n:\;u_{|\Lw\cap I_i}\in \mathcal{B}_1(I_i,\varphi_\Lambda )\}=\prod_i\mathcal{B}_1(I_i,\varphi_\Lambda ).
\end{equation}
Moreover, for any $u\in\prod_i\mathcal{B}_1(I_i,\varphi_\Lambda )$ the monotonicity of $H_1(\cdot,u)$ with respect to set inclusion yields the (almost) subadditivity estimate
\begin{align}\label{almostsub}
H_1(I,u)&\leq \sum_{i}H_1(I_i,u_{|\Lw\cap I_i})+\sum_{i}\sum_{\substack{(x,y)\in \B\\ [x,y]\cap \partial I_i\backslash\partial I\neq\emptyset}}C(|\Lambda |^p+1)\nonumber
\\
&\leq \sum_{i}H_1(I_i,u_{|\Lw\cap I_i})+C(|\Lambda |^p+1)\sum_{i}\mathcal{H}^{d-1}(\partial I_i\backslash\partial I).
\end{align}
From (\ref{tensorize}) and (\ref{almostsub}) we conclude by Fubini's Theorem~that for $C_{\Lambda}\geq C(|\Lambda |^p+1)$ we have
\begin{align*}
\sigma(I)&\leq \sum_{i}\big(\sigma(I_i)-C_{\Lambda}\mathcal{H}^{d-1}(\partial I_i)\big)+C(|\Lambda |^p+1)\sum_i\mathcal{H}^{d-1}(\partial I_i\backslash\partial I)+C_{\Lambda}\mathcal{H}^{d-1}(\partial I)
\\
&\leq\sum_i\sigma(I_i)+\big(C(|\Lambda |^p+1)-C_{\Lambda}\big) \sum_i\mathcal{H}^{d-1}(\partial I_i\backslash\partial I)\leq \sum_i\sigma(I_i),
\end{align*}
that is, the desired subadditivity.

\step4{Conclusion}
By the subadditive ergodic Theorem~(combined with an elementary approximation argument to pass from integer rectangles 
to general rectangles and Lipschitz domains, see for instance Step 4 of the proof of Theorem~3.1 in \cite{glpe}), we obtain the existence of the deterministic field $\overline W(\Lambda)$ satisfying
almost surely for all Lipschitz domains $D$
$$
\overline W(\Lambda)\,=\,\lim_{t\uparrow \infty} \frac{1}{t^d} \mathcal E_1(tD,\varphi_\Lambda).
$$
\end{proof}

Following~\cite{KoLu} we next prove two equivalent characterizations of $\overline W^\beta$.
\begin{lemma}\label{equivalent}
Assume Hypothesis~\ref{Hypo1}. Fix $\Lambda \in\R^{n\times d}$. Then $\overline W^\beta(\Lambda )$ defined in Proposition~\ref{existence}
 almost surely satisfies: For all $\kappa>0$ and all $O\in\mathcal{A}^R(\R^d)$,
\begin{align}
\overline W^\beta(\Lambda )&=\lim_{\e\downarrow 0}-\frac{1}{|O_{\e}|}\log\left( Z^\beta_{\e,O}(\mathcal{N}_p(\overline{\varphi}_\Lambda ,O,\e,\kappa)\cap\mathcal{B}_{\e}(O,\varphi_\Lambda ))\right)\tag{I}
\\
&=\mathcal{F}^+(O,\overline{\varphi}_\Lambda).\tag{II}
\end{align}
\end{lemma}
\begin{proof}[Proof of Lemma~\ref{equivalent}]
We split the proof into 4 steps. Again we drop the superscript $\beta$.

\step1{Existence of the limit}
We first prove that the right hand side of (I) is well-defined. Again we use the subadditive ergodic theorem. To this end, note that due to \eqref{e.volVoronoi} and the definition of the sets $\mathcal{N}_p$ in (\ref{defsets}) there exists a deterministic length of the form $l_{\kappa}=\lceil C_1(|\Lambda |+1)/\kappa\rceil\in\mathbb{N}$ such that, for any $I\in\mathcal{I}$,
\begin{equation}\label{lk_contained}
\varphi_\Lambda \in\mathcal{N}_p(\overline{\varphi}_\Lambda ,l_{\kappa}I,1,\frac{\kappa}{2})\cap\mathcal{B}_{1}(l_{\kappa}I,\varphi_\Lambda ).
\end{equation}
We fix such $C_1$ from now on and define the stochastic process $\sigma_{\kappa}:\mathcal{I}\to L^1(\mathcal{G})$ by
\begin{equation}\label{muk}
\sigma_{\kappa}(I)=-\log\left( Z_{1,I}(\mathcal{N}_p(\overline{\varphi}_\Lambda ,l_{\kappa}I,1,\kappa)\cap\mathcal{B}_{1}(l_{\kappa}I,\varphi_\Lambda ))\right)+C_{\sigma_{\kappa}}\mathcal{H}^{d-1}(\partial (l_{\kappa}I)),
\end{equation}
where $C_{\sigma_{\kappa}}$ will be chosen later to obtain subadditivity. 
To show integrability, we first note that $\sigma_{\kappa}(I)\geq \sigma(l_{\kappa}I)-(C_{\Lambda}-C_{\sigma_{\kappa}})l_{\kappa}\mathcal{H}^{d-1}(\partial I)$, where $\sigma$ is the process defined in the proof of Proposition~\ref{existence}. In order to prove an upper bound, observe that there exists a constant $c>0$ such that
\begin{equation*}
\varphi_\Lambda +\left\{u:l_{\kappa}I\cap\mathcal{L}\to\mathbb{R}^n:\;|u(x)|\leq \min\left\{1,c(|\Lambda |+1)|I|^{\frac{1}{d}}\right\}\right\}\subset\mathcal{N}_p(\overline{\varphi}_\Lambda ,l_{\kappa}I,1,\kappa)\cap\mathcal{B}_{1}(l_{\kappa}I,\varphi_\Lambda ).
\end{equation*} 
Indeed, the set on the left hand side clearly satisfies the boundary conditions and is thus contained in $\mathcal{B}_{1}(l_{\kappa}I,\varphi_\Lambda )$. The remaining inclusion follows by the triangle inequality since \eqref{e.volVoronoi} and (\ref{lk_contained}) yield
\begin{align*}
\|(\varphi_\Lambda )_1-\varphi_\Lambda -u\|_{\ell^p_1(l_{\kappa}I)}&< \frac{\kappa}{2}|l_{\kappa}I|^{\frac{1}{p}+\frac{1}{d}}+C|l_kI|^{\frac{1}{p}}\|u\|_{\infty}\leq \frac{\kappa}{2}|l_{\kappa}I|^{\frac{1}{p}+\frac{1}{d}}+C|l_{\kappa}I|^{\frac{1}{p}}c(|\Lambda |+1)|I|^{\frac{1}{d}}
\\
&\leq\frac{\kappa}{2}|l_{\kappa}I|^{\frac{1}{p}+\frac{1}{d}}+C|l_{\kappa}I|^{\frac{1}{p}}\frac{c}{C_1}\kappa|l_{\kappa}I|^{\frac{1}{d}}\leq\kappa|l_{\kappa}I|^{\frac{1}{p}+\frac{1}{d}},
\end{align*}
provided that $c\leq \frac{C_1}{2C}$. Having in mind the established set-inclusion, the argument for (\ref{upperestimate}) also yields a deterministic upper bound, the proof of which we omit.

Concerning stationarity, we recall that the interpolation in \eqref{discreteapprox} is random as it depends on the Voronoi cells. By stationarity of $G$, which is inherited by the Voronoi tessellation, for every $z\in\Z^d$ we have $(\varphi_\Lambda )_1(x+z,G+z)=(\varphi_\Lambda )_1(x,G)+\Lambda z$. Hence, with a slight abuse of notation,
\begin{equation*}
\left(\mathcal{N}_p(\overline{\varphi}_\Lambda ,l_{\kappa}(I-z),1,\kappa,G)\cap\mathcal{B}_{1}(l_{\kappa}(I-z),\varphi_\Lambda,G)\right)+l_{\kappa}\Lambda z=\mathcal{N}_p(\overline{\varphi}_\Lambda ,l_{\kappa}I,1,\kappa,G+l_{\kappa}z)\cap\mathcal{B}_{1}(l_{\kappa}I,\varphi_\Lambda ,G+l_{\kappa}z).
\end{equation*}
where we used that $l_{\kappa}\in\mathbb{N}$. By a change of variables we obtain the $l_{\kappa}\Z^d$-stationarity condition
\begin{equation*}
\sigma_{\kappa}(I-l_{\kappa}z,G)=\sigma_{\kappa}(I,G+l_{\kappa}z).
\end{equation*}
From now on the proof is very similar to the one of Proposition~\ref{existence}. Just note that for proving subadditivity, given a partition $I=\bigcup_iI_i$ it holds that
\begin{equation}\label{productdecomposition}
\mathcal{N}_p(\overline{\varphi}_\Lambda ,l_{\kappa}I,1,\kappa)\cap\mathcal{B}_{1}(l_{\kappa}I,\varphi_\Lambda )\supset \prod_i\mathcal{N}_p(\overline{\varphi}_\Lambda ,l_{\kappa}I_i,1,\kappa)\cap\mathcal{B}_{1}(l_{\kappa}I_i,\varphi_\Lambda ).
\end{equation}
This is clear for the boundary conditions, while for the discrete neighbourhoods it follows from the inequality $\sum_j|I_j|^{1+\frac{p}{d}}\leq |I|^{1+\frac{p}{d}}$, which is due to the fact that the discrete $\ell^p$-norms are maximal for $p=1$. With (\ref{productdecomposition}) at hand and choosing a suitable $C_{\sigma_{\kappa}}$, we conclude, as we did for Proposition~\ref{existence}, that with probability one the following limit exists
\begin{equation*}
W_{\kappa}(\Lambda )=\lim_{\e\to 0}-\frac{1}{|O_{\e}|}\log\left( Z_{\e,O}(\mathcal{N}_p(\overline{\varphi}_\Lambda ,O,\e,\kappa)\cap\mathcal{B}_{\e}(O,\varphi_\Lambda ))\right),
\end{equation*} 
and is independent of the Lipschitz domain $O$ (note that (\ref{discreteapprox}) and the definition of the sets $\mathcal{N}_p(\overline{\varphi}_\Lambda ,O,\e,\kappa)$ in (\ref{defsets}) are compatible with rescaling in the sense that $(\overline{\varphi}_\Lambda )_{\e}=\e(\overline{\varphi}_\Lambda )_1$ and therefore $u\in \mathcal{N}_p(\overline{\varphi}_\Lambda ,O,\e,\kappa)\cap\mathcal{B}_{\e}(O,\varphi_\Lambda )$ if and only if $u\in \mathcal{N}_p(\overline{\varphi}_\Lambda ,O/\e,1,\kappa)\cap\mathcal{B}_{1}(O/\e,\varphi_\Lambda )$).

\step2{Independence with respect to $\kappa$}
Let us prove that the limit is independent of $\kappa$ which turns out to be useful for proving (I) in the next step. We may assume that the limit of Step 1 exists almost surely for all positive $\kappa\in\mathbb{Q}$. As in \cite{KoLu} we compute the energy on a half-open cube $O=[0,1)^d$ which we subdivide again into $2^d$ smaller half-open cubes $\{O_i\}_{i=1}^{2^d}$ of equal size. In this case we can improve the rescaled version (\ref{productdecomposition}) in the sense that
\begin{equation*}
\mathcal{N}_p(\overline{\varphi}_\Lambda ,O,\e,\frac{\kappa}{2})\cap\mathcal{B}_{\e}(O,\varphi_\Lambda )\supset \prod_i\mathcal{N}_p(\overline{\varphi}_\Lambda ,O_i,\e,\kappa)\cap\mathcal{B}_{\e}(O_i,\varphi_\Lambda ).
\end{equation*}
Indeed, the boundary conditions on the large cube hold by the boundary conditions on the smaller cubes and, for each function $u$ belonging to the right hand side set, the discrete norms in the definition (\ref{defsets}) can be estimated via
\begin{equation*}
\sum_{x\in O_{\e}^\calL}\e^d|(\overline{\varphi}_\Lambda )_{\e}(x)-\e u(x)|^p\leq 2^d \kappa^p |O_1|^{1+\frac{p}{d}}= \left(\frac{\kappa}{2}\right)^p (2^d)^{1+\frac{p}{d}}|O_1|^{1+\frac{p}{d}}=\left(\frac{\kappa}{2}\right)^p|O|^{1+\frac{p}{d}}.
\end{equation*}
A rescaled version of the almost subadditivity estimate (\ref{almostsub}) and Fubini's Theorem~then yield
\begin{align*}
\sum_{i=1}^{2^d}&\left(C(|\Lambda |^p+1)\mathcal{H}^{d-1}(\partial O_i)\e-\frac{1}{|(O_i)_{\e}|}\log\left( Z_{\e,O_i}(\mathcal{N}_p(\overline{\varphi}_\Lambda ,O_i,\e,\kappa)\cap\mathcal{B}_{\e}(O_i,\varphi_\Lambda ))\right)\frac{|(O_i)_{\e}|}{|O_{\e}|}\right)
\\
&\geq-\frac{1}{|O_{\e}|}\log\left( Z_{\e,O}(\mathcal{N}_p(\overline{\varphi}_\Lambda ,O,\e,\frac{\kappa}{2})\cap\mathcal{B}_{\e}(O,\varphi_\Lambda ))\right).
\end{align*}
Passing to the limit when $\e\to 0$ we obtain by definition
\begin{equation*}
W_{\kappa}(\Lambda )\geq W_{\kappa/2}(\Lambda ),
\end{equation*}
where used that the limit is indeed also given by half-open cubes because the contributions at the boundary are negligible. Since the reverse inequality is obvious this proves that the two terms actually agree. In particular, by a sandwich argument we deduce that for a set of full probability the limit exists for all $\kappa>0$ and is independent of $\kappa$.

\step3{Proof of (I)}
We argue by letting $\kappa\to +\infty$. Since we shall use this equality to show that the exceptional set can be taken independent of $\Lambda $, we only use deterministic arguments.

On the one hand, by Hypothesis~\ref{Hypo1}, for any $u\in\mathcal{B}_{\e}(O,\varphi_\Lambda )$ and $\e$ small enough we have by the discrete Poincar\'e inequality of Lemma~\ref{poincare}
\begin{align*}
H_{\e}(u,O)\geq& \frac{1}{C}\|\nabla_\B u\|^p_{\ell^p_{\e}(O)}-C|O_{\e}^\calL|\geq \frac{1}{C}\|\nabla_\B (u-\varphi_\Lambda )\|^p_{\ell_\e^p(O)}-C(|\Lambda |^p+1)|O_{\e}^\calL|
\\
\geq&\frac{1}{C}\|\e u-\e\varphi_\Lambda \|^p_{\ell^p_{\e}(O)}-C(|\Lambda |^p+1)|O_{\e}^\calL|.
\end{align*}
On the other hand, if $u\notin \mathcal{N}_p(\overline{\varphi}_\Lambda ,O,\e,\kappa)$ then from the definition in (\ref{discreteapprox}) and \eqref{e.volVoronoi} we infer that
\begin{equation*}
\kappa^p|O|^{1+\frac{p}{d}}\leq C\sum_{x\in O_{\e}^\calL}\e^d\left(|\e u(x)-\e \Lambda x|^p+|\Lambda |^p\e^p\right)\leq C\e^d\|\e u-\e\varphi_\Lambda \|^p_{\ell^p_{\e}(O)}+C|\Lambda |^p\e^{p+d}|O_{\e}^\calL|.
\end{equation*}
Combining these inequalities, we infer that, given $M>0$ there exists $\kappa_0>0$ such that for all $\kappa>\kappa_0$ it holds that $\mathcal{B}_{\e}(O,\varphi_\Lambda )\backslash\mathcal{N}_p(\overline{\varphi}_\Lambda ,O,\e,\kappa)\subset \mathcal{B}_{\e}(O,\varphi_\Lambda )\backslash \mathcal{S}_M(O,\e)$. Now we choose
\begin{equation}\label{e.choiceM}
M>2\left(\overline{C}+\log(2)|O_{\e}^\calL|^{-1}+\frac{1}{|O_{\e}^\calL|}\log\Big(Z_{\e,O}(\mathcal{B}_{\e}(O,\varphi_\Lambda ))\Big)\right),
\end{equation}
where $\overline{C}$ denotes the constant from Lemma~\ref{tightness}. Note that due to (\ref{lowerestimate}), such $M$ can be chosen independent of $0<\e\leq\e_0$ for some $\e_0=\e_0(O)$ which depends not on the graph $G$. The second estimate of Lemma~\ref{tightness} yields for $\kappa$ large enough and this choice of $M$
\begin{align*}
Z_{\e,O}(\mathcal{B}_{\e}(O,\varphi_\Lambda )\backslash \mathcal{N}_p(\overline{\varphi}_\Lambda ,O,\e,\kappa))&\leq Z_{\e,O}(\mathcal{B}_{\e}(O,\varphi_\Lambda )\backslash \mathcal{S}_M(O,\e))
\\
&\leq\exp(-\frac{M}{2}|O_{\e}^\calL|)\exp(\overline{C}|O_{\e}^\calL|)\leq\frac{1}{2}Z_{\e,O}(\mathcal{B}_{\e}(O,\varphi_\Lambda ))
\end{align*}
or equivalently
\begin{equation*}
Z_{\e,O}(\mathcal{B}_{\e}(O,\varphi_\Lambda ))\geq Z_{\e,O}(\mathcal{B}_{\e}(O,\varphi_\Lambda )\cap \mathcal{N}_p(\overline{\varphi}_\Lambda ,O,\e,\kappa))\geq \frac{1}{2}Z_{\e,O}(\mathcal{B}_{\e}(O,\varphi_\Lambda )).
\end{equation*}
Taking logarithms and dividing by $-|O_{\e}|$ we obtain the claim letting $\e\to 0$ and using Step 2.

\step4{Proof of (II)}
We apply the interpolation inequality in Proposition~\ref{interpolation} with $\varphi=\overline{\varphi}_\Lambda $ and $\phi=\varphi_\Lambda $. Note that for $\e$ small enough it holds that $\varphi_\Lambda \in\mathcal{N}_p(\overline{\varphi}_\Lambda ,O,\e,\kappa)$. Taking logarithms in the interpolation inequality and dividing by $-|O_{\e}|$, we obtain by Steps 2 and 3, and Hypothesis~\ref{Hypo1}
\begin{equation*}
\frac{N-C}{N}\mathcal{F}_{\kappa}^+(O,\overline{\varphi}_\Lambda )\geq \overline W(\Lambda )
-C\left(\Big(\frac{(N\kappa|O|^{\frac{1}{d}})^p}{\delta^p}+\frac{1}{N}\Big)+(|\Lambda |^p+1)\frac{|O^{\delta}|}{|O|}\right).
\end{equation*}
The claim now follows taking first the limit $\kappa\to 0$ and then $N\to +\infty$ and $\delta\to 0$. On the other hand the reverse inequality follows by Step 3 and a monotonicity argument based on set inclusion.
\end{proof}

\begin{remark}\label{onexistence}
The estimates of Step 3 in the proof of Lemma~\ref{equivalent} show that the limit defining $\overline W^{\beta}(\Lambda )$ exists whenever $G$ is admissible and the limit defining $W_{\kappa}(\Lambda )$ exists.
\end{remark}

Using the observation of the previous remark, we now show that the exceptional set where convergence of the free energy may fail can be taken independent of the macroscopic boundary condition $\Lambda$. As a byproduct we obtain the continuity of the maps $\Lambda \mapsto \overline W^{\beta}(\Lambda )$.

\begin{proposition}\label{fullexistence}
Assume Hypothesis~\ref{Hypo1}. Then for almost all $G\in\mathcal{G}$, all $\Lambda \in\mathbb{R}^{n\times d}$, all $\beta>0$ and all bounded Lipschitz domains $D\subset\R^d$ there exists the limit 
\begin{equation*}
\overline W^{\beta}(\Lambda )=\lim_{\e\to 0}\mathcal{E}^{\beta}_{\e}(D,\overline{\varphi}_{\Lambda} )
\end{equation*}
and the further statements of Proposition~\ref{existence} and the equivalent characterizations of Lemma~\ref{equivalent} remain true. In addition the map $\Lambda \mapsto \overline W^{\beta}(\Lambda )$ is continuous.
\end{proposition}
\begin{proof}[Proof of Proposition~\ref{fullexistence}]
In view of Remark~\ref{onexistence}, the first claim proven once we show that there exists a common set $\mathcal{G}^{\prime}\subset\mathcal{G}$ of full probability such that the limit
\begin{equation}\label{kappalimit}
W_{\kappa}(\Lambda )=\lim_{\e\to 0}-\frac{1}{|O_{\e}|}\log\left( Z_{\e,O}(\mathcal{N}_p(\overline{\varphi}_\Lambda ,O,\e,\kappa)\cap\mathcal{B}_{\e}(O,\varphi_\Lambda ))\right)
\end{equation}
exists for all $G\in\mathcal{G}^{\prime}$, all $\kappa>0$, all $\Lambda \in\mathbb{R}^{n\times d}$ and all $O\in\mathcal{A}^R(\R^d)$. By Step 3 of the proof of Lemma~\ref{equivalent} we know that for any $\Lambda \in\mathbb{Q}^{n\times d}$ we find a set $\mathcal{G}_{\Lambda }\subset\mathcal{G}$ of full probability such that the limit in \eqref{kappalimit} exists for all $\kappa>0$, all $G\in\mathcal{G}_{\Lambda }$ and all $O\in\mathcal{A}^R(\R^d)$. Let us set $\mathcal{G}^{\prime}=\bigcap_{\Lambda \in\mathbb{Q}^{n\times d}}\mathcal{G}_{\Lambda }$. We fix $G\in\mathcal{G}^{\prime}$ and set
\begin{equation*}
\begin{split}
&\overline{W}_{\kappa}(\Lambda ,O)=\limsup_{\e\to 0}-\frac{1}{|O_{\e}|}\log\left( Z_{\e,O}(\mathcal{N}_p(\overline{\varphi}_\Lambda ,O,\e,\kappa)\cap\mathcal{B}_{\e}(O,\varphi_\Lambda ))\right),
\\
&\underline{W}_{\kappa}(\Lambda ,O)=\liminf_{\e\to 0}-\frac{1}{|O_{\e}|}\log\left( Z_{\e,O}(\mathcal{N}_p(\overline{\varphi}_\Lambda ,O,\e,\kappa)\cap\mathcal{B}_{\e}(O,\varphi_\Lambda ))\right).
\end{split}
\end{equation*}
We first argue by rational approximation that these two terms actually agree. To this end fix $\Lambda ,\Lambda ^{\prime}\in\mathbb{R}^{n\times d}$ and $O\in\mathcal{A}^R(\R^d)$. For $\delta>0$ we define the set $O^{\delta}=\{x\in O:\;\dist(x,\partial O)>\delta\}$. Taking $\delta$ small enough, we may assume that $O_{2\delta}\in\Ard$ (see, for instance, \cite[Lemma~2.2]{glpe}). We now define suitable interpolations. Let $\theta_{\delta}:O\to[0,1]$ be the Lipschitz-continuous cut-off function
\begin{equation*}
\theta_{\delta}(z)=\min\left\{\max\left\{\frac{1}{\delta}\dist(z,\partial O)-1,0\right\},1\right\}.
\end{equation*}
By construction we have $\theta_{\delta}\equiv 1$ on $\overline{O_{2\delta}}$, $\theta_{\delta}\equiv 0$ on $O\backslash O^{\delta}$ and the Lipschitz constant is bounded by ${\rm Lip}(\theta_{\delta})\leq \frac{1}{\delta}$.
Given $\varphi:O_{\e}^\calL\to\R^n$ and $\psi:(O\backslash O_{2\delta})_{\e}\to\R^n$ we define the interpolation $T_{\e,\delta}(\varphi,\psi):O_{\e}^\calL\to\R^n$ setting
\begin{equation*}
T_{\e,\delta}(\varphi,\psi)(x)=\theta_{\delta}(\e x)\varphi(x)+(1-\theta_{\delta}(\e x))\psi(x).
\end{equation*}
Assume in addition that $\varphi\in\mathcal{B}_{\e}(O_{2\delta},\varphi_\Lambda )\times \mathcal{N}_{\infty}(\varphi_\Lambda ,O\backslash O_{2\delta},\e)$ and $\psi\in \mathcal{N}_{\infty}(\varphi_{\Lambda ^{\prime}},O\backslash O_{2\delta},\e)$. Then Hypothesis~\ref{Hypo1} implies that
\begin{equation}\label{splithamiltonian}
H_{\e}(O,T_{\e,\delta}(\varphi,\psi))\leq H_{\e}(O_{2\delta},\varphi)+C\sum_{\substack{(x,y)\in \mathbb{B}\\ \e x\in O\backslash O_{2\delta},\e y\in O}}(|T_{\e,\delta}(\varphi,\psi)(x)-T_{\e,\delta}(\varphi,\psi)(y)|^p+1)
\end{equation}
To bound the last term we use the algebraic formula 
\begin{align*}
T_{\e,\delta}(\varphi,\psi)(x)-T_{\e,\delta}(\varphi,\psi)(y)=&\theta_{\delta}(\e y)(\varphi(x)-\varphi(y))+(1-\theta_{\delta}(\e y))(\psi(x)-\psi(y))
\\
&+(\theta_{\delta}(\e x)-\theta_{\delta}(\e y))(\varphi(x)-\psi(x)).
\end{align*}
For all $(x,y)\in E$ such that $\e x\in O\backslash O_{2\delta}$ and $\e y\in O$, the boundary values of $\varphi$ and the $L^{\infty}$-restrictions on $\varphi,\psi$ combined with the bound on ${\rm Lip}(\theta_{\delta})$ yield
\begin{equation*}
|T_{\e,\delta}(\varphi,\psi)(x)-T_{\e,\delta}(\varphi,\psi)(y)|\leq C(1+|\Lambda |+|\Lambda ^{\prime}|)+\frac{C\e}{\delta}|\Lambda -\Lambda ^{\prime}||x|\leq C(1+|\Lambda |+|\Lambda ^{\prime}|)+\frac{C}{\delta}|\Lambda -\Lambda ^{\prime}|,
\end{equation*}
where we used that $\e x\in O$ has equibounded norm for fixed $O$. Taking the $p^{th}$ power in this inequality, we can further estimate (\ref{splithamiltonian}) by
\begin{equation}\label{continuitybound}
H_{\e}(O,T_{\e,\delta}(\varphi,\psi))\leq H_{\e}(O_{2\delta},\varphi)+\Big(C(1+|\Lambda |+|\Lambda ^{\prime}|)^p+\frac{C}{\delta^p}|\Lambda -\Lambda ^{\prime}|^p\Big)|(O\backslash O_{2\delta})^{\calL}_{\e}|.
\end{equation}
To reduce notation, we set $\sigma_{\delta}(\Lambda ,\Lambda ^{\prime}):=(1+|\Lambda |+|\Lambda ^{\prime}|)^p+\frac{1}{\delta^p}|\Lambda -\Lambda ^{\prime}|^p$ and define the set
\begin{equation*}
\mathcal{S}=\Big(\mathcal{N}_p(\overline{\varphi}_\Lambda ,O_{2\delta},\e,\kappa)\cap\mathcal{B}_{\e}(O_{2\delta},\varphi_\Lambda )\Big)\times \mathcal{N}_{\infty}(\varphi_\Lambda ,O\backslash O_{2\delta},\e).
\end{equation*}
Using (\ref{continuitybound}) and the fact that $|\mathcal{N}_{\infty}(\varphi_\Lambda ,O\backslash O_{2\delta},\e)|\times |\mathcal{N}_{\infty}(\varphi_{\Lambda ^{\prime}},O\backslash O_{2\delta},\e)|\geq \exp(-C|(O\backslash O_{2\delta})^{\calL}_{\e}|)$, we deduce from Fubini's Theorem~that
\begin{align}\label{switchbound1}
Z_{\e,O_{2\delta}}&(\mathcal{N}_p(\overline{\varphi}_\Lambda ,O_{2\delta},\e,\kappa)\cap\mathcal{B}_{\e}(O_{2\delta},\varphi_\Lambda ))
\exp\big(-C\sigma_{\delta}(\Lambda ,\Lambda ^{\prime})|(O\backslash O_{2\delta})^{\calL}_{\e}|\big)\nonumber
\\
&\leq\int_{\mathcal{S}\times\mathcal{N}_{\infty}(\varphi_{\Lambda ^{\prime}},O\backslash O_{2\delta},\e)}\exp(-H_{\e}(O,T_{\e,\delta}(\varphi,\psi))\,\mathrm{d}\varphi\,\mathrm{d}\psi.
\end{align}
In order to provide an upper bound for the integral on the right hand side, we change the variables. A computation yields for any $(\varphi,\psi)\in\mathcal{S}\times \mathcal{N}_{\infty}(\varphi_{\Lambda ^{\prime}},O\backslash O_{2\delta},\e)$ the estimate
\begin{equation*}
\e^{\frac{d}{p}}\|(\overline{\varphi}_{\Lambda ^{\prime}})_{\e}-\e T_{\e,\delta}(\varphi,\psi)\|_{\ell^p_{\e}(O)}\leq
\kappa |O_{2\delta}|^{\frac{1}{d}+\frac{1}{p}}+C|\Lambda -\Lambda ^{\prime}||O|^{\frac{1}{p}}+C\e(1+|\Lambda |+|\Lambda ^{\prime}|)|(O\backslash O_{2\delta})_{\e}|^{\frac{1}{p}}\e^{\frac{d}{p}}.
\end{equation*}
Hence we find $\eta_{\kappa}=\eta_{\kappa}(O)>0$ such that for all $\e$ small enough we have the implication
\begin{equation}\label{kappainclusion}
|\Lambda -\Lambda ^{\prime}|<\eta_{\kappa}\quad\quad\Rightarrow\quad\quad T_{\e,\delta}(\varphi,\psi)\in\mathcal{N}_p(\overline{\varphi}_{\Lambda ^{\prime}},O,\e,3\kappa)\cap \mathcal{B}_{\e}(O,\varphi_{\Lambda ^{\prime}})
\end{equation}
for all $(\varphi,\psi)\in\mathcal{S}\times \mathcal{N}_{\infty}(\varphi_{\Lambda ^{\prime}},O\backslash O_{2\delta},\e)$. In particular this implication is independent of $\delta$. Introducing the function $b:(O\backslash O_{2\delta})^{\calL}_{\e}\to\R^n$ defined by
\begin{equation*}
b(x)=
\begin{cases}
\Lambda ^{\prime}x &\mbox{if $\theta_{\delta}(\e x)\geq\frac{1}{2}$,}
\\
\Lambda x &\mbox{if $\theta_{\delta}(\e x)<\frac{1}{2}$,}
\end{cases}
\end{equation*}
for $\e$ small enough and $|\Lambda -\Lambda ^{\prime}|<\eta_{\kappa}$ we can define the linear mapping
$\Phi_{\e,\delta}:\mathcal{S}\times \mathcal{N}_{\infty}(\varphi_{\Lambda ^{\prime}},O\backslash O_{2\delta},\e)\to \mathcal{N}_p(\overline{\varphi}_{\Lambda ^{\prime}},O,\e,3\kappa)\times\mathcal{N}_{\infty}(b,O\backslash O_{2\delta},\e)$ by setting (with a slight abuse of notation)
\begin{equation*}
\Phi_{\e,\delta}(\varphi,\psi)(x)=
\begin{cases}
(T_{\e,\delta}(\varphi,\psi)(x),\psi(x)) &\mbox{if $\theta_{\delta}(\e x)\geq \frac{1}{2}$,}\\
(T_{\e,\delta}(\varphi,\psi)(x),\varphi(x)) &\mbox{if $\theta_{\delta}(\e x)<\frac{1}{2}$.}
\end{cases}
\end{equation*}
Note that $\Phi_{\e,\delta}$ is well-defined due to (\ref{kappainclusion}) and bijective onto its range $\mathcal{R}(\Phi_{\e,\delta})$. In order to calculate the Jacobian, it is convenient to number the points $x\in O_{\e}^\calL$ and view the state space as large vectors by putting as first component the value $\varphi(x(1))$ and as second either $\psi(x(1))$ if $x(1)\in O\backslash O_{2\delta}$ or $\varphi(x(2))$ otherwise. Continuing this procedure the matrix representation of $\Phi_{\e,\delta}$ has non-zero entries only in $2\times 2$-matrices around the diagonal. Thus the determinant splits into products and we obtain
\begin{equation*}
|\det(D\Phi_{\e,\delta}(\varphi,\psi))|^{-1}= \left(\prod_{x:\theta_{\delta}(\e x)\geq \frac{1}{2}}|\theta_{\delta}(\e x)|^n\prod_{x:\theta_{\delta}(\e x)<\frac{1}{2}}|1-\theta_{\delta}(\e x)|^n\right)^{-1}\leq \exp(C|(O\backslash O_{2\delta})^{\calL}_{\e}|).
\end{equation*}
Via the change of variables $(g,h)=\Phi_{\e,\delta}(\varphi,\psi)$ and (\ref{kappainclusion}) we can estimate the right hand side integral in \eqref{switchbound1} by
\begin{align}\label{switchbound2}
\int_{\mathcal{S}\times\mathcal{N}_{\infty}(\varphi_{\Lambda ^{\prime}},O\backslash O_{2\delta},\e)}\exp(-H_{\e}(O,T_{\e,\delta}(\varphi,\psi))\,\mathrm{d}\varphi\,\mathrm{d}\psi\leq& Z_{\e,O}(\mathcal{N}_p(\overline{\varphi}_{\Lambda ^{\prime}},O,\e,3\kappa)\cap \mathcal{B}_{\e}(O,\varphi_{\Lambda ^{\prime}}))\nonumber
\\
&\times|\mathcal{N}_{\infty}(b,O\backslash O_{2\delta},\e)|\exp(C|(O\backslash O_{2\delta})^{\calL}_{\e}|).
\end{align}
Combining (\ref{switchbound1}) and (\ref{switchbound2}) we conclude the estimate
\begin{align*}
Z_{\e,O_{2\delta}}(\mathcal{N}_p(\overline{\varphi}_\Lambda ,O_{2\delta},\e,\kappa)\cap\mathcal{B}_{\e}(O_{2\delta},\varphi_\Lambda ))\leq& Z_{\e,O}(\mathcal{N}_p(\Lambda ^{\prime}x,O,\e,3\kappa)\cap \mathcal{B}_{\e}(O,\varphi_{\Lambda ^{\prime}}))
\\
&\times
\exp\big(C\sigma_{\delta}(\Lambda ,\Lambda ^{\prime})|(O\backslash O_{2\delta})^{\calL}_{\e}|\big).
\end{align*}
Taking logarithms, dividing by $-|O_{\e}|$ and taking the limes superior on both sides yields
\begin{equation}\label{almostswitched1}
\frac{|O_{2\delta}|}{|O|}\overline{W}_{\kappa}(\Lambda ,O_{2\delta})\geq \overline{W}_{3\kappa}(\Lambda ^{\prime},O)-C\sigma_{\delta}(\Lambda ,\Lambda ^{\prime})\frac{|O\backslash O_{2\delta}|}{|O|}.
\end{equation}
Replacing $O_{2\delta}$ and $O$ by the sets $O$ and $O^{2\delta}$ respectively, where $O^{2\delta}=\{x\in\R^d:\dist(x,O)<2\delta\}$, and switching the roles of $\Lambda $ and $\Lambda ^{\prime}$ we can prove in exactly the same way the estimate
\begin{equation}\label{almostswitched2}
\underline{W}_{\,3\kappa}(\Lambda ^{\prime},O)\geq \frac{|O^{2\delta}|}{|O|}\underline{W}_{\,9\kappa}(\Lambda ,O^{2\delta})-C\sigma_{\delta}(\Lambda ,\Lambda ^{\prime})\frac{|O^{2\delta}\backslash O|}{|O|}.
\end{equation} 
Further we may assume that $O^{2\delta}\in\mathcal{A}^R(\R^d)$. Choosing $\Lambda =\Lambda _j\in\mathbb{Q}^{n\times d}$ such that $\Lambda _j\to \Lambda ^{\prime}$, the two inequalities (\ref{almostswitched1}) and (\ref{almostswitched2}) yield 
\begin{equation*}
0\leq \overline{W}_{3\kappa}(\Lambda ^{\prime},O)-\underline{W}_{3\kappa}(\Lambda ^{\prime},O)\leq \frac{|O^{2\delta}\backslash O_{2\delta}|}{|O|}\Big(|\overline W(\Lambda _j)|+C\sigma_{\delta}(\Lambda _j,\Lambda ^{\prime})\Big).
\end{equation*}
Letting first $j\to +\infty$ and then $\delta\to 0$ the right hand side vanishes since $\overline W(\Lambda )$ is locally bounded (see the estimates \eqref{upperestimate} and \eqref{lowerestimate}). Hence the limit in (\ref{kappalimit}) indeed exists. 

We now show that it is independent of $O$ and $\kappa$. Choosing $\Lambda _j$ as above, we infer again from (\ref{almostswitched1}) that 
\begin{equation*}
W_{3\kappa}(\Lambda ^{\prime},O)\leq \liminf_j \overline W(\Lambda _j)+\Big(1-\frac{|O_{2\delta}|}{|O|}\Big)\sup_j |\overline W(\Lambda _j)|+C\frac{|O\backslash O_{2\delta}|}{|O|}(1+2|\Lambda ^{\prime}|^p).
\end{equation*}
Letting $\delta\to 0$ we obtain $W_{3\kappa}(\Lambda ^{\prime},O)\leq\liminf_j \overline W(\Lambda _j)$. On the other hand, (\ref{almostswitched2}) and a similar reasoning yield $W_{3\kappa}(\Lambda ^{\prime},O)\geq \limsup_j \overline W(\Lambda _j)$, so that
\begin{equation*}
W_{3\kappa}(\Lambda ^{\prime},O)=\lim_j \overline W(\Lambda _j)
\end{equation*}
is independent of $\kappa$ and $O$. Repeating the deterministic argument from Step 3 of the proof of Lemma~\ref{equivalent} one can show that $W_{3\kappa}(\Lambda ^{\prime},O)=\overline W(\Lambda ^{\prime})$ for all $\kappa>0$. Thus continuity can be proven using again (\ref{almostswitched1}) and (\ref{almostswitched2}) since there is no $\kappa$-dependence any more.

Finally, as the only random construction in the proof of Lemma~\ref{equivalent} was the existence of the limits in \eqref{kappalimit} it is clear that the characterizations still hold true and, by continuity, so do the additional properties stated in Proposition~\ref{existence}.
\end{proof}

\begin{remark}\label{alltemp}
As we intend to vary the temperature in Section \ref{Sec5}, let us observe that the set $\mathcal{G}^{\prime}\subset\mathcal{G}$ of full probability given by Proposition~\ref{fullexistence} can be chosen also independent of $\beta>0$. Indeed, first we choose a set of full probability such that Proposition~\ref{fullexistence} holds for all rational $\beta\geq 1$. Then for given $\beta\geq 1$ we take a rational sequence $\beta_j>\beta$ such that $\beta_j\downarrow\beta$. The remaining argument relies on Remark~\ref{onexistence}: Fix $\kappa>0$ and a set $O\in\mathcal{A}^R(\R^d)$. Then by monotonicity and Lemma~\ref{tightness} with a suitable $M=M(\beta,\Lambda)$ (see e.g. \eqref{e.choiceM}), we have for all $\e$ small enough the inequality
\begin{align*}
Z^{\beta_j}_{\e,O}(\mathcal{N}_p(\overline{\varphi}_\Lambda ,O,\e,\kappa)\cap\mathcal{B}_{\e}(O,\varphi_\Lambda ))&\leq Z^{\beta}_{\e,O}(\mathcal{N}_p(\overline{\varphi}_\Lambda ,O,\e,\kappa)\cap\mathcal{B}_{\e}(O,\varphi_\Lambda ))
\\
&\leq Z^{\beta}_{\e,O}(\mathcal{B}_{\e}(O,\varphi_\Lambda ))\leq 2Z^{\beta}_{\e,O}(\mathcal{B}_{\e}(O,\varphi_\Lambda )\cap\mathcal{S}_M(O,\e))
\\
&\leq 2Z^{\beta_j}_{\e,Q}(\mathcal{B}_{\e}(O,\varphi_\Lambda )\cap\mathcal{S}_M(O,\e))\times \exp(M(\beta_j-\beta)|O_{\e}^\calL|)
\\
&\leq 2Z^{\beta_j}_{\e,O}(\mathcal{B}_{\e}(O,\varphi_\Lambda ))\times \exp(M(\beta_j-\beta)|O_{\e}^\calL|).
\end{align*}
Taking the logarithm and dividing $-\beta|O_{\e}|$, we obtain by Lemma~\ref{equivalent} and Proposition~\ref{fullexistence} that the limit corresponding to $\beta$ exists and is independent of $\kappa$ and $O$. Moreover, as a by-product, we proved a continuous dependence on $\beta$.	
\end{remark}
\subsection{$p$-growth from above and below}
For the limit free energy $\overline W(\Lambda )$ we now prove suitable two-sided growth estimates. Here we keep track of the dependence on the inverse temperature $\beta$.
\begin{lemma}\label{ub}
Assume Hypothesis~\ref{Hypo1}. Let $\overline{W}^{\beta}$ be given by Proposition~\ref{existence}. Then there exists a constant $C>0$ such that for all $\Lambda\in\R^{n\times d}$ and all $\beta>0$
\begin{equation*}
\overline W^{\beta}(\Lambda )\leq C|\Lambda |^p+C\left(1+\frac{1+|\log(\beta)|}{\beta}\right).
\end{equation*}	
\end{lemma}
\begin{proof} [Proof of Lemma~\ref{ub}]
This estimate as an immediate consequence of the bound (\ref{upperestimate}) taking into account that there is a prefactor $\beta$ in the exponential functions.
\end{proof}

We now turn to the lower bound. Here we use the full assumptions on the graph in Definition \ref{defadmissible}.
\begin{lemma}\label{lb}
Assume Hypothesis~\ref{Hypo1} and let $G\in\mathcal{G}$. Let $v\in L^p_{\rm loc}(\R^d,\R^n)$. Then $\mathcal{F}^-(O,v)<+\infty$ only if $v\in W^{1,p}(O,\R^n)$. In this case there exists a constant $c>0$ such that
\begin{equation*}
\mathcal{F}^-(O,v)\geq \frac{c}{|O|}\int_O|\nabla v(z)|^p\,\mathrm{d}z-\frac{1}{c}\left(1+\frac{1+\log(\beta)|}{\beta}\right).
\end{equation*}	
In particular 
\begin{equation*}
\overline W^{\beta}(\Lambda )\geq c|\Lambda |^p-\frac{1}{c}\left(1+\frac{1+|\log(\beta)|}{\beta}\right).
\end{equation*}
\end{lemma}
\begin{proof}[Proof of Lemma~\ref{lb}]
First observe that the lower bound on $\overline W(\Lambda )$ follows by the first estimate and Lemma~\ref{equivalent}. To prove the first estimate, we split the free energy in a purely variational part and an integral part over a translated neighborhood. Given $\e,\kappa$ fixed, we first choose $u_{\e,\kappa}\in \mathcal{N}_p(v,O,\e,\kappa)$ such that
\begin{equation*}
\|\nabla_{\B}u_{\e,\kappa}\|^p_{\ell^p_{\e}(O)}\leq \inf_{u\in\mathcal{N}_p(v,O,\e,\kappa)}\|\nabla_{\B}u\|^p_{\ell^p_{\e}(O)}+1.
\end{equation*}
By convexity, for every $u\in\mathcal{N}_p(\varphi,O,\e,\kappa)$ we have $\frac{1}{2}u_{\e,\kappa}+\frac{1}{2}u\in \mathcal{N}_p(v,O,\e,\kappa)$. Hence by Lemma~\ref{pgrowth} there exists $C_p<2^p$ such that
\begin{align*}
\|\nabla_{\B}u_{\e,\kappa}\|^p_{\ell^p_{\e}(O)}-1&\leq \|\frac{1}{2}\nabla_{\B}u_{\e,\kappa}+\frac{1}{2}\nabla_{\B}u\|^p_{\ell^p_{\e}(O)}\\
&\leq \frac{C_p}{2^p}\|\nabla_{\B}u_{\e,\kappa}\|^p_{\ell^p_{\e}(O)}+\frac{C_p}{2^p}\|\nabla_{\B}u\|^p_{\ell_\e^p(O)}-\frac{1}{2^p}\|\nabla_{\B}(u_{\e,\kappa}-u)\|^p_{^p_{\e}(O)}.
\end{align*}
Subtracting the first and the last term on the right hand side, we infer that
\begin{equation*}
(1-\frac{C_p}{2^p})\|\nabla_\B u_{\e,\kappa}\|^p_{\ell^p_{\e}(O)}+\frac{1}{2^p}\|\nabla_\B (u_{\e,\kappa}-u)\|^p_{\ell^p_{\e}(O)}-1\leq \frac{C_p}{2^p}\|\nabla_\B u\|^p_{\ell^p_{\e}(O)}.
\end{equation*}
As $\frac{C_p}{2^p}<1$, this estimate combined with Hypothesis~\ref{Hypo1} yields
\begin{equation*}
H_{\e}(O,u)\geq \frac{1}{C}\|\nabla_{\B}u\|^p_{\ell^p_{\e}(O)}-C|O_{\e}^\calL|
\geq \frac{1}{C}\|\nabla_{\B}u_{\e,\kappa}\|^p_{\ell^p_{\e}(O)}-C|O_{\e}^\calL|+\frac{1}{C}\|\nabla_{\B}(u_{\e,\kappa}-u)\|^p_{\ell^p_{\e}(O)}.
\end{equation*}
As the function $u_{\e,\kappa}-u$ belongs to $\mathcal{N}_p(0,O,\e,2\kappa)$, after a change of variables we obtain
\begin{align}\label{splitlb}
-\frac{1}{\beta|O_{\e}|}\log(Z^{\beta}_{O,\e}(\mathcal{N}_p(v,O,\e,\kappa)))\geq& \frac{1}{C|O_{\e}|}\|\nabla_{\B}u_{\e,\kappa}\|^p_{\ell^p_{\e}(O)}-C\nonumber\\
&-\frac{1}{\beta|O_{\e}|}\log\bigg(\int_{\mathcal{N}_p(0,O,\e,2\kappa)}\exp\Big(-\frac{\beta}{C}\|\nabla_{\B}u\|^p_{\ell^p_{\e}(O)}\Big)\,\mathrm{d}u\bigg)
\end{align} 

\step 1{Estimate of the first right hand side term of \eqref{splitlb}}
Let us start with estimating $|O_{\e}|^{-1}\|\nabla_\B u_{\e,\kappa}\|^p_{\ell^p_{\e}(O)}$. Here we follow \cite{ACG2} and use a difference quotient estimate. To this end, let $O^{\prime}\subset\subset O$ and fix $h\in\R^d$ with $2|h|\leq \dist(O^{\prime},\partial O)$. For any $y\in \Lw$ we set $U^h_{\e}(y)=\{x\in\Lw:\;\mathcal{C}(x)\cap (\mathcal{C}(y)+\frac{h}{\e})\neq\emptyset\}$. Then
\begin{align}\label{intest}
\int_{O^{\prime}/\e}&|\e u_{\e,\kappa}(z+h/\e)-\e u_{\e,\kappa}(z)|^p\,\mathrm{d}z\leq\sum_{\substack{y\in\Lw\\ \mathcal{C}(y)\cap \frac{O^{\prime}}{\e}\neq\emptyset}}\int_{ \mathcal{C}(y)}\e^p|u_{\e,\kappa}(z+h/\e)-u_{\e,\kappa}( y)|^p\,\mathrm{d}z\nonumber\\
&=\sum_{\substack{y\in\Lw\\ \mathcal{C}(y)\cap \frac{O^{\prime}}{\e}\neq\emptyset}}\int_{ \mathcal{C}(y)+\frac{h}{\e}}\e^p|u_{\e,\kappa}(z)-u_{\e,\kappa}(y)|^p\,\mathrm{d}z\leq \sum_{\substack{y\in\Lw\\ \mathcal{C}(y)\cap \frac{O^{\prime}}{\e}\neq\emptyset}}\sum_{x\in U^h_{\e}(y)}\int_{\mathcal{C}(x)}\e^p|u_{\e,\kappa}(x)-u_{\e,\kappa}(y)|^p\,\mathrm{d}z.
\end{align}
We next derive a pointwise estimate of $\e^p|u_{\e,\kappa}(x)-u_{\e,\kappa}(y)|^p$ for $x\in U^h_{\e}(y)$. Let $P(x,y)$ be a path connecting $x,y$ satisfying the properties of Definition \ref{defadmissible} (iv). From \eqref{e.volVoronoi} we deduce that $|x-y|\leq |h|\e^{-1}+2R$ and thus $\#P(x,y)\leq C(|h|\e^{-1}+1)$. Using Jensen's inequality we obtain
\begin{align}\label{diffbound}
\e^p|u_{\e,\kappa}(x)-u_{\e,\kappa}(y)|^p&\leq \e^p\left(\#P(x,y)\right)^{p-1}\sum_{(x^{\prime},x^{\prime\prime})\in P(x,y)}|u_{\e,\kappa}(x^{\prime})-u_{\e,\kappa}(x^{\prime\prime})|^p\nonumber\\
&\leq (C\e|h|^{p-1}+C\e^p)\sum_{(x^{\prime},x^{\prime\prime})\in P(x,y)}|u_{\e,\kappa}(x^{\prime})-u_{\e,\kappa}(x^{\prime\prime})|^p.
\end{align}
Moreover note that $\mathcal{C}(y)\cap O^{\prime}_{\e}\neq\emptyset$ and $x\in U^h_{\e}(y)$ imply that $x^{\prime},x^{\prime\prime}\in O_{\e}^\calL$ for $\e$ small enough. Indeed, applying the triangle inequality several times one can show that for any $v\in[ x^{\prime},x^{\prime\prime}]$ one has
\begin{equation*}
\dist\left(v,\frac{O^{\prime}}{\e}\right)\leq \frac{|h|}{\e}+(2C_0+3R),
\end{equation*} where $C_0$ is given by Definition \ref{defadmissible}. Conversely, given any $(x^{\prime},x^{\prime\prime})\in \mathbb{B}$, we define the sets
\begin{equation*}
K_{\e}^h(x^{\prime},x^{\prime\prime}):=\{y\in\Lw:\;\exists x\in U^h_{\e}(y)\text{ such that }(x^{\prime},x^{\prime\prime})\in P(x,y)\}.
\end{equation*}
As $G$ is admissible, for any $(x^{\prime},x^{\prime\prime})\in P(x,y)$ there exists $\lambda\in[0,1]$ such that $z=y+\lambda (x-y)$ satisfies $|z-x^{\prime}|\leq C_0$. Hence for any $y\in K^h_{\e}(x^{\prime},x^{\prime\prime})$ we obtain
\begin{equation*}
y=y-z+x^{\prime}+(z-x^{\prime})=-\lambda \frac{h}{\e}+x^{\prime}+\lambda \left(\frac{h}{\e}-(x-y)\right)+(z-x^{\prime})\in [-\frac{h}{\e},0]+x^{\prime}+B_{2R+C_0}(0),
\end{equation*}
where we have used that $|x-y-\frac{h}{\e}|\leq 2R$ for any $x\in U^h_{\e}(y)$. Using again \eqref{e.volVoronoi} we conclude that $\#K_{\e}^h(x^{\prime},x^{\prime\prime})\leq C(|h|\e^{-1}+1)$. Furthermore the set $U_{\e}^h(y)$ has equibounded cardinality, so that the inequalities \eqref{intest}, \eqref{diffbound} and the uniform bound on the measure of the Voronoi cells imply
\begin{align}\label{int-energy}
\int_{O^{\prime}/\e}|\e u_{\e,\kappa}(z+h/\e)-\e u_{\e,\kappa}(z)|^p\,\mathrm{d}z&\leq C(|h|^p+|h|\e^{p-1}+\e|h|^{p-1}+\e^p)\sum_{\substack{(x,y)\in \mathbb{B}\\ x,y\in O_{\e}^\calL}}|u_{\e,\kappa}(x)-u_{\e,\kappa}(y)|^p\nonumber
\\
&=C(|h|^p+|h|\e^{p-1}+\e|h|^{p-1}+\e^{p})\|\nabla_{\B}u_{\e,\kappa}\|^p_{\ell^p_{\e}(O)}.
\end{align}
As $u_{\e,\kappa}\in\mathcal{N}_p(v,O,\e,\kappa)$, the function $v_{\e,\kappa}:O\to\mathbb{R}^n$ defined by $v_{\e,\kappa}(x):=\e u_{\e,\kappa}(x/\e)$ satisfies 
\begin{equation*}
\int_{O^{\prime}}|v_{\e,\kappa}(z)-v_{\e}(z/\e)|^p\,\mathrm{d}z\leq C\e^{d}\|\e u_{\e,\kappa}-v_{\e}\|^p_{\ell_\e^p(O)}\leq C\kappa^p|O|^{1+\frac{p}{d}}.
\end{equation*}
In particular, by Remark~\ref{approximation} it is bounded in $L^p(O^{\prime})$ and thus there exists a subsequence (not relabeled) such that $v_{\e,\kappa}\rightharpoonup v_{\kappa}$ in $L^p(O^{\prime})$. Moreover, by Remark~\ref{approximation} and lower semicontinuity of the $L^p$-norm it holds that $\|v_{\kappa}-v\|_{L^p(O^{\prime})}\leq C\kappa|O|^{\frac{1}{p}+\frac{1}{d}}$. By a change of variables in the left hand side of \eqref{int-energy} we further obtain that
\begin{equation}\label{frechetest}
\int_{O^{\prime}}|v_{\e,\kappa}(z+h)-v_{\e,\kappa}(z)|^p\,\mathrm{d}z\leq C(|h|^p\e^d+|h|\e^{p+d-1}+\e^{d+1}|h|^{p-1}+\e^{p+d})\|\nabla_\B u_{\e,\kappa}\|^p_{\ell^p_{\e}(O)}.
\end{equation}
Applying weak lower semicontinuity in the above estimate we deduce
\begin{equation}\label{sobolevestimate}
\liminf_{\kappa\to 0}\liminf_{\e\to 0}\frac{1}{|O_{\e}|}\|\nabla_{\B}u_{\e,\kappa}\|^p_{\ell^p_{\e}(A)}\geq \frac{1}{C|O|}\int_{O^{\prime}}\left|\frac{v(z+h)-v(z)}{|h|}\right|^p\,\mathrm{d}z.
\end{equation}
Before we conclude Sobolev-regularity of $v$, we have to ensure that the third right hand side term in \eqref{splitlb} remains finite.

\step 2{Control of the third right hand side term of \eqref{splitlb}}
We want to apply Lemma~\ref{forrestargument}. To this end, we observe that
\begin{equation*}
|u(x)|\leq 2\kappa\left(\frac{|O|}{\e^d}\right)^{\frac{1}{p}+\frac{1}{d}}
\end{equation*}
for any $u\in\mathcal{N}_p(0,O,\e,2\kappa)$ and all $x\in O_{\e}^\calL$. Hence, setting $\gamma=2\kappa\left(\frac{|O|}{\e^d}\right)^{\frac{1}{p}+\frac{1}{d}}$, $z_x=0$ and $\alpha=\frac{\beta}{C}$, Lemma~\ref{forrestargument} yields
\begin{align}\label{entropybound}
\log\bigg(\int_{\mathcal{N}_p(0,O,\e,2\kappa)}\exp(-\frac{\beta}{C}\|\nabla_{\B}u\|^p_{\ell^p_{\e}(O)})\,\mathrm{d}u\bigg)\leq &N_{O,\e}\log\bigg(C(2\kappa)^n\left(\frac{|O|}{\e^d}\right)^{\frac{n}{p}+\frac{n}{d}}\bigg)\nonumber
\\
&+C\big(|O_{\e}^\calL|-N_{O,\e}\big)(1+|\log(\beta)|),
\end{align}
where $N_{O,\e}$ denotes the number of connected components of the graph $G_{O,\e}$. Since $O$ has Lipschitz boundary, by Remark~\ref{numbercomponents} and \eqref{surfaceesimtate} it holds that 
\begin{equation*}
N_{O,\e}\leq C\e^{1-d}\mathcal{H}^{d-1}(\partial O)
\end{equation*} 
for $\e$ small enough. Dividing \eqref{entropybound} by $-\beta |O_{\e}|$ and letting $\e\to 0$ we find that
\begin{align*}
\liminf_{\e\to 0}-\frac{1}{\beta|O_{\e}|}\log\Big(\int_{\mathcal{N}_p(0,O,\e,2\kappa)}\exp\Big(-\frac{1}{C}\|\nabla_{\B}u\|^p_{\ell^p_{\e}(O)}\Big)\mathrm{d}u\Big)
\geq-C\left(\frac{1+|\log(\beta)\big|}{\beta}\right).
\end{align*}
From \eqref{splitlb}, \eqref{sobolevestimate} and the previous inequality we finally obtain
\begin{equation*}
\mathcal{F}^-(O,v)\geq \frac{1}{C|O|}\int_{O^{\prime}}\left|\frac{v(z+h)-v(z)}{|h|}\right|^p\,\mathrm{d}z-C\left(1+\frac{1+|\log(\beta)|}{\beta}\right)
\end{equation*}
for every $h\in\R^d$ such that $2|h|\leq \dist(O^{\prime},\partial O)$. Using the difference quotient characterization of $W^{1,p}$-spaces we conclude that $v\in W^{1,p}(O,\R^n)$ and letting $|h|\to 0$ yields by the arbitrariness of $O^{\prime}$ that
\begin{equation*}
\mathcal{F}^-(O,v)\geq \frac{1}{C|O|}\int_O|\nabla v|^p\,\mathrm{d}z-C\left(1+\frac{1+\log(\beta)|}{\beta}\right).
\end{equation*}
\end{proof}

%

\subsection{Quasiconvexity of the limit free energy} 
For the proof of Theorem~\ref{th:W} we next establish a lower semicontinuity result that we use to show the quasiconvexity of the free energy by soft arguments.
\begin{lemma}\label{lsc}
Let $G\in\mathcal{G}$ and let $O\in\mathcal{A}^R(\R^d)$. If $v,\hat{v}\in L^p_{\rm loc}(\R^d,\R^n)$ are such that $v=\hat{v}$ a.e. on $O$, then $\mathcal{F}^{\pm}(O,v)=\mathcal{F}^{\pm}(O,\hat{v})$. Hence the maps $L^p(O,\R^n)\ni v\mapsto \mathcal{F}^{\pm}(O,v)$ are well-defined. Moreover both are lower semicontinuous with respect to the strong $L^p(O,\R^n)$-topology. 
\end{lemma}
\begin{proof}[Proof of Lemma~\ref{lsc}]
Let $v_j,v\in L^p_{\rm loc}(\R^d,\R^n)$ such that $v_j\to v$ in $L^p(O,\R^n)$. Both claims follow if we establish the lower semicontinuity along such sequences. Given $u\in\mathcal{N}_p(v_{j},O,\e,\kappa)$, by \eqref{e.volVoronoi} we have
\begin{align*}
\e^{\frac{d}{p}}\|v_{\e}-\e u\|_{\ell_\e^p(O)}&\leq\e^{\frac{d}{p}}\|v_{\e}-(v_{j})_{\e}\|_{\ell^p_{\e}(O)}+\e^{\frac{d}{p}}\|(v_{j})_{\e}-\e u\|_{\ell^p_{\e}(O)}\\
&\leq C\Big(\sum_{x\in O_{\e}^\calL}\int_{\e \mathcal{C}(x)}|v(z)-v_j(z)|^p\,\mathrm{d}z\Big)^{\frac{1}{p}}+\kappa|O|^{\frac{1}{p}+\frac{1}{d}}\\
&\leq C\Big(\|u-u_j\|^p_{L^p(O)}+\int_{\partial O+B_{R\e}(0)}|v(z)-v_j(z)|^p\,\mathrm{d}z\Big)^{\frac{1}{p}}+\kappa|O|^{\frac{1}{p}+\frac{1}{d}}.
\end{align*}
Since $O\in\Ard$ we deduce that for all $j=j(\kappa)$ large enough there exists $\e_0=\e_0(j)$ such that for all $\e<\e_0$ we have $\mathcal{N}_p(v_{j},O,\e,\kappa)\subset \mathcal{N}_p(v,O,\e,2\kappa)$. For every fixed $\kappa_0$ and $j=j(\kappa_{0})$ large enough this yields
\begin{equation*}
\mathcal{F}^+(O,v_j)=\sup_{\kappa>0}\mathcal{F}_{\kappa}^+(O,v_j)\geq \mathcal{F}^+_{\kappa_0}(O,v_j)\geq \mathcal{F}^+_{2\kappa_0}(O,v).
\end{equation*}
Letting first $j\to +\infty$ and then $\kappa_0\to 0$ we conclude that
\begin{equation*}
\liminf_j\mathcal{F}^+(O,v_j)\geq \lim_{\kappa_0\to 0}\mathcal{F}^+_{2\kappa_0}(O,v)=\mathcal{F}^+(O,v).
\end{equation*}
The proof for $\mathcal{F}^-(O,v)$ is similar.
\end{proof}
We now state an important intermediate result that will imply by more or less standard arguments a large deviation principle for large volume Gibbs measures under clamped boundary conditions. Due to that reason, we postpone its proof to the end of Section \ref{Sec4} on the large deviation principle.
\begin{theorem}\label{quasiconvexetc}
Assume Hypothesis~\ref{Hypo1}. Then for a set of full probability and for any $v\in W^{1,p}(D,\mathbb{R}^n)$ it holds that 
\begin{equation*}
\mathcal{F}^-(D,v)=\mathcal{F}^+(D,v)=\frac{1}{|D|}\int_D\overline W^{\beta}(\nabla v)\,\mathrm{d}x,
\end{equation*}	
where $\overline W$ is given by Proposition~\ref{existence}.
\end{theorem}

\begin{proof}[Proof of Theorem~\ref{th:W}]
The almost sure existence of the limit of the free energy for all $\Lambda\in\R^{n\times d}$, all $\beta>0$ and all bounded Lipschitz domains $D$ follows from Proposition~\ref{fullexistence} and Remark~\ref{alltemp}. The claimed $p$-growth conditions have been proven in the Lemmata \ref{ub} and \ref{lb}. Quasiconvexity is a standard result on necessary conditions for weak lower semicontinuity of integral functionals on $W^{1,p}(D,\R^n)$, provided the integrand is continuous (as proven in Proposition~\ref{fullexistence}) and satisfies the proven $p$-growth. Thus quasiconvexity of the map $\Lambda\mapsto \overline{W}^{\beta}(\Lambda)$ is a consequence of Theorem~\ref{quasiconvexetc}, Lemma~\ref{lsc} and the Sobolev embedding theorem. Finally the claim on the ergodic case is contained in Proposition~\ref{existence}.
\end{proof}


\section{Large deviation principle for the Gibbs measures: Proof of Theorems~\ref{th:Helmholtz}~and~\ref{LDP}}\label{Sec4}

We now turn our attention to the announced large deviation principle for Gibbs measures associated with the discrete Hamiltonian $H_{\e}$. As a by-product we shall prove Theorem~\ref{th:Helmholtz}. 

\subsection{Notation for Gibbs measures and exponential tightness}
In order to avoid technical issues when discretizing the gradient of a Sobolev function on a vanishing set, we restrict our analysis to boundary data $\varphi\in {\rm Lip}(\R^d,\R^n)$. Moreover, in order to identify the discrete variables with a function defined on the continuum we proceed as follows: Recall that given any $v:D\to\R^n$, we have set $u:=\Pi_{1/\e}v:D_{\e}\to\R^n$ as
\begin{equation*}
\Pi_{1/\e}v(z)=\frac{1}{\e} v(\e z).
\end{equation*}
Given such $v$, with a slight abuse of notation we write $u=\Pi_{1/\e}v\in\mathcal{B}_{\e}(D,\varphi)$ if and only if the following conditions are met:
\begin{itemize}
\item[(i)] $\big(\Pi_{1/\e}v\big)_{|\mathcal{C}(x)\cap D_{\e}}$ is constant for all $x\in\calL$;
\item[(ii)] $\big(\Pi_{1/\e}v\big)_{|D_{\e}^{\calL}}\in\mathcal{B}_{\e}(D,\varphi)$ in the usual sense;
\item[(iii)] $\big(\Pi_{1/\e}v\big)_{|\mathcal{C}(x)\cap D_{\e}}=\big(\Pi_{1/\e}\varphi\big)(x)$ whenever $x\in\calL\setminus D_{\e}$
\end{itemize}
Then the Gibbs measure $\mu^{\beta}_{\e,D,g}$ with respect to the Hamiltonian $H_{\e}$ and boundary value $\varphi$ is the probability measure on $L^p(D,\R^n)$ given by the formula \eqref{defgibbs}, that is,
\begin{equation*}
\mu^{\beta}_{\e,D,\varphi}(V)=\frac{1}{Z^{\beta}_{\e,D,\varphi}}\int_{\Pi_{1/\e}V\cap \mathcal{B}_{\e}(D,\varphi)}\exp(-\beta H_{\e}(D,u))\,\mathrm{d}{u},
\end{equation*}
where the partition function $Z^{\beta}_{\e,D,\varphi}$ is the normalizing factor that ensure that $\mu^{\beta}_{\e,D,\varphi}(L^p(D,\R^n))=1$. With what we have proved so far, we are now in a position to state and prove a large deviation principle for these Gibbs measures in the many-particle limit. As usual, we first have to establish an exponential tightness estimate. This will be achieved in the two lemmata below.
\begin{lemma}\label{compactness}
Assume Hypothesis~\ref{Hypo1} and let $G\in\mathcal{G}$. Fix $O\in\mathcal{A}^R(\R^d)$ and $\varphi\in {\rm Lip}(\R^d,\R^n)$. If $u^{\e}\in\mathcal{B}_{\e}(O,\varphi)\cap\mathcal{S}_M(O,\e)$, then there exists a subsequence $u^{\e_j}$ and  $v\in \varphi+W^{1,p}_0(O,\R^n)$ such that $\Pi_{\e_j}u^{\e_j}:=\e_j u^{\e_j}(\e_j^{-1}\cdot)\to v$ in $L^p(O,\R^n)$.
\end{lemma}
\begin{proof}[Proof of Lemma~\ref{compactness}]
We just sketch the argument. First extend $u^{\e}$ to the whole vertex set $\calL$ setting $u^{\e}(x)=\big(\Pi_{1/\e}\varphi\big)(x)$ whenever $x\in\Lw\backslash O_{\e}^\calL$. Now take $O_1,O_2\in\mathcal{A}^R(\mathbb{R}^d)$ such that $O\subset\subset O_1\subset\subset O_2$. To reduce notation, we introduce $v^{\e}:O_2\to\mathbb{R}^n$ as $v^{\e}=\Pi_{\e}u^{\e}$.  Since $u^{\e}\in\mathcal{B}_{\e}(O,\varphi)\cap \mathcal{S}_M(O,\e)$ and $\varphi$ is Lipschitz, one can show that 
\begin{equation*}
\sup_{\e>0}|(O_2)_{\e}|^{-1}H_{\e}(u^{\e},O_2)<+\infty.
\end{equation*}
Using the same construction as for the proof of (\ref{frechetest}) we obtain that, for $h\in\mathbb{R}^d$ such that $2|h|\leq\dist(O_1,\partial O_2)$, it holds that
\begin{equation}\label{sobolevreg}
\int_{O_1}|v^{\e}(z+h)-v^{\e}(z)|^p\,\mathrm{d}z\leq C(|h|^p+|h|\e^{p-1}+\e|h|^{p-1}+\e^p).
\end{equation}
According to \cite[Lemma~4.6]{glpe}, strong $L^p(O_1)$-compactness follows if we prove that $v^{\e}$ is bounded in $L^p(O_2)$. This can be achieved combining the energy bound with the growth assumptions from Hypothesis~\ref{Hypo1} and the properly scaled discrete Poincar\'e inequality stated in Lemma~\ref{poincare}. The regularity of any limit function $v$ follows again by the difference quotient characterization of $W^{1,p}(O_1,\mathbb{R}^n)$ and \eqref{sobolevreg}. Since $\varphi$ is Lipschitz, it can be shown that $v=\varphi$ on $O_1\backslash O$ and therefore $v$ has trace $\varphi$ on $\partial O$.
\end{proof}

\begin{lemma}\label{exptight}
Assume Hypothesis~\ref{Hypo1} and let $G\in\mathcal{G}$. Then, for each $N\in\mathbb{N}$ there exists a compact set $K_N\subset L^p(D,\R^n)$ such that
\begin{equation*}
\limsup_{\e\to 0}\frac{1}{\beta|D_{\e}|}\log\Big(\mu^{\beta}_{\e,D,\varphi}(L^p(D,\R^n)\backslash K_N)\Big)\leq -N.
\end{equation*}
\end{lemma}
\begin{proof}[Proof of Lemma~\ref{exptight}]
For a given number $M>0$ we define the set $K_M\subset L^p(D,\R^n)$ by
\begin{equation*}
K_M:=\bigcup_{0<\e<1}\Big\{v:D\to\R^n:\;\Pi_{1/\e}v\in\mathcal{B}_{\e}(D,\varphi),\,H_{\e}(D,\Pi_{1/\e}v)\leq M|D_{\e}|\Big\},
\end{equation*}
where we identify again discrete functions with piecewise constant function on Voronoi cells. We argue that the set $K_M$ is precompact in $L^p(D,\R^n)$. To this end consider a sequence $\{v_j\}\subset K_M$. Then for each $j$ we find $\e_j$ such that $v_j$ is defined on the nodes of $\e_j\calL$. First let us extend the functions to all of $\e_j\calL$ setting $v_j(\e_j x)=\varphi(\e_j x)$ for $x\in\calL\backslash D_{\e_j}$.  We distinguish two cases: If $\liminf_j\e_j>0$, then we can use the boundary conditions and the energy bound to prove that $v_j$ contains a converging subsequence since it can be identified with an equibounded sequence in a finite dimensional space. Here we use again the fact that each connected component of $G_{D,\e}$ contains a vertex with active boundary conditions. Next we treat the case when $\liminf_j\e_j=0$. In that case we can apply Lemma~\ref{compactness} to conclude that $K_M$ is precompact for every $M$.

For the claimed estimate we have to control the contribution from the partition function. Using the upper bound from Hypothesis~\ref{Hypo1}, Remark~\ref{reverse}, \eqref{volumeestimate} and a change of variables, for $\e$ small enough we obtain
\begin{align}\label{partitionub}
\frac{-1}{\beta|D_{\e}|}\log\Big(Z^{\beta}_{\e,D,\varphi}\Big)&\leq\frac{-1}{\beta|D_{\e}|}\log\Big(\int_{ \mathcal{B}_{\e}(D,\varphi)}\exp\big(-C\beta(\|\nabla_\B u\|^p_{\ell^p_{\e}(D)}+|D_{\e}^{\calL}|)\big)\,\mathrm{d}u\Big) \nonumber \\
&\leq\frac{-1}{\beta|D_{\e}|}\log\Big(\int_{ \mathcal{B}_{\e}(D,\varphi)}\exp\big(- C\beta(\|\nabla_\B (u-\Pi_{1/\e}\varphi)\|^p_{\ell^p_{\e}(D)}+(\|\nabla \varphi\|^p_{\infty}+1)|D_{\e}^{\calL}|)\big)\,\mathrm{d}u\Big) \nonumber \\
&\leq C(\|\nabla \varphi\|_{\infty}^p+1)-\frac{1}{\beta|D_{\e}|}\log\Big(\int_{ \mathcal{B}_{\e}(D,0)}\exp(- C\beta\|u\|^p_{\ell^p_{\e}(D)})\,\mathrm{d}u\nonumber \\
&\leq C(\|\nabla \varphi\|_{\infty}^p+1)+\frac{C}{\beta}\Big|\log\Big(\int_{B_1(0)}\exp(-C\beta|\zeta|^p)\,\mathrm{d}\zeta\Big)\Big|.
\end{align}
Combining this bound with Lemma~\ref{tightness} we obtain the claim choosing $K_{M}$ with $M=M(N,\beta)$ large enough and taking the $L^p$-closure of this set.
\end{proof}

\subsection{Proof of the large deviation principle}
Having established the exponential tightness we can now prove a strong large deviation principle for the large volume Gibbs measures as stated in the main results. Theorem~\ref{th:Helmholtz} will then be a straightforward consequence of the proof.

\begin{proof}[Proof of Theorem~\ref{LDP}]
Observe first that the term $-\frac{1}{\beta|D_{\e}|}\log(Z^{\beta}_{\e,D,\varphi})$ is bounded from above as shown in \eqref{partitionub}. A corresponding lower bound can be achieved using the lower bound of Hypothesis~\ref{Hypo1} and Lemma~\ref{forrestargument}. Hence we may assume that, passing to a subsequence (not relabeled), it holds that
\begin{equation*}
\lim_{\e\to 0}-\frac{1}{\beta|D_{\e}|}\log(Z^{\beta}_{\e,D,\varphi})=c_{\varphi,\beta}
\end{equation*}
for some constant $c_{\varphi,\beta}\in\R$. To reduce notation, we define the functional $I_g:L^p(D,\R^n)\to (-\infty,+\infty]$ via
\begin{equation*}
I^{\beta}_{D,\varphi}(v)=
\begin{cases}
\displaystyle\frac{1}{|D|}\int_D \overline W^{\beta}(\nabla v)\,\mathrm{d}x &\mbox{if $v\in \varphi+W^{1,p}_0(D,\R^n)$},\\
+\infty &\mbox{otherwise.}
\end{cases}
\end{equation*}
Note that by the upper and lower bounds established in Lemma~\ref{ub} and Lemma~\ref{lb}, respectively, as well as the quasiconvexity proven in Theorem~\ref{quasiconvexetc}, we know that $I^{\beta}_{D,\varphi}$ is lower-semicontinuous with respect to strong $L^p(D,\R^n)$-convergence.

\step1{Proof of the lower bound on open sets}
We start the proof with the case of an open set $U\subset L^p(D,\R^n)$. If $U\cap (\varphi+W^{1,p}_0(D,\R^n))=\emptyset$, then there is nothing to prove. Therefore consider $v\in U\cap (\varphi+W^{1,p}_0(D,\R^n))$. Since $U$ is open, given $\eta>0$ we can find $v^{\eta}\in U\cap (\varphi+C_c^{\infty}(D,\R^n))$ such that $\|v^{\eta}-v\|_{W^{1,p}(D)}<\eta$. We claim that, for fixed $\eta>0$, there exist $\kappa_0,\e_0>0$ such that for all $\kappa<\kappa_0$ and $\e<\e_0$ it holds that
\begin{equation}\label{openness}
\Pi_{\e}\Big(\mathcal{N}_p(v^{\eta},D,\e,3\kappa)\cap \mathcal{B}_{\e}(D,\varphi)\Big)\subset U
\end{equation}
Indeed, recalling the definition of $\tilde{v}^{\eta}_{\e}$ in Remark~\ref{approximation}, for every $u\in \mathcal{N}_p(v^{\eta},D,\e,3\kappa)\cap \mathcal{B}_{\e}(D,\varphi)$ we have that
\begin{align}\label{normestimateI}
\|v^{\eta}-\Pi_{\e}u\|_{L^p(D)}&\leq \|v^{\eta}-\tilde{v}^{\eta}_{\e}\|_{L^p(D)}+C\Big(\sum_{x\in D_{\e}}\e^d|v^{\eta}_{\e}(x)-\e u(x)|^p+\sum_{\e\mathcal{C}(x)\cap \partial D\neq\emptyset}\e^d |\tilde{v}^{\eta}_{\e}(\e x)-\varphi(\e x)|^p\Big)^{\frac{1}{p}}\nonumber
\\
&\leq \|v^{\eta}-\tilde{v}^{\eta}_{\e}\|_{L^p(D)}+C\Big((3\kappa)^p|D|^{1+\frac{p}{d}}+\sum_{\e\mathcal{C}(x)\cap \partial D\neq\emptyset}\e^d |\tilde{v}^{\eta}_{\e}(\e x)-\varphi(\e x)|^p\Big)^{\frac{1}{p}}.
\end{align}
The $\e$-dependent terms vanish by Remark~\ref{approximation} combined with an equiintegrability argument for the last sum. Hence (\ref{openness}) holds provided we choose $\e_0,\kappa_0$ small enough. Then from the definition of the Gibbs measure we infer for $\kappa<\kappa_0$ that
\begin{equation*}
\liminf_{\e\to 0}\frac{\log(\mu^{\beta}_{\e,D,\varphi}(U))}{\beta|D_{\e}|}\geq \liminf_{\e\to 0}\frac{1}{\beta|D_{\e}|}\log\Big(Z^{\beta}_{\e,D}(\mathcal{N}_p(v^{\eta},D,\e,3\kappa)\cap \mathcal{B}_{\e}(D,\varphi)\Big)+c_{\varphi,\beta}.
\end{equation*}
Applying Proposition~\ref{interpolation}, we deduce that for any $\delta>0$ small enough and any $N\in\mathbb{N}$ it holds that
\begin{equation*}
\liminf_{\e\to 0}\frac{\log(\mu^{\beta}_{\e,D,\varphi}(U))}{\beta|D_{\e}|}\geq -\frac{N-C}{N}\mathcal{F}^+_{\kappa}(D,v^{\eta})-C(\|\nabla \varphi\|^p_{\infty}+1)\frac{|D_{\delta}|}{|D|}-\Big(\frac{N\kappa|D|^{\frac{1}{d}}}{\delta}\Big)^p-\frac{C}{N}+c_{\varphi,\beta}.
\end{equation*}
Letting first $\kappa\to 0$ and then $N\to +\infty$ as well as $\delta\to 0$, from Theorem~\ref{quasiconvexetc} we infer
\begin{equation*}
\liminf_{\e\to 0}\frac{\log(\mu^{\beta}_{\e,D,\varphi}(U))}{\beta|D_{\e}|}\geq-\frac{1}{|D|}\int_D \overline W^{\beta}(\nabla v^{\eta})\,\mathrm{d}x+c_{\varphi,\beta}.
\end{equation*}
As $\eta>0$ was arbitrary, the continuity of $\Lambda \mapsto \overline W^{\beta}(\Lambda )$ and its growth condition allow to pass from $v^{\eta}$ to $v$ and since $v\in U\cap (\varphi+W_0^{1,p}(D,\R^n))$ was arbitrary too, we conclude the lower bound
\begin{equation*}
\liminf_{\e\to 0}\frac{\log(\mu^{\beta}_{\e,D,\varphi}(U))}{\beta|D_{\e}|}\geq -\inf_{v\in U}I^{\beta}_{D,\varphi}(v)+c_{\varphi,\beta}.
\end{equation*}

\step2{Proof of the upper bound on closed sets}
In order to prove an upper bound, we first recall that due to the exponential tightness established in Lemma~\ref{exptight}, it suffices to consider the case when $V$ is compact (see for example Lemma~1.2.18 in \cite{DeZe}). Then, for $\delta>0$ we define the truncated energy via 
\begin{equation*}
\mathcal{F}_{\delta}(D,v)=\min\left\{\mathcal{F}^-(D,v)-\delta,\frac{1}{\delta}\right\}.
\end{equation*}
Note that by definition of $\mathcal{F}^-(D,v)$, for every $v\in V$ there exists $\kappa>0$ such that
\begin{equation}\label{localub}
-\limsup_{\e\to 0}\frac{1}{\beta|D_{\e}|}\log(Z^{\beta}_{\e,D}(\mathcal{N}_p(v,D,\e,\kappa))\geq \mathcal{F}_{\delta}(D,v).
\end{equation}
Let us fix $C_1$ such that $|\mathcal{C}(x)|^{\frac{1}{p}}\leq C_1$ for all $x\in\calL$. By lower semicontinuity of the functional $I_{D,\varphi}^{\beta}$, up to reducing $\kappa$ we may assume that
\begin{equation}\label{locallsc}
I_{D,\varphi}^{\beta}(v)\leq I_{D,\varphi}^{\beta}(w)+1
\end{equation}
for all $w\in L^p(D,\R^n)$ such that $\|v-w\|_{L^p(D,\R^n)}\leq C_1\kappa |D|^{\frac{1}{p}+\frac{1}{d}}$.
As we show now, for a suitable $0<\kappa^{\prime}<\kappa$ and all $\e$ small enough, we have the inclusion
\begin{equation}\label{containedindiscrete}
B_{\kappa^{\prime}}(v)\cap\Pi_{\e}(\mathcal{B}_{\e}(D,\varphi))\subset \Pi_{\e}\Big(\mathcal{N}_p(v,D,\e,\kappa)\cap\mathcal{B}_{\e}(D,\varphi)\Big),
\end{equation}
where here we denote by $B_{\kappa^{\prime}}(v)$ the $L^p(D,\R^n)$-ball centered at $u$ with radius $\kappa^{\prime}$. Indeed, from \eqref{e.volVoronoi} we deduce that any $u\in\mathcal{B}_{\e}(D,\varphi)$ with $\Pi_{\e}u\in B_{\kappa^{\prime}}(v)$ satisfies 
\begin{align*}
\sum_{x\in D^{\calL}_{\e}}\e^d|v_{\e}(x)-\e u(x)|^p&\leq C\|\tilde{v}_{\e}-\Pi_{\e}u\|^p_{L^p(D)}+C\sum_{\e\mathcal{C}(x)\cap\partial D\neq\emptyset}\e^d\left(|v_{\e}(x)-\varphi(\e x)|^p+\e^p\right)
\\
&\leq C\|\tilde{v}_{\e}-v\|^p_{L^p(D)}+C\kappa^{\prime}+\sum_{\e\mathcal{C}(x)\cap\partial D\neq\emptyset}\e^d\left(|v_{\e}(x)-\varphi(\e x)|^p+\e^p\right)
\end{align*}
and again the $\e$-dependent terms converge to zero by Remark~\ref{approximation} and an equiintegrability argument for the sum in the second line. Since $V$ is compact, we can find a finite covering by the open balls $B_{\kappa^{\prime}}(u)$, that is there exist $v_1,\dots,v_m$ such that $V\subset \bigcup_{i=1}^m B_{\kappa^{\prime}_i}(v_i)$. Together with the the inclusion \eqref{containedindiscrete} this covering implies
\begin{align}\label{ldpubest}
\limsup_{\e\to 0}\frac{\log(\mu^{\beta}_{\e,D,\varphi}(V))}{\beta|D_{\e}|}&\leq\limsup_{\e\to 0}\frac{1}{\beta|D_{\e}|}\log\left(\sum_{i=1}^m\mu^{\beta}_{\e,D,\varphi}(B_{\kappa^{\prime}_i}(v_i))\right)\nonumber
\\
&\leq \max_{i}\limsup_{\e\to 0}\frac{1}{\beta|D_{\e}|}\Big(\log(Z^{\beta}_{\e,D}(\mathcal{N}_p(v_i,D,\e,\kappa_i)\cap\mathcal{B}_{\e}(D,\varphi))\Big)+c_{\varphi,\beta}
\end{align}
and therefore it remains to bound the term for a fixed $v_{i_0}$. First note that if
\begin{equation*}
\limsup_{\e\to 0}\frac{1}{\beta|D_{\e}|}\Big(\log(Z^{\beta}_{\e,D}(\mathcal{N}_p(v_{i_0},D,\e,\kappa_{i_0})\cap\mathcal{B}_{\e}(D,\varphi))\Big)=-\infty,
\end{equation*}
then there is nothing left to prove. Otherwise, Lemma~\ref{tightness} implies that, for a suitable large $M$, it holds that 
\begin{equation*}
\mathcal{N}_p(v_{i_0},D,\e,\kappa_{i_0})\cap\mathcal{B}_{\e}(D,\varphi)\cap \mathcal{S}_M(D,\e)\neq\emptyset
\end{equation*}
along some infinitesimal sequence $\e\to 0$ (not relabeled). Applying Lemma~\ref{compactness} to this element, we deduce that there exists $v\in \varphi+W^{1,p}_0(D,\R^n)$ such that, similar to estimate \eqref{normestimateI}, it holds that
\begin{equation*}
\|v_{i_0}-\varphi\|_{L^p(D)}\leq C_1\kappa_{i_0}|D|^{\frac{1}{p}+\frac{1}{d}}
\end{equation*} 
Together with (\ref{locallsc}) this implies that $v_{i_0}\in \varphi+W^{1,p}_0(D,\R^n)$, too. Therefore, using also \eqref{localub}, we can further estimate \eqref{ldpubest} by 
\begin{align*}
\limsup_{\e\to 0}\frac{\log(\mu_{\e,g}(A))}{\beta|D_{\e}|}&\leq \limsup_{\e\to 0}\frac{1}{\beta|D_{\e}|}\log\Big(Z^{\beta}_{\e,D}(\mathcal{N}_p(v_{i_0},D,\e,\kappa_{i_0}))\Big)+c_{\varphi,\beta}
\\
&\leq -\mathcal{F}_{\delta}(D,v_{i_0})+c_{\varphi,\beta}\leq -\inf_{v\in V\cap g+W^{1,p}_0(D,\R^n)}\mathcal{F}_{\delta}(D,v)+c_{\varphi,\beta}.
\end{align*} 
Letting $\delta\to 0$, by monotonicity and Theorem~\ref{quasiconvexetc} we obtain the estimate
\begin{equation*}
\limsup_{\e\to 0}\frac{1}{\beta|D_{\e}|}\log(\mu^{\beta}_{\e,D,\varphi}(V))\leq -\inf_{v\in V}I^{\beta}_{D,\varphi}(u)+c_{\varphi,\beta}.
\end{equation*}

\step3{Identification of $c_{\varphi,\beta}$ and conclusion}
It remains to show that $c_{\varphi,\beta}$ does not depend on the subsequence. Testing the open and closed set $L^p(D,\R^n)$ it immediately follows that $c_{\varphi,\beta}=\inf_{v\in L^p(D,\R^n)}I^{\beta}_{D,\varphi}(v)$ and this proves the large deviation principle with rate functional $\mathcal{I}_{D,\varphi}^{\beta}$ as claimed in Theorem~\ref{LDP}.
\end{proof}

\begin{proof}[Proof of Theorem~\ref{th:Helmholtz}]
Observe that, by the definitions in \eqref{defpartfunc} and \eqref{deffreeenergy}, in Step~3 above we also proved the claim on the Helmholtz free energy with boundary condition $\varphi$.
\end{proof}

From the large deviation principle we obtain the following qualitative behavior of the Gibbs measures.
\begin{corollary}\label{concentration}
Let $\e_j\to 0$. Under the assumptions of Theorem~\ref{LDP}, for a set of full probability the sequence of measures $\mu^{\beta}_{\e_jD,\varphi}$ is compact with respect to weak$^*$-convergence and each cluster point as $\e_j\to 0$ is a probability measure whose support is contained in the set of minimizers of the rate functional $\mathcal{I}_{D,\varphi}^{\beta}$.
\end{corollary}

\subsection{Asymptotic analysis of the localized partition function}
At the end of this section we now give the technical proof of Theorem~\ref{quasiconvexetc}, which was used in the proof of the large deviation principle. 
\begin{proof}[Proof of Theorem~\ref{quasiconvexetc}]
	Let $G\in\mathcal{G}^{\prime}$ with $\mathcal{G}^{\prime}$ the set of full probability implicitly given by Proposition~\ref{fullexistence}. The argument consists of two steps.
	
	\step1 {Proof of the upper bound}
	We show that 
	\begin{equation}\label{ldub}
	\mathcal{F}^+(D,v)\leq \frac{1}{|D|}\int_D\overline W^{\beta}(\nabla v)\,\mathrm{d}x.
	\end{equation}
	By the lower semicontinuity of $v\mapsto \mathcal{F}^+(D,v)$ established in Lemma~\ref{lsc} and the p-growth conditions and continuity of $\overline W^{\beta}(\Lambda )$ (cf. the Lemmata \ref{ub} and \ref{lb} and Proposition~\ref{fullexistence}) it is enough to prove the estimate for continuous piecewise affine functions. More precisely, we consider a locally finite triangulation $\mathcal{T}=\{T\}$ of $\mathbb{R}^d$ and a Lipschitz function $v\in W^{1,\infty}(\R^d,\mathbb{R}^n)$ such that for all $T\in\mathcal{T}$ there exists $\Lambda _T\in\mathbb{R}^{n\times d}$ and $b_T\in\R^n$ with $v_{|T}(y)=\Lambda _Ty+b_T$. To simplify notation we introduce the almost lower-dimensional set
	\begin{equation*}
	S=\bigcup_{T\in\mathcal{T}}(\partial T\cap D)\cup\bigcup_{T\cap\partial D\neq\emptyset}(T\cap D).
	\end{equation*}
	A direct computation shows that, for fixed $\kappa>0$ we find $\e_0$ such that for all $\e<\e_0$ we have the inclusion
	\begin{equation}\label{triangledecomposition}
	\prod_{T\subset D}\left(\mathcal{N}_p(v,T,\e,\frac{\kappa}{2})\cap \mathcal{B}_{\e}(T,\varphi_{\Lambda_T}+b_T)\right)\times \mathcal{N}_{\infty}(\e^{-1}v_{\e},S,\e)\subset \mathcal{N}_p(v,D,\e,\kappa).
	\end{equation}
	We aim to establish a kind of subadditivity estimate. To this end, observe that when $u$ belongs to the set on the left hand side and $x,x^{\prime}\in D^{\calL}_{\e}$ are such that $|x-x^{\prime}|\leq C_0$, we have the following bounds:
	\begin{itemize}
		\item [(i)]  If $x\in T_{\e}$ and $x^{\prime}\in T^{\prime}_{\e}$ for $T\neq T^{\prime}$, then 
		\begin{equation*}
		|u(x)-u(x^{\prime})|\leq 2+\frac{1}{\e}|v(\e x)-v(\e x^{\prime})|\leq 2+C_0\|\nabla v\|_{\infty}.
		\end{equation*}
		\item[(ii)] If $x\in T_{\e}$ and $x^{\prime}\in S_{\e}$, then by definition of $v_{\e}$ in \eqref{discreteapprox} and \eqref{e.volVoronoi}
		\begin{align*}
		|u(x)-u(x^{\prime})|&\leq 2+\frac{1}{\e}|v(\e x)-v_{\e}(x^{\prime})|\leq 2+R\|\nabla u\|_{\infty}+\frac{1}{\e}|v(\e x)-v(\e x^{\prime})|
		\leq 2+(R+C_0)\|\nabla v\|_{\infty}.
		\end{align*} 
		\item[(iii)] If $x,x^{\prime}\in S_{\e}$, then by same reasoning as for (ii)
		\begin{equation*}
		|u(x)-u(x^{\prime})|\leq 2+\frac{1}{\e}|v_{\e}(x)-v_{\e}(x^{\prime})|\leq 2+(2R+C_0)\|\nabla v\|_{\infty}.
		\end{equation*}
	\end{itemize}
	By Hypothesis~\ref{Hypo1}, the above bounds, \eqref{volumeestimate} and \eqref{surfaceesimtate} there exists a constant $C_{v}$ depending only on $\|\nabla v\|_{\infty}^p$ such that, for $\e$ sufficiently small,
	\begin{equation*}
	H_{\e}(D,u)\leq\sum_{T\subset D}H_{\e}(T,u)+C_{v}\bigg(\e^{1-d}\sum_{T\cap D\neq\emptyset }\mathcal{H}^{d-1}(\partial T)+\e^{-d}\sum_{T\cap\partial D\neq\emptyset}|T\cap D|\bigg).
	\end{equation*}
	Upon taking the inverse exponential, integrating the above inequality over the left hand side set in \eqref{triangledecomposition} and applying Fubini's Theorem~yields the estimate
	\begin{align*}
	Z^{\beta}_{\e,D}(\mathcal{N}_p(v,D,\e,\kappa))\geq& \prod_{T\subset D}Z^{\beta}_{\e,T}(\mathcal{N}_p(v,T,\e,\frac{\kappa}{2})\cap\mathcal{B}_{\e}(T,\varphi_{\Lambda _T}+b_T))
	\\
	&\times \exp\bigg(-C_{v}\beta\e^{1-d}\sum_{T\cap D\neq\emptyset }\mathcal{H}^{d-1}(\partial T)-C_{v}\beta\e^{-d}\sum_{T\cap\partial D\neq\emptyset}|T\cap D|\bigg),
	\end{align*}
	where we incorporated the measure of $\mathcal{N}_{\infty}(\e^{-1}v_{\e},S,\e)$ in the exponential term in the last line, possibly increasing the value of $C_{v}$ by a multiplicative factor. On each simplex $T\subset D$ we use the translation invariance of the Hamiltonian to get rid of the constant $b_T$ and obtain
	\begin{align*}
	Z^{\beta}_{\e,D}(\mathcal{N}_p(v,D,\e,\kappa))\geq& \prod_{T\subset D}Z^{\beta}_{\e,T}(\mathcal{N}_p(\overline{\varphi}_{\Lambda _T},T,\e,\frac{\kappa}{2})\cap\mathcal{B}_{\e}(T,\varphi_{\Lambda _T}))
	\\
	&\times \exp\bigg(-C_{v}\beta\e^{1-d}\sum_{T\cap D\neq\emptyset }\mathcal{H}^{d-1}(\partial T)-C_{v}\beta\e^{-d}\sum_{T\cap\partial D\neq\emptyset}|T\cap D|\bigg),
	\end{align*}
	Taking logarithms and dividing by $-\beta|D_{\e}|$, when $\e\to 0$ we infer from Lemma~\ref{equivalent} and Proposition~\ref{fullexistence} that
	\begin{align*}
	\mathcal{F}_{\kappa}^+(D,v)\leq& \sum_{D\subset T}\frac{|T|}{|D|}\overline W^{\beta}(\nabla v_{|T})+C_{v}\sum_{T\cap\partial D\neq\emptyset}\frac{|T\cap D|}{|D|}
	\\
	\leq&\frac{1}{|D|}\int_D\overline W^{\beta}(\nabla v)\,\mathrm{d}x+\sum_{T\cap\partial D\neq\emptyset}\left(C_u-\overline W^{\beta}(\Lambda _T)\right)\frac{|T\cap D|}{|D|}.
	\end{align*}
	Now keeping $v$ fixed, we let first $\kappa\to 0$ and then we refine the triangulation $\mathcal{T}$ and by the regularity of $\partial D$ the last sum can be made arbitrarily small. This proves \eqref{ldub}.
	
	\step2 {Proof of the lower bound}
	We now turn to the argument for the inequality
	\begin{equation*}
	\mathcal{F}^-(D,v)\geq \frac{1}{|D|}\int_D\overline W^{\beta}(\nabla v)\,\mathrm{d}x.
	\end{equation*}
	Let us assume without loss of generality that $\overline W^{\beta}(\Lambda )\geq 0$. Due to Lemma~\ref{lb} this can be achieved by adding a large constant to the discrete energy density $f$. This perturbation yields a (random) additive constant on both sides due to the superadditive version of the ergodic theorem. We want to apply the blow-up Lemma~proven in \cite{KoLu} that allows to treat $v$ locally as an affine function. To this end, we need some notation. First we extend the target function $v\in W^{1,p}(D,\R^n)$ (without relabeling) to a function $v\in W^{1,p}(\R^d,\R^n)$ with compact support. Next, given $\rho>0$ and $\xi\in\R^d$ we define the periodic lattice $\mathcal{L}_{\rho,\xi}=\xi+\rho\mathbb{Z}^d$. Note that for fixed $\rho$, for almost all $\xi\in\R^d$ the set $\mathcal{L}_{\rho,\xi}$ consists of Lebesgue points of $\nabla v$. Hence we can define for such $\xi$ a (not necessarily continuous) piecewise affine approximation as
	\begin{equation*}
	L_{\rho,\xi}v(y)=\nabla v(z)(y-z)+\frac{1}{\rho^d}\int_{Q(z,\rho)}v(x)\,\mathrm{d}x\quad\quad\text{ if }y\in Q(z,\rho),\,z\in\mathcal{L}_{\rho,\xi}.
	\end{equation*} 
	The main object will be the local difference between the linearization and the function itself defined for $z\in\mathcal{L}_{\rho,\xi}$ as
	\begin{equation*}
	\phi^z_{\rho,\xi}(y)=\frac{1}{\rho}\big(v(z+\rho y)-L_{\rho,\xi}v(z+\rho y)\big)\quad\quad\text{if  }y\in Q(0,1).
	\end{equation*}
	Note that due the bounds in Lemma~\ref{ub}, for any $D^{\prime}\subset\subset D$ the function $x\mapsto \overline W^{\beta}(\nabla v(x))\mathds{1}_{D^{\prime}}(x)$ is integrable on $\R^d$. Hence we are in the position to use the blow-up Lemma~by Koteck\'y and Luckhaus (see \cite[Corollary 1]{KoLu}). Note that we apply it also for the function itself which follows simply by the Poincar\'e inequality since the function $\phi_{\rho,\xi}^z$ has mean value zero on $Q(0,1)$. It states that for each $\eta>0$ we find $\rho_0>0$ such that for each $\rho<\rho_0$ there exists $\xi\in Q(0,\rho)$ with
	\begin{equation}\label{blowup}
	\begin{split}
	&\sum_{z\in\mathcal{L}_{\rho,\xi}}\rho^d\bigg(\int_{Q(0,1)}|\nabla \phi^z_{\rho,\xi}(y)|^p\,\mathrm{d}y+\int_{Q(0,1)}|\phi^z_{\rho,\xi}(y)|^p\,\mathrm{d}y\bigg)<\eta
	\\
	&\sum_{z\in\mathcal{L}_{\rho,\xi}}\rho^d\overline W^{\beta}(\nabla v(z))\mathds{1}_{D^{\prime}}(z)> \int_{D^{\prime}}\overline W(\nabla v)\,\mathrm{d}x-\eta.
	\end{split}
	\end{equation}
	Next, if $u\in\mathcal{N}_p(v,D,\e,\kappa)$, then for all cubes $Q(z,\rho)\subset D$ and $\e=\e(\rho)$ small enough the triangle inequality, the definition \eqref{discreteapprox} and a change of variables imply
	\begin{align*}
	\bigg(\sum_{x\in Q(z,\rho)^{\calL}_{\e}}\e^d|\e u(x)-(L_{\rho,\xi}v)_{\e}(x)|^p\bigg)^{\frac{1}{p}}\leq& \kappa |D|^{\frac{1}{p}+\frac{1}{d}}+\bigg(\sum_{x\in Q(z,\rho)^{\calL}_{\e}}\e^d|v_{\e}(x)-(L_{\rho,\xi}v)_{\e}(x)|^p\bigg)^{\frac{1}{p}}
	\\
	\leq&\kappa |D|^{\frac{1}{p}+\frac{1}{d}}+C\rho^{1+\frac{d}{p}}\bigg(\sum_{\substack{z^{\prime}\in \mathcal{L}_{\rho,\xi}\\ |z-z^{\prime}|\leq \rho}}\int_{Q(0,1)}|\phi_{\rho,\xi}^{z^{\prime}}(y)|^p\,\mathrm{d}y\bigg)^{\frac{1}{p}}
	\\
	=&\kappa |D|^{\frac{1}{p}+\frac{1}{d}}+C |Q(z,\rho)|^{\frac{1}{p}+\frac{1}{d}}\bigg(\sum_{\substack{z^{\prime}\in \mathcal{L}_{\rho,\xi}\\ |z-z^{\prime}|\leq \rho}}\int_{Q(0,1)}|\phi_{\rho,\xi}^{z^{\prime}}(y)|^p\,\mathrm{d}y\bigg)^{\frac{1}{p}}.
	\end{align*}
	On setting $S_{\rho}=S_{\rho}^1\cup S_{\rho}^2$, where $S_{\rho}^1=\displaystyle \bigcup_{z\in\mathcal{L}_{\rho,\xi}}(\partial Q(z,\rho)\cap D)$ and $S_{\rho}^2=\displaystyle \bigcup_{Q(z,\rho)\cap\partial D\neq\emptyset}(Q(z,\rho)\cap D)$ as well as
	\begin{equation*}
	\kappa_z=\max\bigg\{2C\bigg(\sum_{\substack{z^{\prime}\in \mathcal{L}_{\rho,\xi}\\ |z-z^{\prime}|\leq \rho}}\int_{Q(0,1)}|\phi_{\rho,\xi}^{z^{\prime}}(y)|^p\,\mathrm{d}y\bigg)^{\frac{1}{p}},\rho\bigg\},
	\end{equation*}
	we obtain for $\kappa=\kappa(\rho)$ small enough the set inclusion
	\begin{equation}\label{splitset}
	\mathcal{N}_p(v,D,\e,\kappa)\subset \prod_{\substack{z\in\mathcal{L}_{\rho,\xi}\\ Q(z,\rho)\subset D}}\mathcal{N}_p(L_{\rho,\xi}v,Q(z,\rho),\e,\kappa_z)\times \prod_{i=1}^2S_{\e}^i(\kappa,\rho),.
	\end{equation}
	where for $i=1,2$ we define the sets
	\begin{equation*}
	\begin{split}
	S^i_{\e}(\kappa,\rho):=\{u:(S^i_{\rho})^{\calL}_{\e}\to\R^n:\;\|u-\e^{-1}v_{\e}\|_{\infty}\leq \kappa|\e^{-1}D|^{\frac{1}{p}+\frac{1}{d}}\}.
	\end{split}
	\end{equation*}
	In order to control the integration over these two sets, we note that for $\e$ small enough
	\begin{equation}\label{skeletonbound}
	\log\Big(\int_{S^1_{\e}(\kappa,\rho)}\mathrm{d}u\Big)\leq C\log(C\kappa^n\e^{-\frac{nd}{p}-n})\sum_{Q(z,\rho)\cap D\neq\emptyset}\rho^{d-1}\e^{1-d}\leq C\log(C\kappa^n\e^{-\frac{nd}{p}-n})\rho^{-1}\e^{1-d}.
	\end{equation} 
	To treat the contributions from the points in $(S_{\rho}^2)_{\e}$ we have to use once again Lemma~\ref{forrestargument}. To this end we observe that 
	\begin{equation*}
	\partial S_{\rho}^2\subset \partial D\cup\bigcup_{Q(z,\rho)\cap\partial D\neq\emptyset}(\partial Q(z,\rho)\cap D).
	\end{equation*}
	As the union on the right hand side is finite, we can argue as for Remark~\ref{numbercomponents} and (\ref{surfaceesimtate}) to show that, for $\e$ small enough, the number of connected components $N_{\e,\rho}$ of the graph $G_{S_{\rho}^2,\e}$ can be bounded by
	\begin{equation}\label{numberagain}
	N_{\e,\rho}\leq C\e^{1-d}\left(\mathcal{H}^{d-1}(\partial D)+\rho^{-1}\right).
	\end{equation} 
	Due to Hypothesis~\ref{Hypo1} and Lemma~\ref{forrestargument} we deduce the bound
	\begin{align}\label{boundarystripebound}
	\log\Big(\int_{S_{\e}^2(\kappa,\rho)}\exp(-\beta H_{\e}(S_{\rho}^2,u)\,\mathrm{d}u\Big)\leq&\log\Big(\int_{S_{\e}^2(\kappa,\rho)}\exp(-\frac{\beta}{C}\|\nabla_\B u\|^p_{\ell^p_{\e}(S_{\rho}^2)})\,\mathrm{d}u \Big)+C\beta|(S_{\rho}^2)^{\calL}_{\e}|\nonumber
	\\
	\leq& \log(C\kappa^n\e^{-\frac{nd}{p}-n})N_{\e,\rho}+C(1+|\log(\beta)|+\beta)|(S_{\rho}^2)^{\calL}_{\e}|
	\end{align}
	Together with the inequality $H_{\e}(D,u)\geq \sum_{Q(z,\rho)\subset D}H_{\e}(Q(z,\rho),u)+H_{\e}(S_{\rho}^2,u)$, the inclusion in \eqref{splitset} and Fubini's Theorem~imply
	\begin{align*}
	Z^{\beta}_{\e,D}(\mathcal{N}_p(v,D,\e,\kappa))\leq& \prod_{\substack{z\in\mathcal{L}_{\rho,\xi}\\ Q(z,\rho)\subset D}}Z^{\beta}_{\e,Q(z,\rho)}(\mathcal{N}_p(L_{\rho,\xi}v,Q(z,\rho),\e,\kappa_z))
	\\
	&\quad\times \int_{S_{\e}^1(\kappa)}\mathrm{d}u_1\int_{S_{\e}^2(\kappa,\rho)}\exp(-H_{\e}(S_{\rho}^2,u_2)\,\mathrm{d}u_2.
	\end{align*}
	Taking logarithms and dividing by $-\beta|D_{\e}|$ we infer from \eqref{skeletonbound}, \eqref{numberagain} combined with \eqref{boundarystripebound} that
	\begin{equation*}
	\mathcal{F}^-(D,v)\geq \sum_{\substack{z\in\mathcal{L}_{\rho,\xi}\\ Q(z,\rho)\subset D}}\frac{\rho^d}{|D|}\mathcal{F}^-_{\kappa_z}(Q(z,\rho),\overline{\varphi}_{\nabla v(z)})-C\left(\frac{1+|\log(\beta)|+\beta}{\beta}\right)\frac{|\partial D+Q(0,2\rho)|}{|D|},
	\end{equation*}
	where we also used that the energy is invariant under constant shifts so that we can pass from the affine approximation to the linear one. Since we assume that $\overline W(\Lambda )\geq 0$, using Remark~\ref{invariance} and (\ref{blowup}) we infer that for arbitrary $N\in\mathbb{N}$ and $\delta,\rho$ sufficiently small
	\begin{align}\label{almostdone}
	\mathcal{F}^-(D,v)\geq& \sum_{\substack{z\in\mathcal{L}_{\rho,\xi}\\ Q(z,\rho)\subset D}}\frac{\rho^d}{|D|}\left(\overline W^{\beta}(\nabla v(z))-C\left((1+|\nabla v(z)|^p)\delta+\frac{(N\kappa_z)^p}{\delta^p}+\frac{1}{N}\right)\right)-C_{\beta}\frac{|\partial D+Q(0,2\rho)|}{|D|}\nonumber
	\\
	\geq&\frac{1}{|D|}\int_{D^{\prime}}\overline W^{\beta}(\nabla v)\,\mathrm{d}x-\frac{\eta}{|D|}-\sum_{\substack{z\in\mathcal{L}_{\rho,\xi}\\ Q(z,\rho)\subset D}}C\frac{\rho^d}{|D|}\left((1+|\nabla u(z)|^p)\delta+\frac{(N\kappa_z)^p}{\delta^p}+\frac{1}{N}\right)\nonumber
	\\
	&-C_{\beta}\frac{|\partial D+Q(0,2\rho)|}{|D|}.
	\end{align}
	Using again (\ref{blowup}) we can bound the sum of the gradients. Indeed, by a change of variables it holds that
	\begin{align}\label{pnormbound}
	\sum_{\substack{z\in\mathcal{L}_{\rho,\xi}\\ Q(z,\rho)\subset D}}\rho^d|\nabla v(z)|^p&\leq C\sum_{\substack{z\in\mathcal{L}_{\rho,\xi}\\ Q(z,\rho)\subset D}}\int_{Q(z,\rho)}|\nabla v(y)-\nabla v(z)|^p+|\nabla v(y)|^p\,\mathrm{d}y\nonumber
	\\
	&\leq C\sum_{z\in\mathcal{L}_{\rho,\xi}}\rho^d\int_{Q(0,1)}|\nabla \phi_{\rho,\xi}^z(y)|^p\,\mathrm{d}y+C\|\nabla v\|^p_{L^p(D)}\leq C(\eta+\|\nabla v\|_{L^p(D)}^p).
	\end{align} 
	To control the sum over $\kappa_z^p$, note that by (\ref{blowup}) and the definition of $\kappa_z$ we have
	\begin{equation}\label{kappasum}
	\sum_{\substack{z\in\mathcal{L}_{\rho,z}\\ Q(z,\rho)\subset D}}\frac{\rho^d}{|D|}\kappa_z^p\leq C\rho^{p}+\frac{C}{|D|}\sum_{z\in\mathcal{L}_{\rho,\xi}}\rho^d\int_{Q(0,1)}|\phi_{\rho,\xi}^z(y)|^p\,\mathrm{d}y\leq C(\rho^p+\frac{\eta}{|D|}).
	\end{equation}
	Putting together \eqref{almostdone}, \eqref{pnormbound} and \eqref{kappasum} we obtain
	\begin{align*}
	\mathcal{F}^-(D,v)\geq& \frac{1}{|D|}\int_{D^{\prime}}\overline W^{\beta}(\nabla v)\,\mathrm{d}x-\frac{C}{|D|}\left(\left(|D|+\eta+\|\nabla v\|_{L^p(D)}^p\right)\delta+\frac{N^p}{\delta^p}(\rho^p|D|+\eta)+\frac{|D|}{N}\right)
	\\
	&-C_{\beta}\frac{|\partial D+Q(0,\rho)|}{|D|}.
	\end{align*}
	The last inequality concludes the proof after letting first $\rho\to 0$, then $\eta\to 0$ followed by $N\to +\infty$ and $\delta\to 0$ and finally using the arbitrariness of $D^{\prime}\subset\subset D$ (recall the integrability of $\overline W^{\beta}(\nabla v)$).
\end{proof}


\section{Zero temperature limit of the elastic free energy: Proof of Theorem~\ref{th:smallT}}\label{Sec5}

In this section we investigate the asymptotic behavior of the rate functional from the large deviation principle when the temperature vanishes, or equivalently when $\beta\to +\infty$. To this end, we bound from above and below the entropic part whenever we consider the energy difference between a general configuration of the system and the ground state when we prescribe linear boundary conditions. We shall prove that, under the standard $p$-growth conditions \eqref{Hyp-1.1} and an additional local Lipschitz property (see Hypothesis~\ref{Hypo2}) we indeed recover the density of the $\Gamma$-limit of the rescaled versions of the Hamiltonians $H_{\e}(D,v)$. Since $\Gamma$-convergence focuses on the convergence of global minimizers of the Hamiltonian $H_{\e}(D,\cdot)$ (for a general reference on the subject we refer to the standard literature \cite{GCB,DM}), our result shows that at low temperatures entropic effects can be neglected and energy minimization is indeed meaningful also from a statistical physics point of view.

\subsection{Variational results neglecting temperature}
For completeness  we briefly recall  $\Gamma$-convergence results at zero temperature. First we rescale the Hamiltonian and its domain  as for the definition of the Gibbs measure. Given $\e>0$ and a function $u:\calL\to\R^n$ we define the function $v:\e\calL\to\R^n$ setting $v(\e x)=\e u(x)$. As usual this function can be identified with a function that is constant on the scaled Voronoi cells, so that  it belongs to the class
\begin{equation*}
\mathcal{PC}_{\e}:=\{v:\R^d\to\R^n:\;u_{|\e\mathcal{C}(x)}\text{ is constant for all }x\in\calL\}.
\end{equation*}
We may embed $\mathcal{PC}_{\e}\subset L^p(D,\R^n)$. Then, for every $O\in\Ard$, we introduce the rescaled Hamiltonian $\tilde{H}_{\e}(O,\cdot):L^p(D,\R^n)\to [0,+\infty]$ setting
\begin{equation*}
\tilde{H}_{\e}(v,O)=\begin{cases}
\displaystyle\frac{1}{|O_{\e}|}\sum_{\substack{(x,y)\in \mathbb{B}\\ \e x,\e y\in O}}f\left(x-y,\frac{v(\e x)-v(\e y)}{\e}\right) &\mbox{if $v\in\mathcal{PC}_{\e}$.}
\\
+\infty &\mbox{otherwise}
\end{cases}
\end{equation*}
We then define the set of clamped displacements
\begin{equation*}
\mathcal{BC}_{\e}(O,\varphi_\Lambda )=\{u:O_{\e}^\calL\to\mathbb{R}^n:\;u(x)=\Lambda x\text{ if }\dist(x,\partial O_{\e})\leq C_0\}.
\end{equation*}
Note that in contrast to the soft boundary conditions defining the set $\mathcal{B}_{\e}(O,\varphi_{\Lambda ^{\prime}})$ here the boundary conditions are exactly satisfied. The density of the $\Gamma$-limit is then given by the formula
\begin{equation*}
\overline{W}^{\infty}(\Lambda )=\lim_{\e\to 0}\frac{1}{|Q_{\e}|}\inf\{H_{\e}(u,Q):\;u\in\mathcal{BC}_{\e}(Q,\varphi_\Lambda )\},
\end{equation*}
where $Q=(-\frac{1}{2},\frac{1}{2})^d$. The  existence of this limit is a consequence of the subadditive ergodic Theorem,~as for  Proposition~\ref{existence}. By \cite[Theorems 2 \& 3]{ACG2} we have the following $\Gamma$-convergence result:
\begin{theorem}\label{Gamma-limit}
Assume \eqref{Hyp-1.1} and let $G$ be an admissible, stationary random Euclidean graph. Assume in addition that $f$ is continuous in the second variable. Then almost surely the functionals $\tilde{H}_{\e}$ $\Gamma$-converge with respect to the $L^p(D,\R^n)$-topology to the functional $\overline{H}:L^p(D,\R^n)\to [0,+\infty]$ finite only on $W^{1,p}(D,\R^n)$ and characterized by
\begin{equation*}
\overline{H}(v)=\frac{1}{|D|}\int_D \overline{W}^{\infty}(\nabla v(x))\,\mathrm{d}x.
\end{equation*}
Moreover, for any $O\in\Ard$ and any $v\in W^{1,p}(D,\R^n)$ we have the local version
\begin{equation*}
\Gamma\hbox{-}\lim_{\e\to 0}\tilde{H}_{\e}(O,v)=\frac{1}{|O|}\int_O\overline{W}^{\infty}(\nabla v(x))\,\mathrm{d}x.
\end{equation*}
The map $\Lambda\mapsto\overline{W}^{\infty}(\Lambda)$ is quasiconvex and satisfies the $p$-growth condition
\begin{equation*}
\frac{1}{C}|\Lambda|^p-C\leq \overline{W}^{\infty}(\Lambda)\leq C(|\Lambda|^p+1).
\end{equation*}
\end{theorem}
In order to also incorporate Dirichlet boundary conditions $\varphi\in \Lip(\R^d,\R^n)$, we introduce the class
\begin{equation*}
\mathcal{PC}_{\e,\varphi}=\{v\in\mathcal{PC}_{\e}:\;v(\e x)=\varphi(\e x)\text{ for all }x\in\Lw\text{ such that }\dist(\e x,\partial D)\leq C_0\e\}.
\end{equation*}
We restrict the domain of the discrete Hamiltonian $\tilde{H}_{\e}$ to $\mathcal{PC}_{\e,\varphi}$ setting $\tilde{H}_{\e,\varphi}:L^p(D,\R^n)\to [0,+\infty]$ as
\begin{equation*}
\tilde{H}_{\e,\varphi}(v)=
\begin{cases}
\tilde{H}_{\e}(v) &\mbox{if $v\in\mathcal{PC}_{\e,\varphi}$,}
\\
+\infty &\mbox{otherwise.}
\end{cases}
\end{equation*}
Then \cite[Theorem~4]{ACG2} yields the following $\Gamma$-convergence result under Dirichlet boundary conditions.
\begin{theorem}\label{gammabc}
Under the assumptions of Theorem~\ref{Gamma-limit}, the functionals $\tilde{H}_{\e,\varphi}$ $\Gamma$-converge with respect to the $L^p(D,\R^n)$-topology to the functional $\overline{H}_{\varphi}:L^p(D,\R^n)\to [0,+\infty]$ finite only for $\varphi+W_0^{1,p}(D,\R^n)$ and characterized by
\begin{equation*}
\overline{H}_{\varphi}(v)=\frac{1}{|D|}\int_D\overline W^{\infty}(\nabla v(x))\,\mathrm{d}x.
\end{equation*}	
\end{theorem}
\begin{remark}\label{convofmin}
From Lemma~\ref{compactness} and the fundamental property of $\Gamma$-convergence we deduce in particular the convergence of (almost-)minimizers to minimizers of the limit energy. Moreover it follows that
\begin{equation}\label{convofinf}
\lim_{\e\to 0}\Big(\inf_{v\in L^p(D,\R^n)}\tilde{H}_{\e,\varphi}(v)\Big)=\min_{v\in L^p(D,\R^n)}\overline{H}_{\varphi}(v).
\end{equation}	
\end{remark}

Before comparing the $\Gamma$-limit and the limit free energy we prove that one can replace the clamped boundary conditions by the soft version considered for the free energies and obtains the same limit. Note that in what follows both definitions will be used.
\begin{lemma}\label{b-bc}
Assume Hypothesis~\ref{Hypo1}. Fix $\Lambda \in\mathbb{R}^{n\times d}$. Then almost surely it holds that
\begin{equation*}
\overline{W}^{\infty}(\Lambda )=\lim_{\e\to 0}\frac{1}{|Q_{\e}|}\inf\{H_{\e}(u,Q):\;u\in\mathcal{B}_{\e}(Q,\varphi_\Lambda )\}.
\end{equation*}
\end{lemma}
\begin{proof}[Proof of Lemma~\ref{b-bc}]
The result is a special case of $\Gamma$-convergence. Indeed, consider the auxiliary functional $H_{\e,\Lambda }:L^p(Q,\R^n)\to [0,+\infty]$ defined by 
\begin{equation*}
H_{\e,\Lambda }(v)=\begin{cases}
\tilde{H}_{\e}(Q,v) &\mbox{if $v\in\mathcal{PC}_{\e}$ and $\Pi_{1/\e}v\in\mathcal{B}_{\e}(Q,\varphi_\Lambda )$,}
\\
+\infty &\mbox{otherwise.}
\end{cases}
\end{equation*}
Due to Lemma~\ref{compactness} we know that the $\Gamma$-limit of $H_{\e,\Lambda }$ can be finite only for $v\in \overline{\varphi}_\Lambda +W_0^{1,p}(Q,\R^n)$. From Theorems \ref{Gamma-limit} and \ref{gammabc} applied with $D=Q$ and $\varphi=\overline{\varphi}_\Lambda $, for any such $v$ we deduce from monotonicity that
\begin{align*}
\int_Q\overline W^{\infty}(\nabla v(x))\,\mathrm{d}x&\leq \Gamma\hbox{-}\liminf_{\e\to 0} H_{\e,\Lambda }(v)\leq  \Gamma\hbox{-}\limsup_{\e\to 0} H_{\e,\Lambda }(u)  
\\
&\leq\Gamma\hbox{-}\limsup_{\e\to 0} \tilde{H}_{\e,\overline{\varphi}_\Lambda }(v)\leq \int_Q\overline W^{\infty}(\nabla v(x))\,\mathrm{d}x.
\end{align*}
Using \eqref{convofinf}, which holds by the same arguments for the functionals $H_{\e,\Lambda }$, the result follows by quasiconvexity of $\Lambda \mapsto \overline{W}^{\infty}(\Lambda )$ and a rescaling since $\frac{1}{|Q_{\e}|}H_{\e}(Q,\Pi_{1/\e}(\cdot))=\tilde{H}_{\e}(Q,\cdot)$.
\end{proof}

\subsection{Quantitative comparison between $\overline{W}^{\beta}$ and $\overline{W}^{\infty}$}
As announced earlier, in this section we replace Hypothesis~\ref{Hypo1} by the stronger assumptions of Hypothesis~\ref{Hypo2}. In the following two lemmata we establish a quantitative estimate between $\overline{W}^{\beta}(\Lambda)$ and $\overline{W}^{\infty}(\Lambda)$ in the regime $\beta\gtrsim 1$. Note that the additional properties of Hypothesis~\ref{Hypo2} are only needed for Lemma~\ref{tempub} below.


\begin{lemma}\label{tempub}
Assume Hypothesis~\ref{Hypo2}. Then for every $\Lambda \in\R^{n\times d}$ there exists a constant $0<C_\Lambda \leq C(1+|\Lambda |^{p-1})$ such that, for all $\beta\geq \exp(1)$,
\begin{equation*}
\overline W^{\beta}(\Lambda)-W^{\infty}_{{\rm hom}}(\Lambda )\leq C_\Lambda \frac{\log(\beta)}{\beta}.
\end{equation*} 	
\end{lemma}
\begin{proof}[Proof of Lemma~\ref{tempub}]
Due to Proposition~\ref{existence} we can consider the free energy on the unit cube $Q$. First note that due to the assumptions, the discrete energy $H_{\e}(Q,\cdot)$ is equicoercive and continuous on the closed set $\mathcal{BC}_{\e}(Q,\varphi_\Lambda )$. Hence the minimum is attained and we denote by $\hat{u}_{\e}$ a minimizer. By testing the function $\varphi_\Lambda $ and using the $p$-growth conditions of Hypothesis~\ref{Hypo2}, we obtain the a priori bound
\begin{equation}\label{aprioribound}
\|\nabla_{\B}\hat{u}_{\e}\|^p_{\ell^p_{\e}(Q)}\leq CH_{\e}(\hat{u}_{\e})+C|Q^{\calL}_{\e}|\leq C(1+|\Lambda |^p)|Q^{\calL}_{\e}|.
\end{equation}
According to the continuity property in Hypothesis~\ref{Hypo2}, for any $u\in\mathcal{B}_{\e}(Q,\varphi_\Lambda )$ we have the estimate
\begin{align*}
H_{\e}(Q,u)-H_{\e}(Q,\hat{u}_{\e})&\leq \sum_{\substack{(x,y)\in \mathbb{B}\\ x,y\in Q_{\e}}}|f(x-y,u(x)-u(y))-f(x-y,\hat{u}_{\e}(x)-\hat{u}_{\e}(y))|
\\
&\leq C\sum_{\substack{(x,y)\in \mathbb{B}\\ x,y\in Q_{\e}}}(1+|u(x)-u(y)|^{p-1}+|\hat{u}_{\e}(x)-\hat{u}_{\e}(y)|^{p-1})|(u-\hat{u}_{\e})(x)-(u-\hat{u}_{\e})(y)|
\\
&\leq C\big(|Q^{\calL}_{\e}|^{\frac{p-1}{p}}+\|\nabla_{\B}(u-\hat{u}_{\e})\|^{p-1}_{\ell^p_{\e}(Q)}+\|\nabla_{\B}\hat{u}_{\e}\|^{p-1}_{\ell^p_{\e}(Q)}\big) \|\nabla_{\B}(u-\hat{u}_{\e})\|_{\ell^p_{\e}(Q)}
\\
&\leq C|Q^{\calL}_{\e}|^{\frac{p-1}{p}}(1+|\Lambda |^{p-1})\|\nabla_{\B}(u-\hat{u}_{\e})\|_{\ell^p_{\e}(Q)}+C\|\nabla_{\B}(u-\hat{u}_{\e})\|^p_{\ell^p_{\e}(Q)}
\\
&\leq C|Q^{\calL}_{\e}|^{\frac{p-1}{p}}(1+|\Lambda |^{p-1})\|u-\hat{u}_{\e}\|_{\ell^p_{\e}(Q)}+C\|u-\hat{u}_{\e}\|^p_{\ell^p_{\e}(Q)},
\end{align*}
where we have used H\"{o}lder's inequality, \eqref{aprioribound} and Remark~\ref{reverse}.
From the change of variables $u\mapsto u-\hat{u}_{\e}$, which maps one-to-one from $\mathcal{B}_{\e}(Q,\varphi_{\Lambda})$ to $\mathcal{B}_{\e}(Q,0)$, we infer that the discrete error can be bounded by
\begin{align}\label{tempest1}
e_{\e,\beta}:&=\mathcal{E}^{\beta}_{\e}(Q,\overline{\varphi}_{\Lambda} )-\frac{1}{|Q_{\e}|}H_{\e}(\hat{u}_{\e})\nonumber
\\
&=-\frac{1}{\beta|Q_{\e}|}\log\Big(\int_{\mathcal{B}_{\e}(Q,\varphi_\Lambda )}\exp\big(-\beta( H_{\e}(Q,u)-H_{\e}(Q,\hat{u}_{\e})\big)\big)\,\mathrm{d}u\Big)\nonumber
\\
&\leq -\frac{1}{\beta|Q_{\e}|}\log\Big(\int_{\mathcal{B}_{\e}(Q,0)}\exp\big(-\beta C(|Q_{\e}^{\calL}|^{\frac{p-1}{p}}(1+|\Lambda |^{p-1})\|u\|_{\ell^p_{\e}(Q)}+\|u\|^p_{\ell^p_{\e}(Q)})\big)\,\mathrm{d}u\Big).
\end{align}
In this last integral we aim to get rid of the boundary conditions. To this end, let us write the domain of integration as a product and implicitly define the numbers $d_{\e,i}$ and $d_{\e,b}$ via
\begin{equation*} 
\mathcal{B}_{\e}(Q,0)=\Bigg(\prod_{\substack{x\in Q^{\calL}_{\e}\\ \dist(x,\partial Q_{\e})>C_0}}\R^n\Bigg)\times\Bigg( \prod_{\substack{x\in Q^{\calL}_{\e}\\ \dist(x,\partial Q_{\e})\leq C_0}}B_1(0)\Bigg)= (\R^n)^{d_{\e,{\rm int}}}\times \left(B_1(0)\right)^{d_{\e,{\rm bd}}}.
\end{equation*}
With a slight abuse of notation, we write any $u\in\mathcal{B}_{\e}(Q,0)$ as a sum via $u=u_1+u_2$, where $u_1\in\mathcal{BC}_{\e}(Q,0)$ and $|u_2(x)|\leq 1$ with support contained in $\{x\in Q_{\e}:\;\dist(x,\partial Q_{\e})\leq C_0\}$. Interpreting a deformation as a large vector $u\in\R^{n|Q^{\calL}_{\e}|}$ we denote its standard $p$-norm by $|u|_p$. As on $\R^n$ all norms are equivalent, it holds that $\|u\|_{\ell^p_{\e}(Q)}\leq C|u|_p$. By the triangle inequality and the structure of $u_2$ we get
\begin{align*}
|Q^{\calL}_{\e}|^{\frac{p-1}{p}}(1+|\Lambda |^{p-1})\|u\|_{\ell^p_{\e}(Q)}+\|u\|^p_{\ell^p_{\e}(Q)}&\leq C\left(|Q^{\calL}_{\e}|^{\frac{p-1}{p}}(1+|\Lambda |^{p-1})\left(|u_1+u_2|_p\right)+(|u_1+u_2|^p_p)\right)
\\
&\leq C\Big(|Q^{\calL}_{\e}|^{\frac{p-1}{p}}(1+|\Lambda |^{p-1})(|u_1|_p+(d_{\e,{\rm bd}})^{\frac{1}{p}})+|u_1|_p^p+d_{\e,{\rm bd}}\Big).
\end{align*}
Using Fubini's Theorem~we can factorize the integral and therefore \eqref{tempest1} yields
\begin{align*}
e_{\e,\beta}\leq& -\frac{1}{\beta|Q_{\e}|}\log\Big(\int_{\R^{nd_{\e,{\rm int}}}}\exp\big(-\beta C(|Q^{\calL}_{\e}|^{\frac{p-1}{p}}(1+|\Lambda |^{p-1})|u_1|_p+|u_1|_p^p)\big)\mathrm{d}u_1\Big)\\
&-\frac{1}{\beta|Q_{\e}|}\log\Big(|B_1(0)|^{d_{\e,{\rm bd}}}\exp\big(-\beta C((1+|\Lambda |^{p-1})|Q^{\calL}_{\e}|^{\frac{p-1}{p}}d_{\e,{\rm bd}}^{\frac{1}{p}}+ d_{\e,{\rm bd}})\big)\Big)=:e_{\e,\beta}^{\rm int}+e_{\e,\beta}^{\rm bd}.
\end{align*}
We first argue that $e_{\e,\beta}^{\rm bd}$ vanishes when $\e\to 0$. Indeed, as $d_{\e,{\rm bd}}\leq C\e^{1-d}$ by \eqref{surfaceesimtate} and the Lipschitz regularity of $\partial Q$, for $\e$ small enough it holds that 
\begin{equation}\label{boundaryvanishes}
|e_{\e,\beta}^{\rm bd}|\leq (C\beta^{-1}+C)\e+C(1+|\Lambda |^{p-1})\e^{\frac{1}{p}}.
\end{equation}

To treat the contribution of $e_{\e,\beta}^{\rm int}$, we make use of the coarea formula. Therefore we consider the Lipschitz-continuous function $f_{\e}:\R^{nd_{\e,{\rm int}}}\to [0,+\infty)$ defined by $y\mapsto f_{\e}(y)=|Q^{\calL}_{\e}|^{-\frac{1}{p}}|y|_p$. For $y\neq 0$ it is differentiable and, since $f_{\e}$ is $|Q^{\calL}_{\e}|^{-\frac{1}{p}}$-Lipschitz with respect to the $p$-norm, for $\e$ small enough we have the rough estimate
\begin{align}\label{derivativebound}
|\nabla f_{\e}(y)|_2&=\sup_{|x|_2=1}\langle \nabla f_{\e}(y),x\rangle=\sup_{|x|_2=1}\lim_{t\to 0}\frac{f_{\e}(y+tx)-f_{\e}(y)}{t}\nonumber
\\
&\leq \sup_{|x|_2=1}|Q^{\calL}_{\e}|^{-\frac{1}{p}}|x|_p\leq |Q^{\calL}_{\e}|^{-\frac{1}{p}}\max\{1,(nd_{\e,{\rm int}})^{\frac{1}{p}-\frac{1}{2}}\}\leq 1.
\end{align}
Using \eqref{derivativebound}, we deduce from the coarea formula that for arbitrary $t_*>0$
\begin{align}\label{tempest2}
e^{\rm int}_{\e,\beta}&\leq -\frac{1}{\beta|Q_{\e}|}\log\Big(\int_{\R^{nd_{\e,{\rm int}}}}|\nabla f_{\e}(u_1)|_2\exp\big(-\beta C(|Q^{\calL}_{\e}|(1+|\Lambda |^{p-1})f_{\e}(u_1)+|Q^{\calL}_{\e}|f_{\e}(u_1)^p)\big)\,\mathrm{d}u_1\Big)\nonumber
\\
&=-\frac{1}{\beta|Q_{\e}|}\log\Big(\int_0^{\infty}\mathcal{H}^{nd_{\e,{\rm int}}-1}(\{f_{\e}=t\})\exp\big(-\beta C|Q^{\calL}_{\e}|((1+|\Lambda |^{p-1})t+t^p)\big)\,\mathrm{d}t\Big)\nonumber
\\
&\leq-\frac{1}{\beta|Q_{\e}|}\log\Big(\int_0^{t_{*}}\mathcal{H}^{nd_{\e,{\rm int}}-1}(\{|y|_p=|Q^{\calL}_{\e}|^{\frac{1}{p}}t\})\exp\big(-\beta C|Q^{\calL}_{\e}|(|(1+|\Lambda |^{p-1})t+t^p)\big)\,\mathrm{d}t\Big),
\end{align} 
We next bound from below the surface measure inside the integral. To this end, we make the restriction $t_*\leq 1$. By Lemma~\ref{surfacevolume} and the scaling properties of the Hausdorff measure, for some small constant $c=c(n,p)$ we have the lower bound
\begin{align*}
\mathcal{H}^{nd_{\e,{\rm int}}-1}(\{|y|_p=|Q^{\calL}_{\e}|^{\frac{1}{p}}t\}) &
\geq (|Q^{\calL}_{\e}|^{\frac{1}{p}}t)^{nd_{\e,{\rm int}}-1}\left(\frac{c_p}{nd_{\e,{\rm int}}}\right)^{\frac{nd_{\e,{\rm int}}}{p}}
\\
&\geq |Q^{\calL}_{\e}|^{-\frac{1}{p}}(ct)^{nd_{\e,{\rm int}}}\left(\frac{|Q^{\calL}_{\e}|}{d_{\e,{\rm int}}}\right)^{\frac{nd_{\e,{\rm int}}}{p}}
\geq |Q^{\calL}_{\e}|^{-\frac{1}{p}}(ct)^{n|Q^{\calL}_{\e}|},
\end{align*}
where we used that $t\leq 1$. Plugging this bound into (\ref{tempest2}), for any $t_*\leq 1$ we can further estimate
\begin{align}\label{tempest3}
e^{\rm int}_{\e,\beta}\leq& -\frac{1}{\beta|Q_{\e}|}\log\bigg(\int_0^{t_{*}}\exp\Big(|Q^{\calL}_{\e}|\big(n\log(ct)-\beta C((1+|\Lambda |^{p-1})t+t^p)\big)\Big)\,\mathrm{d}t\bigg)\nonumber
\\
&+\frac{1}{p\beta|Q_{\e}|}\log(|Q^{\calL}_{\e}|).
\end{align} 
We now choose an appropriate $t_*\leq 1$. More precisely, we try to find $t_*$ and $\overline{C}=\overline{C}(\Lambda ,n,p)$ such that for all $t\leq t_*$
\begin{equation*}
n\log(ct)-\beta C((1+|\Lambda |^{p-1})t+t^p)\geq \overline{C}\log(t).
\end{equation*}
To this end, first observe that the function $t\mapsto (n-\overline{C})\log(ct)-\beta C((1+|\Lambda |^{p-1})t+t^p)$ is decreasing whenever $\overline{C}\geq n$. Moreover, if we set $\overline{C}=n+C(2+|\Lambda |^{p-1})$ and $t_*=\frac{1}{\beta}$, then for $\beta>\exp(1)$ and $c\leq 1$ we have
\begin{align*}
(n-\overline{C})\log(ct_*)-\beta C((1+|\Lambda |^{p-1})t_*+t_*^p) &\geq (n-\overline{C})(\log(t_*)-\beta C(2+|\Lambda |^{p-1})t_*
\\
&\geq C(\log(\beta)-1)(2+|\Lambda |^{p-1})> 0.
\end{align*}
Thus with our choice of $t_*$ and $\overline{C}$ we infer from (\ref{tempest3}) that
\begin{align*}
e^{\rm int}_{\e,\beta}&\leq -\frac{1}{\beta|Q_{\e}|}\log\Big(\int_0^{t_*}t^{\overline{C}|Q^{\calL}_{\e}|}\,\mathrm{d}t\Big)+\frac{1}{p\beta|Q_{\e}|}\log(|Q^{\calL}_{\e}|)
=-\frac{1}{\beta|Q_{\e}|}\log\Big(\frac{t_*^{\overline{C}|Q^{\calL}_{\e}|+1}}{\overline{C}|Q^{\calL}_{\e}|+1}\Big)+\frac{1}{p\beta|Q_{\e}|}\log(|Q_{\e}|)
\end{align*}
and we can conclude from \eqref{boundaryvanishes} and \eqref{e.volVoronoi} that
\begin{equation*}
\overline W(\Lambda ,\beta)-\overline W^{\infty}(\Lambda ))=\lim_{\e\to 0}e_{\e,\beta}\leq \limsup_{\e\to 0}\frac{\overline{C}|Q^{\calL}_{\e}|+1}{\beta|Q_{\e}|}|\log(t_*)|\leq \overline{C}\left(\frac{2R}{r}\right)^d\frac{\log(\beta)}{\beta}.
\end{equation*}
This proves the claim by our definition of $\overline{C}$.
\end{proof}

\begin{lemma}\label{templb}
Assume Hypothesis~\ref{Hypo1}. Then for every $\Lambda \in\R^{n\times d}$ there exists a constant $0<C_\Lambda \leq C(1+\log(1+|\Lambda |))$ such that, for all $\beta\geq 1$,
\begin{equation*}
\overline W^{\beta}(\Lambda)-\overline W^{\infty}(\Lambda )\geq -\frac{C_\Lambda }{\beta}.
\end{equation*} 	
\end{lemma}
\begin{proof}[Proof of Lemma~\ref{templb}]
Again we compute the energy densities with respect to the unit cube $Q$. Having in mind Lemma~\ref{b-bc}, we let  $\tilde{u}_{\e}$ be a minimizer of the Hamiltonian $H_{\e}$ on the set $\mathcal{B}_{\e}(Q,\varphi_\Lambda )$. Note that we assume without loss of generality that a minimizer exists, otherwise we could take an almost minimizer with an energy close to the infimum at a rate that vanishes much faster than $\e^{d}$. For any $u\in\mathcal{B}_{\e}(Q,\varphi_\Lambda )$, by the $p$-growth condition in Hypothesis~\ref{Hypo2} and \eqref{aprioribound}, we have the inequality
\begin{equation*}
H_{\e}(Q,u)-H_{\e}(Q,\tilde{u}_{\e})\geq \frac{1}{C}\|\nabla_{\B}u\|^p_{\ell^p_{\e}(Q)}-C(1+|\Lambda |^p)|Q^{\calL}_{\e}|\geq \frac{1}{C}\|\nabla_{\B}(u-\varphi_\Lambda )\|^p_{\ell^p_{\e}(Q)}-C(1+|\Lambda |^p)|Q^{\calL}_{\e}|.
\end{equation*}
While this estimate turns out to be useful for deformations with large energy, we also need a suitable lower bound for deformations with small energy. To this end we observe that by minimality
$H_{\e}(Q,u)-H_{\e}(Q,\tilde{u}_{\e})\geq 0$, so that we can write
\begin{equation}\label{softbound}
H_{\e}(Q,u)-H_{\e}(Q,\tilde{u}_{\e})\geq\max\left\{0,\frac{1}{C}\|\nabla_{\B}(u-\varphi_\Lambda )\|^p_{\ell^p_{\e}(Q)}-C(1+|\Lambda |^p)|Q^{\calL}_{\e}|\right\}.
\end{equation}
This inequality motivates the partition $\mathcal{B}_{\e}(Q,0)=B_{1,\e}\cup B_{2,\e}$, where
\begin{equation*}
\begin{split}
& B_{1,\e}:=\Big\{\varphi\in \mathcal{B}_{\e}(Q,0):\;\|\nabla_{\B}u\|^p_{\ell^p_{\e}(Q)}\leq 2C^2(1+|\Lambda |^p)|Q^{\calL}_{\e}|\Big\},
\\
& B_{2,\e}:=\mathcal{B}_{\e}(Q,0)\setminus B_{1,\e}.
\end{split}
\end{equation*}
With the change of variables $u\mapsto u-\varphi_\Lambda $ and \eqref{softbound} we then obtain
\begin{align}
e^1_{\e,\beta}:&=\mathcal{E}^{\beta}_{\e}(\overline{\varphi}_{\Lambda} )-\frac{1}{|Q_{\e}|}H_{\e}(Q,\tilde{u}_{\e})
=-\frac{1}{\beta|Q_{\e}|}\log\Big(\int_{\mathcal{B}_{\e}(Q,\varphi_\Lambda )}\exp\big(-\beta( H_{\e}(Q,u)-H_{\e}(Q,\tilde{u}_{\e})\big)\big)\,\mathrm{d}u\Big)\nonumber
\\
&\geq -\frac{1}{\beta|Q_{\e}|}\log\Big(|\mathcal{B}_{1,\e}|+\int_{\mathcal{B}_{\e}(Q,0)}\exp\big(-\frac{\beta}{2C} \|\nabla_{\B}u\|^p_{\ell^p_{\e}(Q)}\big)\,\mathrm{d}u\Big).\label{errorlb}
\end{align}
We treat the two terms inside the logarithm separately. Let us start with the integral. Using Lemma~\ref{forrestargument}, for $\e$ small enough (independent of $\beta$) and $\beta\geq 1$, we obtain the bound
\begin{align}\label{largegradlb}
\int_{\mathcal{B}_{\e}(Q,0)}\exp\big(-\frac{\beta}{2C} \|\nabla_{\B}u\|^p_{\ell^p_{\e}(Q)}\big)\,\mathrm{d}u&\leq C^{N_{Q,\e}}\left(\left(\frac{\beta}{2C}\right)^{-\frac{n}{p}}C\right)^{|Q^{\calL}_{\e}|-N_{Q,\e}}\leq \Big(C(p,n)\Big)^{n|Q^{\calL}_{\e}|}.
\end{align}
In order to provide a bound for the measure of $\mathcal{B}_{1,\e}$ we first enlarge the set and then perform a suitable change of variables. To this end we number the vertices by the following algorithm: For every connected component $G_j=(V_j,\mathbb{B}_j)$ of the graph $G_{Q_{\e}}$ we choose a minimal spanning tree $ST_j=(V_j,\mathbb{B}^{\prime}_j)$ and a vertex where the boundary conditions are active (see Remark~\ref{numbercomponents} for the existence of such a vertex). To this vertex we assign the number $k_j:=\Big(\sum_{i<j}|V_i|\Big)+1$. Then we start any path in the spanning tree and number the vertices consecutively until we cannot go on. If we have numbered all vertices of the connected component we go to the next one, otherwise we continue at the first (with respect to the numbering) vertex with multiple path possibilities and continue with the same procedure. Since we consider a minimal spanning tree, every vertex gets assigned a unique number. Moreover, for each vertex number $l\backslash\{k_j\}$, we find a number $l^{\prime}<l$ such that $(x_{l^\prime},x_l)\in \mathbb{B}^{\prime}_j$. Then we have the following set inclusion:
\begin{equation*}
\mathcal{B}_{1,\e}\subset \Big\{u\in (\R^n)^{|Q^{\calL}_{\e}|}:|u_{k_j}|<1\text{ for all }j\text{ and }\sum_{j}\sum_{l=k_j+1}^{k_{j+1}-1}|u_{l}-u_{l^{\prime}}|^p\leq 2C^2(1+|\Lambda |^p)|Q_{\e}^{\calL}|\Big\}=:\mathcal{U}_{1,\e}
\end{equation*}
We now define a linear transformation on $\mathcal{U}_{\e,1}$ setting $T:\mathcal{U}_{1,\e}\to(\mathbb{R}^n)^{|Q^{\calL}_{\e}|}$ as
\begin{equation*}
(T\varphi)_l=
\begin{cases}
\varphi_l &\mbox{if $l=k_j$ for some $j$,}\\
\varphi_l-\varphi_{l^{\prime}} &\mbox{otherwise.}
\end{cases}
\end{equation*}
Note that the mapping $l\mapsto l^{\prime}$ is independent of $\varphi$, so that $T$ is indeed linear. Moreover, it is straightforward to check that $T$ is injective and thus a diffeomorphism onto its image allowing to perform a change of variables. Observe that its derivative $DT$ admits a lower triangle matrix representation since the $l^{th}$ component of $T\varphi$ depends only on entries with smaller index. On the diagonal we have all entries equal to $1$. Hence it holds that $\det(DT)=1$. By a change of variables we conclude that
\begin{equation*}
|\mathcal{B}_{\e,1}|\leq |\mathcal{U}_{\e,1}|=|T(\mathcal{U}_{\e,1})|
\end{equation*}
and by construction it holds that
\begin{align*}
T(\mathcal{U}_{\e,1})&=\Bigg(\prod_{j=1}^{N_{Q,\e}}B_1(0)\Bigg)\times \Big\{\varphi\in (\R^n)^{|Q_{\e}|-N_{Q,\e}}:\;\|\varphi\|^p_p\leq 2C^2(1+|\Lambda |^p)|Q^{\calL}_{\e}|\Big\}
\\
&\subset\Big\{u\in \R^{n|Q^{\calL}_{\e}|}:\;|u|_p\leq C^{\prime}(1+|\Lambda |^p)^{\frac{1}{p}}|Q^{\calL}_{\e}|^{\frac{1}{p}}\Big\},
\end{align*}
where the larger constant $C^{\prime}$ contains a factor derived from the equivalence of norms on $\R^n$. The last set is a high-dimensional ball with respect to the corresponding $\ell^p$-norm, for which there exist exact formulas for the volume. Denoting (just in this proof) by $\Gamma$ Euler's Gamma-function we deduce that, for $\e$ small enough,
\begin{equation}\label{volumebound}
|\mathcal{B}_{\e,1}|\leq\frac{ \Big(C(p)(1+|\Lambda |^p)|Q^{\calL}_{\e}|\Big)^{\frac{n|Q_{\e}^{\calL}|}{p}}}{\Gamma\Big(\frac{n|Q^{\calL}_{\e}|}{p}+1\Big)}\leq \Big(C(p,n)(1+|\Lambda |)\Big)^{n|Q^{\calL}_{\e}|},
\end{equation}
where we used the lower bound $\Gamma(z+1)\geq (z/e)^z$ for all $z\geq 1$. Combining \eqref{errorlb}, \eqref{largegradlb} and \eqref{volumebound} we infer that, for $\e$ small enough (but independent of $\beta$) and $\beta\geq 1$,
\begin{equation*}
e^1_{\e,\beta}\geq -\frac{1}{\beta |Q_{\e}|}\log\Big((C(p,n)(1+|\Lambda |))^{n|Q^{\calL}_{\e}|}\Big)=-C(p,n)\left(\frac{2R}{r}\right)^d\frac{(1+\log(1+|\Lambda |)}{\beta}.
\end{equation*}
Thanks due Lemma~\ref{b-bc} and the definition of $e^1_{\e,\beta}$ the claim now follows after letting $\e\to 0$.
\end{proof}

\subsection{$\Gamma$-convergence of the LDP rate functionals}
The estimates proved in Lemmata \ref{tempub} and \ref{templb} lead to Theorem~\ref{th:smallT} that also relates the support of the limits of Gibbs measures (see Corollary \ref{concentration}) to the minimizers of the $\Gamma$-limit at small temperatures.
\begin{proof}[Proof of Theorem~\ref{th:smallT}]
We let $G\in\mathcal{G}^{\prime}$, where $\mathcal{G}^{\prime}$ is given by Proposition~\ref{fullexistence} (see also Remark~\ref{alltemp}). Fix an arbitrary sequence $\beta_j\to +\infty$. For the moment we consider the functionals $F_j,F:L^p(D,\R^n)\to \R\cup\{+\infty\}$ finite only on $\varphi+W^{1,p}_0(D,\R^n)$ and characterized by
\begin{equation*}
F_j(v)=\fint_D\overline W^{\beta_j}(\nabla v)\dx,\quad\quad F(v)=\fint_D \overline{W}^{\infty}(\nabla v)\dx.
\end{equation*}
By Lemmata \ref{tempub} and \ref{templb} we have that $F_j\to F$ pointwise when $j\to +\infty$. Hence for all $v\in L^p(D,\R^n)$ 
\begin{equation*}
\Gamma\hbox{-}\limsup_{j\to +\infty}F_j(v)\leq F(v).
\end{equation*} In order to prove the $\liminf$-inequality, consider $v\in L^p(D,\R^n)$ and a sequence $(v_j)\subset L^p(D,\R^n)$ such that $v_j\to v$ in $L^p(D,\R^n)$ and $\sup_j F_j(u_j)<+\infty$. Since Lemma~\ref{templb} yields
\begin{equation*}
\overline W^{\beta_j}(\Lambda)-W^{\infty}_{\rm hom}(\Lambda )\geq -\frac{C}{\beta}(1+\log(1+|\Lambda |),
\end{equation*}
where the constant $C$ is independent of $\Lambda $ and for $j$ large enough it holds that $\overline W^{\beta_j}(\Lambda )\geq\frac{1}{C}|\Lambda |^p-C$, we infer that $v\in\varphi+W_0^{1,p}(D,\R^n)$ and
\begin{equation*}
\liminf_{j\to +\infty}F_j(v_j)\geq \liminf_{j\to +\infty}F(v_j)\geq F(v),
\end{equation*} 
where we used that $F$ is lower semicontinuous due to the quasiconvexity of the map $\Lambda \mapsto \overline W^{\infty}(\Lambda )$. is quasiconvex, the lower bound follows from weak lower semicontinuity. Thus $F_j$ $\Gamma$-converges to $F$ with respect to the $L^p(D<\R^n)$-topology. Since the $\Gamma$-convergence implies the convergence of the infimum values $\lim_j\inf_v F_j(v)= \inf_v F(v)$, Theorem~\ref{th:smallT} is proven.
\end{proof}

\subsection{The phantom model}
The convergence result proved in this section can be made much more precise when the discrete Hamiltonian is quadratic, that is,
\begin{equation}\label{quadratic}
H_{\e}(O,u)=\sum_{\substack{(x,y)\in \mathbb{B}\\ x,y\in O_{\e}}}\langle u(x)-u(y),A(x-y)(u(x)-u(y))\rangle
\end{equation}
with a function $A:\R^d\to\mathbb{M}_{\rm sym}^{d\times d}$ uniformly positive definite and bounded on $B_{C_0}(0)$, where $C_0$ is the maximal range of interactions given by Definition \ref{defadmissible}. 
The phantom model (see e.g.~\cite[Section~7.2.2]{Rub}), which is an approximation of rubber elasticity of polymer-chain networks at 
small deformation and finite temperature, indeed corresponds to the 
homogenization of this energy density (using a self-consistent approach rather than the usual cell-formula).
As the following shows, the limit free energy agrees with the density of the $\Gamma$-limit up to an additive constant which depends on $\beta$ but not on $\Lambda$.
In particular, this justifies the use of the self-consistent approach in \cite[Section~7.2.2]{Rub} even at finite temperature.
\begin{corollary}\label{cor.quad}
Assume that $H_{\e}$ is given by \eqref{quadratic}. Then it holds that 
\begin{equation*}
\overline W^{\beta}(\Lambda)=\overline W^{\infty}(\Lambda )+\overline W^{\beta}(0).
\end{equation*}
In particular the function $\Lambda \mapsto \overline W^{\beta}(\Lambda )$ is uniformly convex and quadratic.	
\end{corollary}
\begin{proof}[Proof of Corollary \ref{cor.quad}]
We consider the free energy on the unit cube $Q$ and first find a unique minimizer of $u\mapsto H_{\e}(Q,u)$ on the set $\mathcal{BC}_{\e}(Q,\varphi_\Lambda )$, that we denote by $\hat{u}_{\e}$. We extend it to $\calL$ setting $\hat{u}_{\e}(x)=\varphi_{\Lambda}(x)$ for all $x\in\calL\backslash Q_{\e}$. Given $u\in\mathcal{B}_{\e}(Q,\varphi_\Lambda )$, we decompose it as $u=u_1+u_2$ with $u_1\in\mathcal{BC}_{\e}(Q,\varphi_\Lambda )$ and $u_2:Q_{\e}\to\R^n$ satisfying $|u_2(x)|\leq 1$ for all $x\in\calL$ and $u_2(x)=0$ for all $x\in\Lw$ such that $\dist(x,\partial Q_{\e})>C_0$. By the quadratic structure we have
\begin{align}\label{quadstructure}
H_{\e}(Q,u)-H_{\e}(Q,\hat{u}_{\e})=&H_{\e}(Q,u-\hat{u}_{\e})+2\sum_{\substack{(x,y)\in \mathbb{B}\\ x,y\in Q_{\e}}}\Big\langle(u-\hat{u}_{\e})(x)-(u-\hat{u}_{\e})(y),A(x-y)(\hat{u}_{\e}(x)-\hat{u}_{\e}(y))\Big\rangle\nonumber
\\
=&H_{\e}(Q,u-\hat{u}_{\e})+2\sum_{\substack{(x,y)\in \mathbb{B}\\ x,y\in Q_{\e}}}\Big\langle u_2(x)-u_2(y),A(x-y)(\hat{u}_{\e}(x)-\hat{u}_{\e}(y))\Big\rangle,
\end{align}
where we used pointwise symmetry of $A(z)$, the weak Euler-Lagrange equation satisfied by $\hat{u}_{\e}$ and that $u_1-\hat{u}_{\e}\in\mathcal{BC}_{\e}(Q,0)$ is an admissible test function for this equation. In order to bound the last term, we first introduce the set $\partial_{\e}Q=\{z\in Q_{\e}:\;\dist(z,\partial Q_{\e})\leq 2C_0\}$. Then by the properties of $u_2$ and boundedness of $A(z)$ we have
\begin{equation*}
\Big|2\sum_{\substack{(x,y)\in \mathbb{B}\\ x,y\in Q_{\e}}}\Big\langle u_2(x)-u_2(y),A(x-y)(\hat{u}_{\e}(x)-\hat{u}_{\e}(y))\Big\rangle\Big|\leq C\sum_{\substack{(x,y)\in \mathbb{B}\\ x,y\in \partial_{\e}Q}}c_0|\hat{u}_{\e}(x)-\hat{u}_{\e}(y)|,
\end{equation*} 
where $c_0$ is a lower bound for the smallest eigenvalue of $A(z)$ for all $z\in B_{C_0}(0)$. Note that the right hand side does not depend on $u$ any more. Since the change of variables $u\mapsto u-\hat{u}_{\e}$ maps $\mathcal{B}_{\e}(Q,\varphi_{\Lambda})$ one-to-one to $\mathcal{B}_{\e}(Q,0)$, we conclude by the very definition of the terms $\overline W^{\beta}(\Lambda)$ and $\overline W^{\infty}(\Lambda)$ and equation \eqref{quadstructure} that
\begin{align*}
|\overline W^{\beta}(\Lambda)-\overline W^{\infty}(\Lambda )-\overline W^{\beta}(0)|\leq C\limsup_{\e\to 0}\frac{1}{|Q_{\e}|}\sum_{\substack{(x,y)\in \mathbb{B}\\ x,y\in \partial_{\e}Q}}c_0|\hat{u}_{\e}(x)-\hat{u}_{\e}(y)|.
\end{align*}
It remains to prove that the right hand side is zero. To this end, we note that due to Lemma~\ref{compactness} we can assume that $v_{\e}\in\mathcal{PC}_{\e}$ defined by $v_{\e}(\e x)=\e\hat{u}_{\e}(x)$ converges in $L^p(D,\R^n)$ to some function $v\in \overline{\varphi}_\Lambda +W^{1,p}_0(Q,\R^n)$ (actually one can prove that $v=\overline{\varphi}_\Lambda $, but this is not needed here). Given $\delta>0$, Jensen's inequality and (quasi)convexity of $\Lambda \mapsto \overline{W}^{\infty}(\Lambda )$ imply
\begin{align*}
\limsup_{\e\to 0}&\Big(\frac{1}{|Q_{\e}|}\sum_{\substack{(x,y)\in \mathbb{B}\\ x,y\in \partial_{\e}Q}}c_0|\hat{u}_{\e}(x)-\hat{u}_{\e}(y)|\Big)^2\ll \limsup_{\e\to 0}\frac{1}{|Q_{\e}|}\sum_{\substack{(x,y)\in E\\ x,y\in \partial^{\e}Q}}c_0|\hat{u}_{\e}(x)-\hat{u}_{\e}(y)|^2
\\
&\leq \overline{W}^{\infty}(\Lambda )-\frac{|(1-\delta)Q|}{|Q|}\liminf_{\e\to 0}\tilde{H}_{\e}((1-\delta)Q,u_{\e})\leq \int_{Q\backslash (1-\delta)Q}\overline{W}^{\infty}(\nabla v(x))\,\mathrm{d}x,
\end{align*}
where for the last estimate we also used the local $\Gamma$-convergence result on $(1-\delta)Q$ stated in Theorem~\ref{Gamma-limit}. Letting $\delta\to 0$ we deduce the claim since the right hand side vanishes. The quadratic structure follows by the general theory of $\Gamma$-convergence of quadratic functionals (see \cite[Theorem~11.10]{DM}) while the uniform coercivity of $A$ is conserved in the limit, too. This  proves uniform convexity of $\Lambda\mapsto\overline{W}^{\beta}(\Lambda)$.
\end{proof}


\section{Penalizing volume changes}\label{Sec6}

Having in mind the model presented in the introduction, we now explain how to include the volumetric term defined in \eqref{def:volHam} in our previous analysis. We assume throughout this whole section that $n=d$ and $p\geq d$. For the notation we also refer to Subsection \ref{subsec:definitions}.

In order to estimate the volumetric part of the Hamiltonian, it is convenient to rewrite the integral as sums over $d$-simplices, that is,
\begin{equation}\label{defvol_discrete}
H_{\rm vol,\e}(O,u)=\sum_{\mathcal{C}_1(x)\in\mathcal{V}_{1,\e}(O)}|\mathcal{C}_1(x)|W\Big(\sum_{T\cap\mathcal{C}_1(x)\neq\emptyset}\frac{\det(\nabla u_{\rm aff_{|T}})|T\cap\mathcal{C}_1(x))|}{|\mathcal{C}_1(x)|}\Big).
\end{equation}

The Hamiltonian has a different structure than in the previous sections since it depends on multi-body interactions through the volumetric term. However, we emphasize that for our analysis in the previous sections, the precise structure was needed only for proving stationarity. The reader might remember that elsewhere we just used bounds from above and below and a certain locality given by the finite range of interactions (see also Remark~\ref{r.subadditivity}). As we will prove in the lemmata below, the bounds from Hypothesis~\ref{Hypo1} lead to similar local bounds for the Hamiltonian and stationarity of the corresponding stochastic processes is preserved, too. To extend the validity of the results of Section \ref{Sec5}, it then suffices to reprove a global continuity estimate. These technical details then allow to include the multi-body volumetric term in the proofs of the previous sections and in that sense complete the proofs of Theorems \ref{th:W}, \ref{th:Helmholtz}, \ref{LDP} and \ref{th:smallT}. We leave the details  to the reader.

\begin{lemma}\label{boundsvol}
Assume that $W$ satisfies Hypothesis~\ref{Hypo1} and let $G\in\mathcal{G}$. Then there exists a constant $C>0$ such that, for every $u:\calL\to\R^d$ and all $x\in\calL_1$, it holds that
\begin{equation*}
W\Big(\fint_{\mathcal{C}_1(x)}\det(\nabla u_{\rm aff})\,\mathrm{d}z\Big)\leq C\Big(1+\sum_{T\cap \mathcal{C}_1(x)\neq\emptyset}\sum_{y,y^{\prime}\in\mathcal{L}_1\cap T}|u(y)-u(y^{\prime})|^p\Big).
\end{equation*}	
\end{lemma}
\begin{remark}\label{r.subadditivity}
Points $y,y^{\prime}$ appearing in the above upper bound satisfy $|x-y|,|x-y^{\prime}|\leq 3R$. Hence the condition $6R<C_0$ allows to establish almost subadditivity estimates using boundary values by the definition of the interior Voronoi cells $\mathcal{V}_{1,\e}(O)$.
\end{remark}
\begin{proof}[Proof of Lemma~\ref{boundsvol}]
We rewrite the left hand side term similar to \eqref{defvol_discrete}. Since $\calL_1\subset\calL$, the volume of the Voronoi cells is uniformly bounded from below. Hence from the upper bound in Hypothesis~\ref{Hypo1} and the area formula we deduce
\begin{align}\label{Wbound}
W\Big(\fint_{\mathcal{C}_1(x)}\det(\nabla u_{\rm aff})\,\mathrm{d}z\Big)\leq&  C+C\Big(\sum_{T\cap\mathcal{C}_1(x)\neq\emptyset}|\det(\nabla u_{\rm aff_{|T}}||T\cap\mathcal{C}_1(x))|\Big)^\frac{p}{d}\nonumber
\\
\leq& C+C\Big(\sum_{T\cap\mathcal{C}_1(x)\neq\emptyset}|u_{\rm aff}(T)|\Big)^\frac{p}{d}.
\end{align} 
We claim that for each $T\in\mathbb{T}$ with $|T\cap\mathcal{C}_1(x)|>0$ it holds that
\begin{equation}\label{toshowestimate1}
|u_{\rm aff}(T)|\leq C\Big(\sum_{y,y^{\prime}\in\mathcal{L}_1\cap T}|u(y)-u(y^{\prime})|\Big)^d.
\end{equation}
Indeed, if $\det(\nabla u_{\rm aff_{|T}})=0$, then there is nothing to prove. Otherwise, the set $u_{\rm aff}(T)$ is again a $d$-simplex with vertices $\{u(y)\}_{y\in\calL_1\cap T}$. By convexity its diameter can be bounded by 
\begin{equation*}
{\rm diam}(u_{\rm aff}(T))\leq \sum_{y,y^{\prime}\in\calL_1\cap T}|u(y)-u(y^{\prime})|,
\end{equation*}
so that the bound $|u_{\rm aff}(T)|\leq {\rm diam}(u_{\rm aff}(T))^d$ implies \eqref{toshowestimate1}. Combining \eqref{Wbound} and \eqref{toshowestimate1} we conclude that
\begin{equation*}
W\Big(\fint_{\mathcal{C}_1(x)}\det(\nabla u_{\rm aff})\,\mathrm{d}z\Big)\leq C+C\Big(\sum_{T\cap\mathcal{C}_1(x)\neq\emptyset}\sum_{y,y^{\prime}\in\calL_1\cap T}|u(y)-u(y^{\prime})|\Big)^p
\end{equation*}
and the statement follows by Jensen's inequality since the number of terms in the above sum is equibounded with respect to $x\in\calL_1$ and $G\in\mathcal{G}$.
\end{proof}
We continue our series of lemmas with the proof of stationarity as used when applying the ergodic theorem.
\begin{lemma}\label{vol_stationary}
Let $u:\calL\to\R^d$ and let $G\in\mathcal{G}$. Then, for all $z\in\mathbb{Z}^d$, all $I\in\mathcal{I}$ and every $\a\in\R^n$, it holds that
\begin{equation*}
H_{\rm vol,1}(u(\cdot+z)+\a,I,G-z)=H_{\rm vol,1}(u,I+z,G).
\end{equation*}
\end{lemma}
\begin{proof}[Proof of Lemma~\ref{vol_stationary}]
Since we assume that $\mathcal{L}_1$ is stationary and in general position, the Delaunay tessellation $\mathbb{T}$ of $\R^d$ with respect to $\calL_1$ is unique and hence also stationary. The claim then follows by the stationarity of $\calL_1$ and $\mathbb{T}$ combined with the linearity of piecewise affine interpolations, a discrete change of variables and translation invariance of the Lebesgue measure.
\end{proof}
The last point left in order to repeat the analysis of the previous sections for the volumetric term concerns the quantitative continuity in order to prove the convergence in the zero temperature limit. 

Before we prove the latter, we introduce some further notation. Define the set of neighbours for the volumetric points $\mathcal{L}_1$ by
\begin{equation*}
\mathcal{N}_1:=\{(x,y)\in\mathcal{L}_1\times\mathcal{L}_1:\;{\rm dim}(\mathcal{C}_1(x)\cap\mathcal{C}_1(y))=d-1\}.
\end{equation*}
Since we assume $\calL_1$ to be in general position, two points $x,y\in\calL_1$ belong to the same simplex $T\in\mathbb{T}$ if and only if they are neighbours. Given $u:\calL\to\R^d$ we set
\begin{equation*}
\|\nabla_{\mathcal{N}}u\|_{\ell^p_{\e}(D)}=\Big(\sum_{\substack{(x,y)\in\mathcal{N}_1\\ \e x,\e y\in D}}|u(x)-u(y)|^p\Big)^{\frac{1}{p}}.
\end{equation*}
Then the continuity estimate reads as follows:
\begin{lemma}\label{globalcont}
Assume that $W$ satisfies Hypothesis~\ref{Hypo2} and let $G\in\mathcal{G}$. Then, for any $u,\,\zeta:\calL\to\R^d$ and any bounded Lipschitz domain $D\subset\R^d$, we have the global continuity estimate
\begin{equation*}
|H_{\rm vol,\e}(D,u)-H_{\rm vol,\e}(D,\zeta)|\leq C\Big(|D^{\calL}_{\e}|^{\frac{p-1}{p}}+\|\nabla_{\mathcal{N}}u\|^{p-1}_{\ell^p_{\e}(D)}+\|\nabla_{\mathcal{N}}\zeta\|^{p-1}_{\ell^p_{\e}(D)}\Big)\|\nabla_{\mathcal{N}}(u-\zeta)\|_{\ell^p_{\e}(D)}.
\end{equation*}
\end{lemma}
\begin{proof}[Proof of Lemma~\ref{globalcont}]
Fix $x\in\calL_1$. Applying \eqref{Hyp-2.2} we infer from the area formula and Jensen's inequality (recall that $p\geq d$) that
\begin{align*}
\Bigg|W&\Big(\fint_{\mathcal{C}_1(x)}\det(\nabla u_{\rm aff})\,\mathrm{d}z\Big)-W\Big(\fint_{\mathcal{C}_1(x)}\det(\nabla \zeta_{\rm aff})\,\mathrm{d}z\Big)\Bigg|
\\
\leq& C\Big(1+ \sum_{T\cap\mathcal{C}_1(x)\neq\emptyset}|u_{\rm aff}(T\cap\mathcal{C}_1(x))|^{\frac{p}{d}-1}+|\zeta_{\rm aff}(T\cap\mathcal{C}_1(x))|^{\frac{p}{d}-1}\Big)\times\fint_{\mathcal{C}_1(x)}|\det(\nabla u_{\rm aff})-\det(\nabla \zeta_{\rm aff})|\,\mathrm{d}z
\\
\leq& C\Big(1+ \sum_{T\cap\mathcal{C}_1(x)\neq\emptyset}|u_{\rm aff}(T)|^{\frac{p}{d}-1}+|\zeta_{\rm aff}(T)|^{\frac{p}{d}-1}\Big)
\times \sum_{T\cap\mathcal{C}_1(x)\neq\emptyset}|\det(\nabla u_{\rm aff_{|T}})-\det(\nabla \zeta_{\rm aff_{|T}})||T|.
\end{align*}
Taking into account \eqref{toshowestimate1}, we can again use Jensen's inequality to further estimate 
\begin{align}\label{longestimate}
\Bigg|W&\Big(\fint_{\mathcal{C}_1(x)}\det(\nabla u_{\rm aff})\,\mathrm{d}z\Big)-W\Big(\fint_{\mathcal{C}_1(x)}\det(\nabla \zeta_{\rm aff})\,\mathrm{d}z\Big)\Bigg|\nonumber
\\
\leq& C\Big(1+\sum_{T\cap \mathcal{C}_1(x)\neq\emptyset}\sum_{y,y^{\prime}\in\mathcal{L}_1\cap T}\Big(|u(y)-u(y^{\prime})|^{p-d}+|\zeta(y)-\zeta(y^{\prime})|^{p-d}\Big)\Big)\nonumber
\\
&\times \sum_{T\cap\mathcal{C}_1(x)\neq\emptyset}|\det(\nabla u_{\rm aff_{|T}})-\det(\nabla \zeta_{\rm aff_{|T}})||T|.
\end{align}
We bound the difference in each term of the last sum. Write $T={\rm co}(x_0,\dots,x_d)$. Then by the volume formula for simplices
\begin{equation*}
\det(\nabla u_{\rm aff_{|T}})|T|=\frac{1}{d!}\det(\nabla u_{\rm aff_{|T}})\det\big(x_1-x_0|\dots|x_d-x_0\big)=\frac{1}{d!}\det\big(u(x_1)-u(x_0)|\dots|u(x_d)-u(x_0)\big).
\end{equation*}
The same formula holds with $\zeta$ in place of $u$. From the standard continuity estimate for determinants, we deduce that
\begin{align*}
|\det(\nabla u_{\rm aff_{|T}})-\det(\nabla \zeta_{\rm aff_{|T}})||T|\leq& C\max_{i}\Big(|u(x_i)-u(x_0)|+|\zeta(x_i)-\zeta(x_0)|\Big)^{d-1}
\\
&\times\sum_{y,y^{\prime}\in\mathcal{L}_1\cap T}|u(y)-u(y^{\prime})-\zeta(y)+\zeta(y^{\prime})|
\end{align*}
Recall that $d\leq p$. Hence inserting the above estimate into \eqref{longestimate} it follows that
\begin{align}\label{esth1}
\Big|W_1&\Big(\sum_{T\cap\mathcal{C}_1(x)\neq\emptyset}\frac{|u_{\rm aff}(T\cap\mathcal{C}_1(x))|}{|\mathcal{C}_1(x)|}\Big)-W_1\Big(\sum_{T\cap\mathcal{C}_1(x)\neq\emptyset}\frac{|\zeta_{\rm aff}(T\cap\mathcal{C}_1(x))|}{|\mathcal{C}_1(x)|}\Big)\Big|\nonumber
\\
\leq&C\Big(1+\sum_{T\cap\mathcal{C}_1(x)\neq\emptyset}\sum_{y,y^{\prime}\in\mathcal{L}_1\cap T}\Big(|u(y)-u(y^{\prime})|^{p-1}+|\zeta(y)-\zeta(y^{\prime})|^{p-1}\Big)\nonumber
\\
&\times\sum_{T\cap\mathcal{C}_1(x)\neq\emptyset}\sum_{y,y^{\prime}\in\mathcal{L}_1\cap T}|u(y)-u(y^{\prime})-\zeta(y)+\zeta(y^{\prime})|\Big).
\end{align}
Note that each $T\in\mathbb{T}$ can intersect only a uniformly bounded number of Voronoi cells $\mathcal{C}_1$. Hence, summing \eqref{esth1} over all $\mathcal{C}_1(x)\in\mathcal{V}_{1,\e}(D)$ and using H\"older's inequality for the products yields
\begin{align*}
|H_{\rm vol,\e}&(D,u)-H_{\rm vol,\e}(D,\zeta)|
\\ 
\leq&C\Big(|D^{\calL}_{\e}|^{\frac{p-1}{p}}+\Big(\sum_{T\subset D_{\e}}\sum_{y,y^{\prime}\in\mathcal{L}_1\cap T}\big(|u(y)-u(y^{\prime})|^{p}\Big)^{\frac{p-1}{p}}+\Big(\sum_{T\subset \frac{D}{\e}}\sum_{y,y^{\prime}\in\mathcal{L}_1\cap T}\big(|\zeta(y)-\zeta(y^{\prime})|^{p}\Big)^{\frac{p-1}{p}}\Big)
\\
&\times\Big(\sum_{T\subset D_{\e}}\sum_{y,y^{\prime}\in\mathcal{L}_1\cap T}|u(y)-u(y^{\prime})-\zeta(y)+\zeta(y^{\prime})|^p\Big)^{\frac{1}{p}}
\\
\leq&C\Big(|D^{\calL}_{\e}|^{\frac{p-1}{p}}+\|\nabla_{\mathcal{N}}u\|^{p-1}_{\ell^p_{\e}(D)}+\|\nabla_{\mathcal{N}}\zeta\|^{p-1}_{\ell^p_{\e}(D)}\Big)\|\nabla_{\mathcal{N}}(u-\zeta)\|_{\ell^p_{\e}(D)},
\end{align*}
where we used in the last inequality that each element in $\mathbb{T}$ has as vertices only nearest neighbours and that a vertex can belong to only a uniformly bounded number of different cells. The last estimate yields the claim.
\end{proof}

\appendix
\section{}\label{appA}
In this appendix we collect and prove some of the results we used in the paper. We start with the technical proof of the interpolation inequality.

\begin{proof}[Proof of Proposition~\ref{interpolation}]
We set $\beta=1$ to reduce the notation. Given $\delta>0$ and $N\in\mathbb{N}$, for $i\in\{1,\dots,N+1\}$ we introduce the open sets
\begin{equation*}
O_i=\left\{x\in O:\;\dist(x,\partial O)>(i+1)\frac{\delta}{2N}\right\}.
\end{equation*}
Then the stripes $S_i:=O_{i-1}\backslash \overline{O_{i+2}}$ fulfill $S_i\cap S_j=\emptyset$ whenever $|i-j|>2$. Thus for every $u:O_{\e}^\calL\to\R^n$ we obtain by averaging
\begin{equation*}
\frac{1}{N}\sum_{i=1}^{N}H_{\e}(S_i,u)\leq \frac{3}{N}H_{\e}(O,u),
\end{equation*}
so that we can decompose the set $\mathcal{N}_p(v,O,\e,\kappa)=\bigcup_{i=1}^N\mathcal{P}_{i,\e}$ (we omit the dependence on $O$ and $\kappa$), where
\begin{equation*}
\mathcal{P}_{i,\e}=\left\{u\in\mathcal{N}_p(v,O,\e,\kappa):\;H_{\e}(S_i,u)\leq \frac{3}{N}H_{\e}(O,u)\right\}.
\end{equation*}
Let $\theta_i:O\to[0,1]$ be the Lipschitz-continuous cut-off function defined by
\begin{equation*}
\theta_i(z)=\min\left\{\max\left\{\frac{2N}{\delta}\dist(z,\partial O)-(i+1),0\right\},1\right\},
\end{equation*}
so that $\theta_i\equiv 1$ on $\overline{O_{i+1}}$, $\theta_i\equiv 0$ on $O\backslash O_{i}$ and its Lipschitz constant can be bounded by $\Lip(\theta_i)\leq \frac{2N}{\delta}$. We then define an interpolation between functions $u,\psi:O_{\e}^\calL\to\R^n$ as
\begin{equation*}
T_{i,\e}(u,\psi)(x)=\theta_i(\e x)u(x)+(1-\theta_i(\e x))\psi(x).
\end{equation*}
Observe that if $u\in\mathcal{P}_{i,\e}$ as well as $\varphi\in\mathcal{N}_p(v,O,\e,\kappa)$ and  $\psi\in\mathcal{N}_{\infty}(\varphi,O\backslash \overline{O_{i+1}},\e)$, by the Minkowski inequality we have
\begin{equation*}
\e^{\frac{d}{p}}\|v_{\e}-\e T_{i,\e}(u,\psi)\|_{\ell^p_{\e}(O)}\leq 2\kappa|O|^{\frac{1}{p}+\frac{1}{d}}+\e^{\frac{d}{p}}\|\e\psi-\e\varphi\|_{\ell^p_{\e}(O\backslash\overline{O_{i+1}})}
\leq 2\kappa |O|^{\frac{1}{p}+\frac{1}{d}}+C\e |O|^{\frac{1}{p}},
\end{equation*}
so that $T_{i,\e}(u,\psi)\in\mathcal{N}_p(v,O,\e,3\kappa)$ for $\e$ small enough. 
	
For technical reasons the interpolations will not suffice to prove the estimates. For every $i$ let us choose $t_i\in[\frac{1}{4},\frac{3}{4}]$ such that, setting $S_i^t=\{x\in O:\;\theta_i(x)=t\}$, the coarea formula implies
\begin{equation}\label{prepareaveraging}
\frac{1}{2}\mathcal{H}^{d-1}( S_i^{t_i})\leq\int_{\frac{1}{4}}^{\frac{3}{4}}\mathcal{H}^{d-1}(S_i^t)\,\mathrm{d}t\leq \int_0^1\mathcal{H}^{d-1}( S_i^t)\,\mathrm{d}t=\int_{O}|\nabla\theta_i|\leq \frac{2N}{\delta}|O_i\backslash{O_{i+1}}|.
\end{equation}
We set $S_i^*=\{x\in O:\theta_i(x)<t_i\}$. Note that for $\delta$ small enough (depending only on $O$), we have $S_i^*\in\Ard$ (see for instance \cite[Lemma~2.2]{glpe}). Let us introduce the product set
\begin{equation*}
\mathcal{U}_{\e}^i(M):=\big(\mathcal{P}_{i,\e}\cap\mathcal{S}_M(O,\e)\big)\times\mathcal{N}_{\infty}(\varphi,O\backslash\overline{O_{i+1}},\e),
\end{equation*}
as well as the integral
\begin{equation*}
e^i_{\e}(M):=\bigg(\int_{\mathcal{U}_{\e}^i(M)}\exp\Big(- H_{\e}(O,T_{i,\e}(u,\psi))-c_0\|\nabla_{\B}u\|^p_{\ell^p_{\e}(S_i^*)}\Big)\,\mathrm{d}u\,\mathrm{d}\psi \bigg),
\end{equation*}
where $c_0>0$ is a small constant such that $c_0|\xi|^p\leq f(\cdot,\xi)+c_0^{-1}$ (cf. Hypothesis~\ref{Hypo1}). This integral quantity will be the main ingredient to prove the interpolation inequality. We split the remaining argument into several steps.

\step1{Energy bounds for the interpolation}
To bound the energy of $T_{i,\e}(u,\psi)$, we use the pointwise inequality 
\begin{equation*}
|\psi(x)-\psi(y)|^p\leq C|(\psi-\varphi)(x)-(\psi-\varphi)(y)|^p+C|\varphi(x)-\varphi(y)|^p\leq C+C|\varphi(x)-\varphi(y)|^p,
\end{equation*}
which is valid for all $x,y\in (O\backslash \overline{O_{i+1}})_{\e}$. Combined with the two-sided growth condition in Hypothesis~\ref{Hypo1} we infer that
\begin{align}\label{firstbound}
H_{\e}(O,T_{i,\e}(u,\psi))&\leq H_{\e}(O_{i+1},u)+H_{\e}(O\backslash\overline{O_i},\psi)+H_{\e}(S_i,T_{i,\e}(u,\psi))\nonumber\\
&\leq H_{\e}(O_{i+1},u)+CH_{\e}(O^{\delta},\varphi)+C|(O^{\delta})^{\calL}_{\e}|+H_{\e}(S_i,T_{i,\e}(u,\psi)),
\end{align}
where $O^{\delta}$ is defined in the statement of Proposition~\ref{interpolation}. In order to estimate the last term on the right hand side we use the formula
\begin{align*}
T_{i,\e}(u,\psi)(x)-T_{i,\e}(u,\psi)(y)=&\big(\theta_i(\e x)-\theta_i(\e y)\big)\big(u(x)-\psi(x)\big)
\\
&+\theta_i(\e y)\big(u(x)-u(y)\big)+(1-\theta_i(\e y))\big(\psi(x)-\psi(y)\big)
\end{align*}
and the bound on the Lipschitz constant of $\theta_i$ to estimate the energy on the interpolation stripe via
\begin{align}\label{improvedbound}
H_{\e}(S_i,T_{i,\e}(u,\psi))&\leq C\|\nabla_{\B}T_{i,\e}(u,\psi)\|^p_{\ell^p_{\e}(S_i)}+C|(S_i)^{\calL}_{\e}|\nonumber\\
&\leq C\left(\|\nabla_{\B}u\|^p_{\ell^p_{\e}(S_i)}+\|\nabla_{\B}\psi\|^p_{\ell^p_{\e}(S_i)}+\frac{(N\e)^p}{\delta^p}\|u-\psi\|^p_{\ell^p_{\e}(S_i)}+|(S_i)^{\calL}_{\e}|\right)\nonumber
\\
&\leq \frac{C}{N}H_{\e}(O,u)+CH_{\e}(O^{\delta},\varphi)+C|(O^{\delta})^{\calL}_{\e}|+\frac{CN^p}{\delta^p}\kappa^p|O|^{1+\frac{p}{d}}\e^{-d},
\end{align}
where we have used again that the degree of each vertex is equibounded and that, after suitable extension, $\psi\in \mathcal{N}_p(v,O,\e,2\kappa)$ for $\e$ small enough. Combining (\ref{firstbound}) and (\ref{improvedbound}) we infer that
\begin{equation}\label{energyboundinterpolation}
H_{\e}(O,T_{i,\e}(u,\psi))\leq H_{\e}(O_{i+1},u)+\frac{C}{N}H_{\e}(O,u)+CH_{\e}(O^{\delta},\varphi)+C|(O^{\delta})^{\calL}_{\e}|+\frac{CN^p}{\delta^p}\kappa^p|O|^{\frac{p}{d}}|O_{\e}^\calL|.
\end{equation} 
	
\step2{Lower bound for $e^i_{\e}(M)$}
In order to prove a lower bound for the integral, first note that due to Hypothesis~\ref{Hypo1} and the definition of $S_i^*$ 
\begin{equation*}
c_0\|\nabla_{\B}u\|^p_{\ell^p_{\e}(S_i^*)}\leq H_{\e}(O\setminus O_{i+1},u)+C|(O^{\delta})^{\calL}_{\e}|, 
\end{equation*}
so that, up to increasing $C$, we can add this inequality to \eqref{energyboundinterpolation} and obtain the estimate
\begin{equation*}
H_{\e}(O,T_{i,\e}(u,\psi))+c_0\|\nabla_{\B}u\|^p_{\ell^p_{\e}(S_i^*)}\leq \left(1+\frac{C}{N}\right)H_{\e}(O,u)+CH_{\e}(O^{\delta},\varphi)+C|(O^{\delta})^{\calL}_{\e}|+\frac{CN^p}{\delta^p}\kappa^p|O|^{\frac{p}{d}}|O_{\e}^\calL|	
\end{equation*}
Rearranging the terms we obtain by Fubini's Theorem~that
\begin{align}\label{firstinterpolationbound}
e_{\e}^i(M)\geq&\exp\bigg(-C \Big(|(O^{\delta})^{\calL}_{\e}|+\frac{(N\kappa|O|^{\frac{1}{d}})^p}{\delta^p}|O_{\e}^\calL|+\frac{M}{N}|O_{\e}^\calL|+ H_{\e}(O^{\delta},\varphi)\Big)\bigg)\int_{\mathcal{N}_{\infty}(\varphi,O\backslash\overline{O_{i+1}},\e)}\mathrm{d}\psi\nonumber
\\
&\times\int_{\mathcal{P}_{i,\e}\cap\mathcal{S}_M(O,\e)}\exp(- H_{\e}(O,u))\,\mathrm{d}u\nonumber
\\
\geq&\exp\bigg(-C \Big(|(O^{\delta})^{\calL}_{\e}|+\frac{(N\kappa|O|^{\frac{1}{d}})^p}{\delta^p}|O_{\e}^\calL|+\frac{M}{N}|O_{\e}^\calL|+ H_{\e}(O^{\delta},\varphi)\Big)\bigg)\nonumber\\
&\times Z_{\e,O}(\mathcal{P}_{i,\e}\cap\mathcal{S}_M(O,\e))
\end{align}
where we used that the measure of $\mathcal{N}_{\infty}(\varphi,O\backslash\overline{O_{i+1}},\e)$ can be bounded from below by $\exp(-C|(O^{\delta})^{\calL}_{\e}|)$.
	
\step3{Upper bound for $e^i_{\e}(M)$ and conclusion}
To estimate $e^i_{\e}(M)$ from above, similar to \cite{KoLu} we perform a suitable change of variables. Define 
$\Phi_{i,\e}:\mathcal{N}_p(v,O,\e,\kappa)\times \mathcal{N}_{\infty}(\varphi,O\backslash\overline{O_{i+1}},\e)\to \mathcal{N}_p(u,O,\e,3\kappa)\times\mathcal{N}_p(u,O\backslash\overline{O_{i+1}},\e,3\kappa)$ by
\begin{equation*}
\Phi_{i,\e}(u,\psi)(x)=\begin{cases}
(T_{i,\e}(u,\psi)(x),\psi(x)) &\mbox{if $\theta_i(\e x)\geq t_i$,}\\
(T_{i,\e}(u,\psi)(x),u(x)) &\mbox{if $\theta_i(\e x)<t_i$.}
\end{cases}
\end{equation*}
Note that for $\e$ small enough $\Phi_{i,\e}$ is well-defined and bijective onto its range $\mathcal{R}(\Phi_{i,\e})$. For the idea how to calculate the Jacobian, we refer to the proof of Proposition~\ref{fullexistence}. As $t_i\in[\frac{1}{4},\frac{3}{4}]$, it holds that
\begin{equation*}
|\det(D\Phi_{i,\e}(u,\psi))|^{-1}= \left(\prod_{x:\theta_i(\e x)\geq t_i}|\theta_i(\e x)|^n\prod_{x:\theta_i(\e x)<t_i}|1-\theta_i(\e x)|^n\right)^{-1}\leq \exp(C|(O^{\delta})^{\calL}_{\e}|).
\end{equation*}
Setting $(g,h)=\Phi_{i,\e}(u,\psi)$, by construction of the interpolation we have
\begin{equation*}
\begin{split}
g&\in\mathcal{N}_p(u,O,\e,3\kappa)\cap \mathcal{B}_{\e}(O,\varphi),
\\
h&=(h_1,h_2)\in \mathcal{N}_{\infty}(\varphi,O\backslash (\overline{O_{i+1}}\cup S_i^*),\e)\times \underbrace{\{h:(S_i^*)_{\e}\to\R^n:\;\|h-\e^{-1}v_{\e}\|_{\infty}\leq C\kappa|O_{\e}|^{\frac{1}{p}+\frac{1}{d}}\}}_{=:R_{i,\e}}.
\end{split}
\end{equation*}
As the measure of the set $\mathcal{N}_{\infty}(\varphi,O\backslash (\overline{O_{i+1}}\cup S_i^*),\e)$ can be bounded by $\exp(C|(O^{\delta})^{\calL}_{\e}|)$, the above change of variables and Fubini's Theorem~imply
\begin{align}\label{secondinterpolationbound}
e_{\e}^i(M)\leq&\exp(C|(O^{\delta})^{\calL}_{\e}|)\int_{\mathcal{R}(\Phi_{i,\e})}\exp\Big(- H_{\e}(g,O)-c_0\|\nabla_{\B}h\|^p_{\ell^p_{\e}(S_i^*)}\Big)\,\mathrm{d}g\,\mathrm{d}h\nonumber
\\
\leq&\exp(C|(O^{\delta})^{\calL}_{\e}|)\int_{\mathcal{N}_{\infty}(\varphi,O\backslash (\overline{O_{i+1}}\cup S_i^*),\e)}\mathrm{d}h_1\int_{R_{i,\e}}\exp(-c_0\|\nabla_\B  h_2\|^p_{\ell^p_{\e}(S_i^*)})\,\mathrm{d}h_2\nonumber
\\
&\times Z(\mathcal{N}_p(u,O,\e,3\kappa)\cap\mathcal{B}_{\e}(O,\varphi))\nonumber
\\
\leq&\exp(C|(O^{\delta})^{\calL}_{\e}|)\int_{R_{i,\e}}\exp(-c_0\|\nabla_\B  h_2\|^p_{\ell^p_{\e}(S_i^t)})\,\mathrm{d}h_2\;\times Z_{\e,O}(\mathcal{N}_p(u,O,\e,3\kappa)\cap\mathcal{B}_{\e}(O,\varphi)).
\end{align}
In order to bound the integral on the right hand side, we apply Lemma~\ref{forrestargument} to the graph $G_{S_i^*,\e}$ with $\a=c_0$ and $\gamma=C\kappa |O_{\e}|^{\frac{1}{p}+\frac{1}{d}}$ and infer
\begin{align*}
\int_{R_{i,\e}}\exp(-c_0\|\nabla_\B  h_2\|^p_{\ell^p_{\e}(S_i^*)})\,\mathrm{d}h_2&\leq \Big(C\kappa^n |O_{\e}|^{\frac{n}{p}+\frac{n}{d}}\Big)^{N_{i,\e}}C^{|(S_i^*)^{\calL}_{\e}|-N_{i,\e}},
\end{align*}
where we denoted by $N_{i,\e}$ the number of connected components of the graph $G_{S_i^*,\e}$.
For $\e$ small enough (possibly depending on $N,\delta$), by Remark~\ref{numbercomponents}, \eqref{surfaceesimtate} and the fact that $S_i^*\in\Ard$ we can bound the number of components via
\begin{equation*}
N_{i,\e}\leq \#\{x\in O_{\e}^\calL:\;\dist(x,\partial (S_i^*)_{\e})\leq C_0\}\leq C\e^{1-d}(\mathcal{H}^{d-1}(S_i^{t_i})+\mathcal{H}^{d-1}(\partial O)).
\end{equation*}
In particular, for $N,\delta$ and $\kappa>0$ fixed, due to \eqref{prepareaveraging} there exists $\e_0$ such that for all $\e<\e_0$
\begin{equation*}
\int_{R_{i,\e}}\exp(-c_0\|\nabla_\B  h_2\|^p_{\ell^p_{\e}(S_i^*)})\,\mathrm{d}h_2\leq \exp(C|(O^{\delta})^{\calL}_{\e}|).
\end{equation*}
Plugging this bound into \eqref{secondinterpolationbound} and comparing with \eqref{firstinterpolationbound} yields
\begin{align*}
Z_{\e,O}(\mathcal{P}_{\e}^i\cap\mathcal{S}_M(O,\e))\leq&Z_{\e,O}(\mathcal{N}_p(v,O,\e,3\kappa)\cap\mathcal{B}_{\e}(O,\varphi)) \\
&\times \exp\Big(C\big(|(O^{\delta})^{\calL}_{\e}|+\frac{(N\kappa|O|^{\frac{1}{d}})^p}{\delta^p}|O_{\e}^\calL|+\frac{M}{N}|O_{\e}^\calL|+ H_{\e}(O^{\delta},\varphi)\big)\Big)
\end{align*}
Summing this inequality over $i$, by the definition of the sets $\mathcal{P}_{i,\e}$ we infer that
\begin{align}\label{truncatedestimate}
Z_{\e,O}(\mathcal{N}_p(v,O,\e,\kappa)\cap\mathcal{S}_M(O,\e))\leq&\sum_{i=1}^N Z_{\e,O}(\mathcal{P}_{\e}^i\cap\mathcal{S}_M(O,\e))\nonumber
\\
\leq& Z_{\e,O}(\mathcal{N}_p(v,O,\e,3\kappa)\cap\mathcal{B}_{\e}(O,\varphi))\nonumber
\\
&\times N\exp\Big(C\big(|(O^{\delta})^{\calL}_{\e}|+\frac{(N\kappa|O|^{\frac{1}{d}})^p}{\delta^p}|O_{\e}^\calL|+\frac{M}{N}|O_{\e}^\calL|+ H_{\e}(O^{\delta},\varphi)\big)\Big)
\end{align}
Now choosing 
\begin{equation*}
M=2\left(\frac{1}{|O_{\e}^\calL|}\log\Big(Z_{\e,O}(\mathcal{N}_p(v,O,\e,\kappa))\Big)+\overline{C}+\frac{\log(2)}{|O_{\e}^\calL|}\right),
\end{equation*}
where $\overline{C}$ is the constant of Lemma~\ref{tightness}, we obtain by the same Lemma~and Remark~\ref{Mdepends} that, for any $\kappa>0$ fixed and all $\e$ small enough,
\begin{equation*}
Z_{\e,O}(\mathcal{N}_p(v,O,\e,\kappa)\backslash\mathcal{S}_M(O,\e))\leq \frac{1}{2}Z_{\e,O}(\mathcal{N}_p(v,O,\e,\kappa)).
\end{equation*}
Thus \eqref{truncatedestimate} and the definition of $M$ yield the final estimate
	\begin{align*}
	Z_{\e,O}(\mathcal{N}_p(v,O,\e,\kappa))\leq& 2Z_{\e,O}(\mathcal{N}_p(v,O,\e,\kappa)\cap\mathcal{S}_M(O,\e))
	\\
	\leq&Z_{\e,O}(\mathcal{N}_p(v,O,\e,3\kappa)\cap\mathcal{B}_{\e}(O,\varphi))\,Z_{\e,O}(\mathcal{N}_p(v,O,\e,\kappa))^{\frac{C}{N}}\nonumber
	\\
	&\times 2N\exp\Big(C\big(|(O^{\delta})^{\calL}_{\e}|+(\frac{(N\kappa|O|^{\frac{1}{d}})^p}{\delta^p}+\frac{C}{N})|O_{\e}^\calL|+ H_{\e}(O^{\delta},\varphi)\big)\Big).
	\end{align*}
\end{proof}
\begin{remark}\label{invariance}
	Note that the restriction on $\delta$ in the interpolation inequality comes only from the requirement that tubular neighbourhoods of the boundary have again Lipschitz boundary. In particular, if $\delta$ satisfies the condition for a set $O\subset\R^d$, then $\delta^{\prime}=\delta \rho$ satisfies the condition for all sets of the form $O^{\prime}=z+\rho O$. Applying	this fact to the family of cubes $Q(z,\rho)$ with $z\in D$ and $\rho>0$, we obtain that there exists $\delta_0>0$ such that for all $\delta<\delta_0$, all $N\in\mathbb{N}$ and all $\kappa>0$ it holds that
	\begin{align*}
	\frac{N-C}{N}\mathcal{F}_{\kappa}^-(Q(z,\rho), \overline{\varphi}_\Lambda)\geq& \overline W(\Lambda )-C(1+|\Lambda |^p)\frac{|Q(z,\rho)^{\delta\rho}|}{|Q(z,\rho)|}-C\left(\frac{(N\kappa|O(z,\rho)|^{\frac{1}{d}})^p}{(\delta\rho)^p}+\frac{1}{N}\right)
	\\
	\geq&\overline W(\Lambda )-C\left((1+|\Lambda |^p)\delta+\frac{(N\kappa)^p}{\delta^p}+\frac{1}{N}\right)
	\end{align*}
	$\mathbb{P}$-almost surely. Here we used Lemma~\ref{equivalent}, Proposition~\ref{fullexistence} and \eqref{e.volVoronoi} in order to pass to the limit as $\e\to 0$ in the interpolation inequality almost surely. Note that the last bound is independent of $\rho$ and $z$.
\end{remark}

\begin{lemma}\label{pgrowth}
Let $p\in(1,+\infty)$. For all $u,v\in\R^n$ it holds that
\begin{equation*}
|u-v|^p+|u+v|^p\leq \max\{2^{p-1},2\}(|u|^p+|v|^p).
\end{equation*}
\end{lemma}
\begin{proof}[Proof of Lemma~\ref{pgrowth}]
For $p\geq 2$ the claimed estimate follows from Clarkson's inequality. If $p<2$, then $(x_1^p+x_2^p)^{\frac{1}{p}}\geq (x_1^2+x_2^2)^{\frac{1}{2}}$ for all $x_1,x_2\geq 0$. Moreover, with elementary analysis one can show that $(x_1^p+x_2^p)^{\frac{1}{p}}\leq 2^{\frac{1}{p}-\frac{1}{2}}(x_1^2+x_2^2)^{\frac{1}{2}}$. Applying these two inequalities first with $x_1=|u-v|$ and $x_2=|u+v|$ and then with $x_1=|u|$ and $x_2=|v|$ we obtain
\begin{equation*}
(|u-v|^p+|u+v|^p)^{\frac{1}{p}}\leq 2^{\frac{1}{p}-\frac{1}{2}}(|u-v|^2+|u+v|^2)^{\frac{1}{2}}=2^{\frac{1}{p}}(|u|^2+|v|^2)^{\frac{1}{2}}\leq 2^{\frac{1}{p}}(|u|^p+|v|^p)^{\frac{1}{p}}.
\end{equation*}
\end{proof}

\begin{lemma}\label{surfacevolume}
Let $p\in(1,+\infty)$. Then there exists a constant $c_{p}$ such that the  Hausdorff measure of the sphere $S^{n-1}_{p}=\{y\in\R^n:\;|y|_p=1\}$ fulfills
\begin{equation*}
\mathcal{H}^{n-1}(S^{n-1}_{p})\geq \left(\frac{c_{p}}{n}\right)^{\frac{n}{p}}.
\end{equation*}	
\end{lemma}
\begin{proof}[Proof of Lemma~\ref{surfacevolume}]
Note that $S^{n-1}_{p}$ is a compact smooth $(n-1)$-dimensional manifold. Hence we can characterize its Hausdorff measure by its Minkowski content. To be more precise, it holds that
\begin{equation}\label{minkowski}
\mathcal{H}^{n-1}(S^{n-1}_{p})=\lim_{\e\to 0}\frac{\mathcal{H}^{n}(S^{n-1}_{p}+B_{\e}(0))}{2\e},
\end{equation}
where the factor $2$ comes from the Lebesgue measure of the 1D unit ball $[-1,1]$. Note however that $B_{\e}(0)$ is a ball with respect to the Euclidean metric on $\R^n$. We now give a lower bound for the nominator on the right hand side of (\ref{minkowski}). To this end, set $c_{n,p}=\max\{1,n^{\frac{1}{2}-\frac{1}{p}}\}$. Then, for $y\neq 0$, we have 
\begin{equation*}
\left|y-\frac{y}{|y|_p}\right|_2\leq ||y|_p-1|\frac{|y|_2}{|y|_p}\leq ||y|_p-1|c_{n,p},
\end{equation*}
where we used that by definition $|y|_2\leq c_{n,p}|y|_p$ for all $y\in\R^n$. We conclude that
\begin{equation*}
\{y\in\R^n:\;1-c^{-1}_{n,p}\e<|y|_p<1+c^{-1}_{n,p}\e\}\subset S^{n-1}_p+B_{\e}(0).
\end{equation*}
Hence we deduce from \eqref{minkowski} and the well-know formula for the volume of $p$-norm balls that
\begin{align*}
\mathcal{H}^{n-1}(S^{n-1}_p)&\geq \liminf_{\e\to 0}\frac{\mathcal{H}^n(\{|y|_p<1+c_{n,p}^{-1}\e\})-\mathcal{H}^n(\{|y|_p<1-c^{-1}_{n,p}\e\})}{2\e}\\
&=\frac{(2\Gamma(\frac{1}{p}+1))^n}{\Gamma(\frac{n}{p}+1)}\lim_{\e\to 0}\frac{(1+c^{-1}_{n,p}\e)^n-(1-c^{-1}_{n,p}\e)^n}{2\e}=\frac{(2\Gamma(\frac{1}{p}+1))^n}{\Gamma(\frac{n}{p}+1)}nc^{-1}_{n,p}\geq \frac{(2\Gamma(\frac{1}{p}+1))^n}{\Gamma(\frac{n}{p}+1)}n^{\frac{1}{2}}.
\end{align*}
We conclude the proof using Stirling's formula in the form of the upper bound
\begin{equation*}
\Gamma\left(\frac{n}{p}+1\right)\leq \left(\frac{2\pi n}{p}\right)^{\frac{1}{2}}\left(\frac{n}{pe}\right)^{\frac{n}{p}}\exp(p/12).
\end{equation*}
\end{proof}

\section*{Acknowledgements}
AG warmly thanks Fran\c{c}ois Lequeux and Michael Rubinstein for inspiring discussions on polymer physics.
AG and MR acknowledge 
the financial support from the European Research Council under
the European Community's Seventh Framework Programme (FP7/2014-2019 Grant Agreement
QUANTHOM 335410). The work of MC was supported by the DFG Collaborative Research Center TRR 109, 
``Discretization in Geometry and Dynamics''.

\end{document}